\documentclass[11pt]{article}

\makeatletter
\def\maxwidth{ %
  \ifdim\Gin@nat@width>\linewidth
    \linewidth
  \else
    \Gin@nat@width
  \fi
}
\makeatother

\RequirePackage{amsmath,amssymb}
\RequirePackage{amsthm}
\RequirePackage{natbib}
\RequirePackage[colorlinks,allcolors=blue]{hyperref}

\usepackage{fullpage}
\usepackage{setspace}
\usepackage{nicefrac}

\usepackage{bbm,graphicx,bm} 
\usepackage{enumitem}
\usepackage{multirow}
\usepackage[all,import]{xy}
\usepackage{appendix}
\usepackage{comment}
\usepackage{color}
\usepackage{soul}
\usepackage{pdflscape}
\usepackage{subcaption}

\usepackage{booktabs}
\usepackage{longtable}
\usepackage{array}
\usepackage[table]{xcolor}
\usepackage{wrapfig}
\usepackage{float}
\usepackage{colortbl}
\usepackage{pdflscape}
\usepackage{tabu}
\usepackage{threeparttable}
\usepackage{threeparttablex}
\usepackage[normalem]{ulem}
\usepackage{makecell}
\usepackage{afterpage}

\usepackage{flafter} 

\theoremstyle{definition}
\newtheorem{assumption}{Assumption}

\newtheorem{theorem}{Theorem}
\newtheorem{lemma}{Lemma}

\newtheorem{proposition}{Proposition}

\newcommand\indep{\protect\mathpalette{\protect\independenT}{\perp}}
\def\independenT#1#2{\mathrel{\rlap{$#1#2$}\mkern2mu{#1#2}}}
\newcommand{\R}{\ensuremath{\mathbb{R}}}

\newcommand{\bbone}{\ensuremath{\mathbbm{1}}}

\newcommand{\E}{\ensuremath{\mathbb{E}}}
\newcommand{\calE}{\ensuremath{\mathcal{E}}}

\newcommand{\calD}{\ensuremath{\mathcal{D}}}
\newcommand{\calL}{\ensuremath{\mathcal{L}}}

\newcommand{\Var}{\text{Var}}

\newcommand{\sep}{\text{sep}}
\newcommand{\pool}{\text{pool}}

\newcommand{\scm}{\text{scm}}

\newcommand{\pre}{\text{pre}}

\def\super{\textsuperscript}

\def\b1{\boldsymbol{1}}

\definecolor{RED}{RGB}{255,0,0}

\usepackage{soul}
\usepackage[utf8]{inputenc}

\title{Synthetic Controls with Staggered Adoption\thanks{email: \texttt{afeller@berkeley.edu}. We would like to thank Alberto Abadie, Howard Bloom, Peng Ding, Arin Dube, Guido Imbens, Skip Hirshberg, Brian Jacob, Luke Keele, Luke Miratrix, Joe Ornstein, Agustina Paglayan, Sam Pimentel, Jake Soloff, Panos Toulis, Chelsea Zhang, and Ben Zipperer for useful discussion and comments, as well as participants at the 2019 Atlantic Causal Inference Conference. We also thank the associate editor and reviewers for constructive feedback. This research was supported in part by the Opportunity Lab and the Institute for Research on Labor and Employment at UC Berkeley, as well as the Institute of Education Sciences, U.S. Department of Education, through Grant R305D200010. The opinions expressed are those of the authors and do not represent views of the Institute or the U.S. Department of Education.}}

\author{Eli Ben-Michael, Avi Feller, and Jesse Rothstein\\[1em] UC Berkeley}
\date{\today}

\begin{document}

\maketitle

\thispagestyle{empty}
\pagenumbering{gobble}

\begin{abstract}
\singlespacing
Staggered adoption of policies by different units at different times creates promising opportunities for observational causal inference.
Estimation remains challenging, however, and common regression methods can give misleading results. A promising alternative is the synthetic control method (SCM), which finds a weighted average of control units that closely balances the treated unit's pre-treatment outcomes.
In this paper, we generalize SCM, originally designed to study a single treated unit, to the staggered adoption setting.
We first bound the error for the average effect and show that it depends on both the imbalance for each treated unit separately and the imbalance for the average of the treated units. 
We then propose ``partially pooled'' SCM weights to minimize a weighted combination of these measures; approaches that focus only on balancing one of the two components can lead to bias.
We extend this approach to incorporate unit-level intercept shifts and auxiliary covariates.
We assess the performance of the proposed method via extensive simulations and apply our results to the question of whether teacher collective bargaining leads to higher school spending, finding minimal impacts. We implement the proposed method in the \texttt{augsynth} R package. 
\end{abstract}

\clearpage
\pagenumbering{arabic}
\onehalfspacing

\section{Introduction}

Jurisdictions often adopt policies at different times, creating promising opportunities for observational causal inference.
In our motivating application, 33 states passed laws between 1964 and 1987 mandating that school districts bargain with teachers unions \citep{hoxby1996teachers, paglayan2019public}; our goal is to estimate the impact of these laws on teacher salaries and school expenditures.

Estimating causal effects under staggered adoption remains challenging, however. Workhorse methods, such as the regression-based two-way fixed effects model,
rely on strong modeling assumptions and can give misleading estimates when treatment timing varies \citep{abraham2018estimating, borusyak2017revisiting, goodman2018difference}.
A promising alternative is the \emph{synthetic control method} \citep[SCM;][]{AbadieAlbertoDiamond2010, Abadie2015}.  
SCM estimates the counterfactual untreated outcome via a weighted average of untreated units, with weights chosen to match the treated unit's pre-treatment outcomes as closely as possible. 
SCM, however, was developed for settings where only a single unit is treated, and proposals for extending SCM to the staggered adoption case have been ad hoc.
One common strategy is to estimate SCM weights separately for each treated unit and then average the estimates \citep[see, e.g.,][]{dube2015pooling, donohue2019right}. However, this relies on being able to find good synthetic controls for every treated unit, which is not possible in our application.

In this paper, we develop SCM for the staggered adoption setting.
Under two common data generating processes for panel data, an autoregressive model and a linear factor model, we bound the error of a weighting estimator for the average effect and show that it depends on both the unit-specific imbalance for each treated unit and the imbalance for the average of the treated units. This leads to our main proposal, \emph{partially pooled SCM}, which 
minimizes a weighted average of the two imbalances.
This approach nests two special cases: \emph{separate SCM}, which reflects the current practice of estimating weights that separately minimize the pre-treatment imbalance for each treated unit; and \emph{pooled SCM}, which instead minimizes the average pre-treatment imbalance across all treated units.
Both special cases have drawbacks. Separate SCM can lead to poor fit for the average, leading to possible bias when the average treatment effect is the estimand of interest. Pooled SCM, by contrast, can achieve nearly perfect fit for the average treated unit but can yield substantially worse unit-specific fits.
This can lead to poor estimates of unit-level treatment effects and to bias for the average effect if the data generating process varies over time.
Partially pooled SCM moves smoothly between these two extremes, with a hyperparameter denoting the relative weight of the two balance measures in the optimization problem.
We discuss how to select weights to trade off between these two quantities in practice.

We then explore several extensions. 
First, we incorporate an intercept shift into the SCM problem, following proposals by \citet{Doudchenko2017} and \citet{ferman2018revisiting}. The resulting treatment effect estimator has the form of a weighted difference-in-differences estimator, connecting our proposed approach to a large econometric literature \citep{abraham2018estimating, Callaway2018}. We recommend this approach as a reasonable default in practice; it amounts to applying our partially pooled SCM estimator to de-meaned outcome series.
Second, we modify the SCM problem to incorporate auxiliary covariates alongside lagged outcomes. 
We also briefly address inference for SCM-like estimates in the staggered adoption setting.
We implement the proposed methodology in the \texttt{augsynth} package for \texttt{R}, available at \href{https://github.com/ebenmichael/augsynth}{\texttt{https://github.com/ebenmichael/augsynth}}.

We apply our methods to estimating the impact of mandatory teacher collective bargaining and show that they achieve better pre-treatment balance than existing approaches.
We find no impact of teacher collective bargaining laws on either teacher salaries or student expenditures, consistent with several recent papers  \citep{frandsen2016effects, paglayan2019public} but counter to earlier claims \citep[most notably][]{hoxby1996teachers}.

\paragraph{Related work.} Our paper contributes to several active methodological literatures. First, there is a large and active applied econometrics literature on challenges and remedies for two-way fixed effects models with multiple treated units;
see \citet{borusyak2017revisiting, abraham2018estimating, athey2018design, goodman2018difference, Callaway2018}. See also \citet{Xu2017} and \citet{athey2017mcp} for recent generalizations of these models.

SCM has also attracted a great deal of attention; see \citet{abadie2019synthreview} for a recent review.
Several recent papers have explored SCM with multiple treated units. In the case where all units adopt treatment at the same time, some propose to first average the units and then estimate SCM weights for the average, analogous to our fully pooled SCM estimate; for discussion, see \citet{kreif2016examination, Robbins2017}. 
An alternative is \citet{Abadie_LHour}, who instead propose to estimate separate SCM weights for each treated unit. 
In particular, they propose a penalized SCM approach that aims to reduce interpolation bias, allowing for weights that move continuously between standard SCM and nearest-neighbor matching.
Our approach complements these papers by adapting some of these ideas to the staggered adoption setting.
For some other examples of SCM under staggered adoption, see also \citet{dube2015pooling, toulis2018testing,  donohue2019right, cao2019synthetic}.

\paragraph{Motivating example: Teacher collective bargaining.}
\label{sec:example}
The United States, like other developed countries, spends substantial resources on public education. Approximately 80\% of education spending goes to teacher salaries and benefits \citep{nces_facts}, and research points to teacher quality as a key determinant of student outcomes \citep{jackson2014teacher}. Over recent decades, the teacher employment relationship has changed dramatically via the introduction of unions and collective bargaining agreements \citep{goldstein2015teacher}. Critics identify these as a ``harmful anachronism'' and ``the most daunting impediments'' to education reform \citep{hess2006better}, while proponents argue that collective bargaining raises pay and thereby helps to attract and retain high-quality teachers. 
A major 2018 Supreme Court decision, \emph{Janus v AFSCME}, is expected to weaken teachers' unions, bringing renewed attention to this area and raising interest in understanding the effects of teacher collective bargaining.

Since 1964, a number of states have passed laws mandating that school districts bargain with teachers' unions.\footnote{Another 10 states allow but do not require collective bargaining, while 7 prohibit it. We focus on estimating the effects of mandates.} 
Given the strong criticism directed at teachers' unions, there is surprisingly little evidence that they, or the mandatory bargaining laws, have any effect at all.
In a seminal study, \citet{hoxby1996teachers} uses state-level changes in collective bargaining laws to argue that teacher collective bargaining raises teacher salaries and school expenditures but reduces student outcomes. 
Several more recent papers have disputed Hoxby's conclusions, however.
Using a panel of school districts, \citet{lovenheim2009effect} 
finds little effect of unionization on teacher pay or class size. \citet{frandsen2016effects} similarly finds little effect of state unionization laws on teacher pay.
Finally,
\citet{paglayan2019public} extends the historical state-level data set from \citet{hoxby1996teachers}. Using a variant of the two-way fixed effect model, she finds precisely estimated zero effects of mandatory bargaining laws on per-pupil school expenditures\footnote{\citet{paglayan2019public} defines this as ``the total current operational expenditures (regardless of funding source) that are devoted to public schools in a state divided by the number of public school students in that state.''} and teacher salaries.
Motivated in part by recent criticisms of such models \citep{goodman2018difference}, we revisit the \citet{paglayan2019public} analysis using different methods.

\begin{figure}[tb]
      \centering \includegraphics[width=0.5\maxwidth]{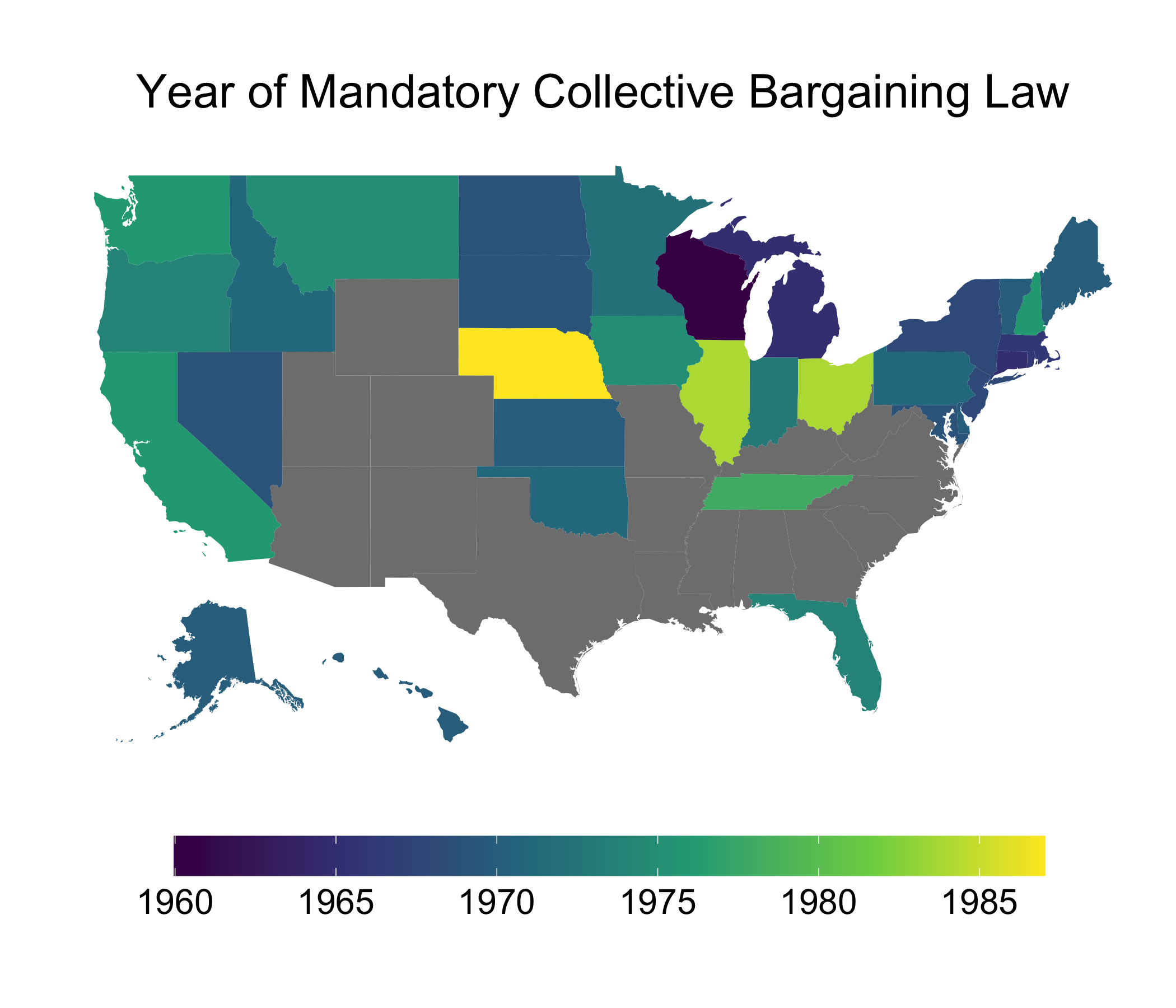}
    \caption{Staggered adoption of mandatory collective bargaining laws from 1964 to 1990.}
    \label{fig:state_map}
  \end{figure}

Figure \ref{fig:state_map} shows adoption times of state mandatory bargaining laws between 1964 and 1990. 
Adoptions were spread across 14 separate years, though 16 states adopted laws between 1965 and 1970. 
Following \citet{paglayan2019public}, our main outcomes of interest are per-pupil student expenditures and teacher salaries, both measured in log 2010 dollars.
We observe these outcomes back to 1959 for 49 states; we exclude Wisconsin, which adopted a mandatory bargaining law in 1960 and thus has only one year of pre-intervention data, as well as Washington, DC.
This gives between 6 and 28 years of data before the adoption of mandatory bargaining, with an average of 13 years.

\paragraph{Paper roadmap} 
Section \ref{sec:preliminaries} lays out the technical background and introduces the synthetic control estimator for a single treated unit.
Section \ref{sec:error_bounds} bounds the estimation error for general weighting estimators under two families of data generating process, an autoregressive model and a linear factor model, with staggered adoption. 
Section \ref{sec:generalizing_scm_overview} introduces partially pooled SCM as a solution to the problem of minimizing estimation error and considers two special cases, separate SCM and pooled SCM. 
Section \ref{sec:extensions} proposes several important extensions, including incorporating an intercept shift and auxiliary covariates, and briefly discusses inference. 
Section \ref{sec:sim_study_main} describes a calibrated simulation study.
Section \ref{sec:application} gives additional results for the teacher collective bargaining application.
Finally, Section \ref{sec:discussion} discusses some directions for future work.
The appendix includes further analyses and technical results. In particular, we provide an alternative motivation for our proposed partially pooled estimator, which we show is based on partially pooling parameters in the Lagrangian dual of the SCM constrained optimization problem.

\section{Preliminaries}
\label{sec:preliminaries}

\subsection{Setup and notation}
\label{sec:setup}
We consider a panel data setting where we observe outcomes $Y_{it}$ for $i = 1, \ldots, N$ units over $t = 1, \ldots, T$ time periods. 
In the teacher collective bargaining application, $N = 49$ and $T = 39$ years. 
Some but not all 
of the units adopt the treatment during the panel; 
once units adopt treatment, they stay treated for the remainder of the panel. 
Let $T_i$ represent the time period that unit $i$ receives treatment, with $T_i=\infty$ denoting never-treated units. Without loss of generality, we order units so that $T_1\leq T_2 \leq \dots \leq T_N$.
We assume that there are a non-zero number of never-treated units, $N_0\equiv N-\sum_i \bbone_{T_i = \infty}$, and we let $J=N-N_0=\sum_i \bbone_{T_i \neq \infty}$.
To clearly differentiate units that are eventually treated, we index them by $j=1,\dots,J$.

We adopt a potential outcomes framework to express causal quantities~\citep{neyman1923, rubin1974} and assume stable treatment and no interference between units \citep[SUTVA;][]{rubin1980}.
In principle, each unit $i$ in each time $t$ might have a distinct potential outcome for each potential treatment time $s$, $Y_{it}(s)$, for $s=1,\ldots,T,\infty$.
Following \citet{athey2018design},  we assume that prior to treatment, a unit's potential outcomes are equal to its never-treated potential outcome:
\begin{assumption}[No anticipation]
  \label{a:anticipation}
  $Y_{it}(s) = Y_{it}(\infty)$ for $t < s$, with treatment time $s$.
\end{assumption}

\noindent This relatively innocuous assumption generalizes the consistency assumption typically employed in cross-sectional studies. We maintain it throughout.
With it, the observed outcome is
$Y_{it} = \bbone\{t < T_i\} Y_{it}(\infty) + \bbone\{t \geq T_i\}Y_{it}(T_i)$.

\subsection{Estimands}
As is common in many panel data settings, we focus on effects a specified duration after treatment onset, known as \emph{event time}.
For treated unit $j$, we index event time relative to treatment time $T_j$ by $k=t-T_j$.
The unit-level treatment effect for treated unit $j$ at event time $k$ is the difference between the potential outcome at time $T_j + k$ under treatment at time $T_j$ and under never treatment:
$$\tau_{jk} = Y_{jT_j+k}(T_j)-Y_{jT_j+k}(\infty).$$
By Assumption \ref{a:anticipation}, $\tau_{jk}=0$ for any $k<0$. 

The unit-specific effects, $\tau_{jk}$, are often the central quantities of interest in many synthetic controls analyses. In addition to these effects, we also focus on their average. 
Our primary averaged estimand 
is the Average Treatment Effect on the Treated (ATT) $k$ periods after treatment onset:
$$\text{ATT}_k \equiv \frac{1}{J}\sum_{j=1}^J \tau_{jk} = \frac{1}{J} \sum_{j = 1}^J Y_{j,T_j+k}(T_j)-Y_{j,T_j+k}(\infty).$$
We are also interested in the average post-treatment effect, averaging across $k$: $\text{ATT} = \frac{1}{K}\sum_{k=1}^K \text{ATT}_k$.  Our methods generalize to many other estimands; see \citet{Callaway2018} for examples in this setting.

A challenge for staggered adoption analyses is that a panel that is balanced in calendar time is necessarily imbalanced in event time. That is, we observe outcomes $\ell$ periods before treatment only for units treated after period $\ell$, and we observe outcomes $k$ periods after treatment only for treated units treated before $T-k$. This means that populations of treated units over which one can average treatment effects vary with $k$, as do the possible donors. 
To minimize this problem,
we assume that all treated units are observed for at least several periods before being treated (i.e., $T_1 \gg 1$)
and for at least $K\geq0$ periods after treatment ($T_J\leq T-K$). 
For treated unit $j$, we will consider outcomes up to $L_j \leq T_j - 1$ periods before treatment, with $L \equiv \max_{j\leq J} L_j$ denoting the maximum number of lagged outcomes.

With this, the challenge in estimating $\text{ATT}_k$ for $k\leq K$ is to impute the average of the missing never-treated potential outcomes.
We define the set of possible ``donor units'' for treated unit $j$ at event time $k$ as those units $i$ for which we observe $Y_{iT_j+k}(\infty)$, which we denote $\calD_{jk}\equiv \{i: T_i > T_j+k\}$. The composition of $\calD_{jk}$ varies with both treated unit $j$ and event time $k$; in particular, unit $i$ with $T_i<\infty$ is in $\calD_{jk}$ for $k < T_i - T_j$ but not for $k \geq T_i-T_j$. We focus on fixed donor pools $\calD_{jK}$ rather than allowing the donor pools to vary with $k$. 
This limits the number of potential donors, but ensures that estimated counterfactual outcomes do not vary spuriously across event time due to changing composition of the donor pool. Our proposed estimator does not require this restriction, but it greatly simplifies exposition.
If $K \geq T_J - T_1$ then $\calD_{jk}$ will only include never treated units as donors; otherwise $\calD_{jk}$ will include \emph{both} never treated and not-yet-treated units.

In our empirical application we exclude Wisconsin --- which adopted a mandatory collective bargaining law in the second year of the sample --- so the first treated state is Connecticut with $T_1 = 7$. We follow \citet{paglayan2019public} in considering treatment effects only up to event time $K = 10$, and use as potential donors for treated state $j$ any states that are not treated by $T_j+10$.

\subsection{Restrictions on the data generating process}
\label{sec:restrict}

We now detail various restrictions on the data generating process that we will consider below. Because we are interested in treatment effects on treated units --- and observe potential outcomes under treatment --- we will place restrictions only on the potential outcomes under the never treated condition $Y_{it}(\infty)$. Throughout, we follow \citet{chernozhukov2017exact} and \citet{BenMichael_2018_AugSCM} and write these potential outcomes as a model component plus additive noise.
We consider two alternative restrictions on the model terms and noise terms, corresponding to two common data generating processes for $Y_{it}(\infty)$: a time-varying autoregressive process and a linear factor model.
\begin{assumption}[Data generating processes] 
We consider the following:
  \begin{enumerate}[label = (\alph*),ref={\theassumption\alph*}]
    \item \label{a:ar}
    The untreated potential outcomes $Y_{it}(\infty)$ follow a time-varying AR($L$) process 
  \begin{equation}
      \label{eq:time_ar}
      Y_{it}(\infty) = \sum_{\ell=1}^L \rho_{t\ell} Y_{i t - \ell}(\infty) + \varepsilon_{it},
  \end{equation}
  where $\varepsilon_{it}$ are mean zero and independent across units and time, with $\varepsilon_{is+k} \indep \bbone\{T_i = s\}$ 
  for $k \geq 0$ for all $i=1,\ldots,N$.

  \item \label{a:lfm}
  There are $F$ latent time-varying factors, where $F$ is typically small relative to both $N$ and $T$. The factors, $\mu_t \in \R^F$, are bounded, $\max_{t} \|\mu_t\|_\infty \leq M$.
  Each unit has a vector of time-invariant factor loadings $\phi_i \in \R^F$, and the untreated potential outcomes $Y_{it}(\infty)$ are generated as:
\begin{equation}
    \label{eq:factor_model}
    Y_{it}(\infty) = \phi_i \cdot \mu_t + \varepsilon_{it},
\end{equation}
where
$\varepsilon_{it}$ are mean zero, independent across units and time and 
$\varepsilon_{it} \indep T_i$ for all $i=1,\ldots,N$, $t=1,\ldots,T$.
  \end{enumerate}
\end{assumption}

\noindent Assumptions \ref{a:ar} and \ref{a:lfm} impose different restrictions on the noise terms. Assumption \ref{a:lfm} rules out correlation between treatment timing and the noise terms for any period while Assumption \ref{a:ar} only excludes correlation for noise terms \emph{after} treatment. Therefore, under Assumption \ref{a:lfm} treatment timing and pre-treatment outcomes are only dependent through the factor loadings, while under Assumption \ref{a:ar} there is no restriction on their dependence.

Finally, under each process, we assume that the noise terms do not have fat tails.
\begin{assumption}
  \label{a:subgaus}
  $\varepsilon_{it}$ are sub-Gaussian random variables with scale parameter $\sigma$.
\end{assumption}
\noindent We use this restriction on the tail behavior for the finite sample estimation error bounds we introduce in Section \ref{sec:error_bounds}.

\subsection{The Synthetic Control Method}
\label{sec:separate_scm_intro}

In the synthetic control method (SCM), the counterfactual outcome under control is estimated from a weighted average, known as a \emph{synthetic control}, of untreated units, where weights are chosen to minimize the squared imbalance between the lagged outcomes for the treated unit and the weighted control (``donor'') units. 

We consider a modified version of the original SCM estimator of  \citet{AbadieAlbertoDiamond2010, Abadie2015} for a single treated unit $j$. 
In this version, the SCM weights $\hat{\gamma}_j$ are the solution to a constrained optimization problem:
\begin{equation}
  \label{eq:scm_unit_j}
  \min_{\gamma_j \in \Delta^{\rm{scm}}_j} \;\; \underbrace{\frac{1}{L_j}\sum_{\ell = 1}^{L_j} \left(Y_{j T_j-\ell} \;-\; \sum_{i = 1}^N \gamma_{ij} Y_{iT_j-\ell}\right)^2}_{\text{objective}} \; + \; \underbrace{\lambda \sum_{i=1}^N \gamma_{ij}^2}_{\text{regularization}},
\end{equation}
where $\gamma_j \in \Delta^{\rm{scm}}_j$ has elements $\{\gamma_{ij}\}$ that satisfy $\gamma_{ij}\geq 0$ for all $i$, $\sum_{i} \gamma_{ij}=1$, and $\gamma_{ij}=0$ whenever $i$ is not a possible donor, $i \not \in \calD_{jK}$.

Given an $N$-vector of weights $\hat{\gamma}_{ij}$ that solve Equation \eqref{eq:scm_unit_j}, the SCM estimate of the missing potential outcome for treated unit $j$ at event time $k$, $Y_{jT_j + k}(\infty)$, is:
$$\hat{Y}_{jT_j + k}(\infty) = \sum_{i = 1}^N \hat{\gamma}_{ij}\, Y_{iT_j + k},$$
with estimated treatment effect
$\hat{\tau}_{jk} = Y_{jT_j+k} - \hat{Y}_{jT_j + k}(\infty)$. This formulation can also be applied when $k<0$, generating \emph{placebo} treatment effect estimates, often referred to as ``gaps.''
We denote the vector of 
placebo pre-treatment effect estimates
as $\hat{\tau}_j^\pre = (\hat{\tau}_{j(-L)},\ldots,\hat{\tau}_{j(-1)}) \in \R^L$, where we define $\hat{\tau}_{j(-\ell)}$ to be zero for $\ell > L_j$.
With this notation, the synthetic controls objective in Equation \eqref{eq:scm_unit_j} is the mean squared placebo treatment effect on pre-treatment outcomes:
\begin{equation}
	\label{eq:SCM1_balance}
(q_j(\hat{\gamma}_j))^2 \;\;\equiv\;\; \frac{1}{L_j} \left\|\hat{\tau}_j^\pre\right\|_2^2 \;\;=\;\; \frac{1}{L_j}\sum_{\ell = 1}^{L_j} \left(Y_{j T_j-\ell} \;-\; \sum_{i = 1}^N \hat{\gamma}_{ij} Y_{iT_j-\ell}\right)^2.
\end{equation}

The optimization problem in Equation \eqref{eq:scm_unit_j} modifies the original SCM proposal in two key ways. First, where \citet{AbadieAlbertoDiamond2010, Abadie2015} balance auxiliary covariates, we focus exclusively on lagged outcomes; we re-introduce auxiliary covariates in Section \ref{sec:covs}.
Second, following a suggestion in \citet{Abadie2015}, we include a term that penalizes the weights toward uniformity, with hyperparameter $\lambda$. While we penalize the sum of the squared weights, there are many options, e.g., an entropy or elastic net penalty \citep[see][]{Doudchenko2017, Abadie_LHour}. 
In settings where it is possible to achieve perfect balance, selecting $\lambda > 0$ ensures that Equation \eqref{eq:scm_unit_j} has a unique solution. This is not the case in our setting, however, and so we largely view this term as a technical convenience. 

\citet{abadie2019synthreview} gives several reasons for preferring SCM to outcome models such as linear regression or directly fitting the factor model. In particular, SCM weights are guaranteed to be non-negative, and are generally sparse and interpretable. By contrast, alternatives based on explicit models for $Y_{it}(\infty)$ often imply negative weights and thus unchecked extrapolation outside the support of the donor units.
Outcome modeling can also be sensitive to model mis-specification, such as selecting an incorrect number of factors in a factor model.
Finally, as we emphasize in our theoretical results in the next section, SCM can be appropriate under multiple data generating processes (e.g., both the autoregressive model and the linear factor model) so that it is not necessary for the applied researcher to take a strong stand on which is correct.

A central question for SCM is how to assess whether $\hat{Y}_{j T_j + k}(\infty)$ is a reasonable estimate for $Y_{j T_j + k}(\infty)$.
A minimal condition is that the SCM weights achieve a low root mean squared placebo treatment effect, i.e., $q_j(\hat{\gamma}_j)$ is close to zero.
If it is not close to zero, there is a concern that estimated effects also capture systematic differences between
$\hat{Y}_{j T_j + k}(\infty)$ and $Y_{j T_j + k}(\infty)$.
Under versions of either Assumptions \ref{a:ar} or \ref{a:lfm} and for a single treated unit, 
\citet{AbadieAlbertoDiamond2010} show that if $q_j(\hat{\gamma}_j) = 0$ then the bias will tend to zero as $L_j \to \infty$, and 
\citet{BenMichael_2018_AugSCM} bound the estimation error of $\hat{\tau}_{jk}$ in terms of $q_j(\hat{\gamma}_j)$.
\citet{AbadieAlbertoDiamond2010, Abadie2015} recommend that researchers only proceed with an SCM analysis if the pre-treatment fit is excellent, while \citet{BenMichael_2018_AugSCM} propose an augmented SCM estimator that attempts to salvage cases where it is not.

\section{Estimation error under staggered adoption}
\label{sec:error_bounds}

In order to extend SCM to the staggered adoption setting, we first develop 
appropriate balance measures for synthetic control-style weighting estimators under staggered adoption. We use these to develop bounds on the estimation error for the ATT for our two example data generating processes. These bounds in turn motivate our proposal for partially pooled SCM as a way to choose weights under staggered adoption.

\subsection{Weights and measures of balance}
With multiple treated units, we can generalize the above setup to allow for weights for each treated unit. For each $j\leq J$, let $\gamma_j \in \Delta^\scm_j$ be an $N$-vector of weights on potential donor units, where $\gamma_{ij}$ is the weight on unit $i$ in the synthetic control for treated unit $j$. We collect the weights into an $N$-by-$J$ matrix $\Gamma = [\gamma_1,\ldots,\gamma_J] \in \Delta^\scm$, where $\Delta^\scm = \Delta_1^\scm \times \ldots \times \Delta_J^\scm$.
The estimated treatment effect on unit $j$ at event time $k$ is then $\hat{\tau}_{jk}$ as defined above, and the estimated ATT averages over the unit-level effect estimates: 
\begin{equation}\label{eq:att_estimator}
  \widehat{\text{ATT}}_k 
        = \frac{1}{J} \sum_{j=1}^J \hat{\tau}_{jk}
  		= \frac{1}{J} \sum_{j=1}^J \left[ Y_{j T_j+k} - \sum_{i = 1}^N \hat{\gamma}_{ij}\, Y_{i T_j + k}\right] = \frac{1}{J} \sum_{j=1}^J Y_{jT_j+k} - \sum_{i=1}^N \sum_{j=1}^J \frac{\hat{\gamma}_{ij}}{J} Y_{iT_j+k}.
\end{equation}
Equation \eqref{eq:att_estimator} highlights two equivalent interpretations of the estimator: as the average of unit-specific SCM estimates and as an SCM estimate for the average treated unit.

Using the two interpretations of the ATT estimator in Equation \eqref{eq:att_estimator}, we construct goodness-of-fit measures for the ATT by aggregating  $\hat{\tau}_j^\pre$ in two ways. 
First, we consider the root mean square of the pre-treatment fits across treated units,
$$
q^{\sep}(\widehat{\Gamma}) \equiv \sqrt{\frac{1}{J}\sum_{j=1}^J q_j^2(\hat{\gamma}_j)} =  \sqrt{\frac{1}{J}\sum_{j=1}^J \frac{1}{L_j}\|\hat{\tau}_j^\pre\|_2^2} = \sqrt{\frac{1}{J}\sum_{j=1}^J \frac{1}{L_j}\sum_{\ell = 1}^{L_j} \left(Y_{j T_j-\ell} \;-\; \sum_{i = 1}^N \hat{\gamma}_{ij} Y_{iT_j-\ell}\right)^2}.
$$
This is a useful measure of overall imbalance when SCM is estimated separately for each treated unit and generalizes the objective for the single synthetic control problem.
Second, we consider the pre-treatment fit for the average of the treated units,
$$
q^{\pool}(\widehat{\Gamma}) \equiv \frac{1}{\sqrt{L}} \left\|\frac{1}{J}\sum_{j=1}^J \hat{\tau}^\pre_j\right\|_2 = \sqrt{\frac{1}{L} \sum_{\ell = 1}^{L}\left[ \frac{1}{J}\sum_{T_j > \ell} Y_{j T_j - \ell} - \sum_{i=1}^N \hat{\gamma}_{ij}Y_{i T_j - \ell}\right]^2}.
$$
We refer to this interchangeably as the \emph{pooled} or \emph{global} fit.

Both $q^\pool$ and $q^\sep$ are on the same scale as the estimated treatment effect, $\widehat{\text{ATT}}_{k}$. However, the measures differ in whether they average \emph{before} or \emph{after} evaluating the pre-treatment fit. 
Thus, we typically expect $(q^{\pool})^2 \ll (q^{\sep})^2$, since the lagged outcomes for the \emph{average} of the treated units are less extreme than the lagged outcomes for the units themselves.
In practice, we therefore consider normalizing the imbalance measures by their values computed with weights $\hat{\Gamma}^{\sep}$, the set of solutions to Equation \eqref{eq:scm_unit_j} applied separately to each treated unit. 
We define $\tilde{q}^{\pool}(\Gamma) \equiv \nicefrac{q^{\pool}(\Gamma)}{q^{\pool}(\hat{\Gamma}^{\sep})}$ and $\tilde{q}^{\sep}(\Gamma) \equiv \nicefrac{q^{\sep}(\Gamma)}{q^{\sep}(\hat{\Gamma}^{\sep})}$.
We use these normalized measures in our proposed estimator in Section \ref{sec:partial_pool_scm} below.

Ideally, both $q^\sep$ and $q^\pool$ would be close to zero; indeed if $q^\sep=0$ then $q^\pool=0$ is also zero.
When this is not possible, there is a trade off between these two sources of imbalance. 
Our proposed ``partially pooled'' SCM estimator generalizes Equation \eqref{eq:scm_unit_j} to minimize a weighted average of their normalized squares, $\nu (\tilde{q}^{\pool})^2 + (1-\nu) (\tilde{q}^{\sep})^2$, where $\nu$ is a hyperparameter selected by the researcher. To motivate this and to inform the choice of $\nu$, we develop error bounds for SCM-style weights under our two data generating models.

\subsection{Error bounds}

\subsubsection{Autoregressive model}
We first bound the estimation error for the ATT at event time $k=0$, $\text{ATT}_0$, under the autoregressive process in Assumption \ref{a:ar}. Two summaries of the autoregressive coefficients are important to our analysis: $\bar{\rho} = \frac{1}{J}\sum_{j=1}^J \rho_{T_j}$, the \emph{average} autoregression coefficient across the $J$ treatment times, and $S^2_\rho \equiv \frac{1}{J}\sum_{j=1}^J \|\rho_{T_j} - \bar{\rho}\|_2^2$, the corresponding \emph{variance}; under simultaneous adoption $S^2_\rho = 0$.
  \begin{theorem}
    \label{thm:time_ar_error}
   Under Assumptions \ref{a:ar} and \ref{a:subgaus} with $L_j = L < T_1$ for $j=1,\ldots,J$, for $\hat{\Gamma} \in \Delta^\scm$, where $\hat{\gamma}_j$ is independent of $\varepsilon_{\cdot T_j + k}$, and for any $\delta > 0$, the error for $\widehat{\text{ATT}}_0$ is
      \[
      \begin{aligned}
        \left | \widehat{\text{ATT}}_0 - \text{ATT}_0 \right | & \leq \underbrace{\sqrt{L}\|\bar{\rho}\|_2 \; q^\pool(\hat{\Gamma})}_{\text{pooled fit}}
        + \underbrace{\sqrt{L}S_\rho \; q^\sep(\hat{\Gamma})}_{\text{unit-specific fit}}
        &+ \underbrace{\frac{\delta\sigma}{\sqrt{J}} \left(1 + \|\Gamma\|_F\right)}_{\text{noise}}
      \end{aligned}
      \]
  with probability at least $1 - 2e^{-\frac{\delta^2}{2}}$,  where for a matrix $A \in \R^{n \times m}$, $\|A\|_F = \sqrt{\sum_{i=1}^n\sum_{j=1}^m A_{ij}^2}$ is the Frobenius norm.
  \end{theorem}

\noindent Theorem \ref{thm:time_ar_error} shows that the error for the ATT is bounded by several distinct terms, giving guidance for the choice of the weights $\Gamma$.
First, error arises from the level of both the global fit and the unit-specific fits. The relative importance of these fits is governed by the ratio of the average coefficient value $\|\bar{\rho}\|_2$ and the standard deviation $S_\rho$ for the autoregressive coefficients over time.

Second, there is error due to post-treatment noise, inherent to any weighting method. Because the weights are independent of post-treatment outcomes, this term has mean zero and enters the finite sample bound above through the standard deviation, which is proportional to the Frobenius norm of the weight matrix, $\|\hat{\Gamma}\|_F$. Thus, when selecting among weight matrices that yield similar unit-specific and pooled balance, we should prefer the one that minimizes $\|\hat{\Gamma}\|_F$. This motivates a penalty term similar to that in Equation \eqref{eq:scm_unit_j}.

\subsubsection{Linear factor model}
Next we consider the linear factor model in Assumption \ref{a:lfm} and begin by defining additional notation. Let $\Omega_j \in \R^{L \times F}$ denote the matrix of factor values for time $T_j-L$ to $T_j - 1$, and denote $ P^{(j)} = \sqrt{L}(\Omega_j' \Omega_j)^{-1} \Omega_j' \in \R^{F \times L}$ as the scaled projection matrix from outcomes to factors. Analogous to the autoregressive process above, the average (projected) factor value across the $J$ treatment times, $\bar{\mu}_k = \frac{1}{J}\sum_{j=1}^J P^{(j) \prime}\mu_{T_j + k}$, and the \emph{variance}, $S^2_k = \frac{1}{J}\sum_{j=1}^J \|P^{(j) \prime}\mu_{T_j + k} - \bar{\mu}_k\|_2^2$, determine the relative importance of the pooled and unit-specific fits,  respectively. 
\begin{theorem}
  \label{thm:lfm_error}
  Assume that $\Omega_j$ is non-singular and $\|\frac{1}{\sqrt{L}}\Omega_j\|_2 = 1$ for $j=1,\ldots,J$. With $L_j = L < T_1$ for $j=1,\ldots,J$, $\hat{\gamma}_1,\ldots,\hat{\gamma}_J \in \Delta^\scm$ where $\hat{\gamma}_j$ is independent of $\varepsilon_{\cdot T_j + k}$, $K \geq 0$, and $\delta > 0$, under Assumptions \ref{a:lfm} and \ref{a:subgaus}
the error for $\widehat{\text{ATT}}_k$ is 
  \[
    \left|\widehat{\text{ATT}}_k - \text{ATT}_k\right|  \leq  \underbrace{\|\bar{\mu}_k\|_2 \; q^\pool(\widehat{\Gamma})}_{\text{pooled fit}} + \underbrace{S_k \; q^\sep(\widehat{\Gamma})}_{\text{unit-specific fit}} +
    \underbrace{\frac{\sigma M^2 F}{\sqrt{L}}\left( 3\delta  + 2\sqrt{\log NJ}\right)}_{\text{approximation error}} +\underbrace{\frac{\delta \sigma}{\sqrt{J}}\left(1 + \|\hat{\Gamma}\|_F\right)}_{\text{noise}}
  \]
with probability at least $1 - 6 e^{-\frac{\delta^2}{2}}$, where $\max_t \|\mu_t\|_\infty \leq M$.
\end{theorem}
Theorem \ref{thm:lfm_error} shows that under the linear factor model the error for the ATT can again be controlled by the level of pooled fit and unit-specific fits.
As in Theorem \ref{thm:time_ar_error}, the relative importance of these fits is governed by the ratio of the average factor value $\bar{\mu}_k$ and the standard deviation $S_k$;
similarly, under simultaneous adoption, $S_k = 0$ and $q^\sep$
does not enter the bound. 

Unlike in Theorem \ref{thm:time_ar_error}, this bound also includes an approximation error that arises due to balancing --- and possibly over-fitting to --- noisy outcomes rather than to the true underlying factor loadings. 
In the worst case, the $J$ synthetic controls match on the noise rather than the factors.
Constraining the weights to lie in the simplex reduces the impact of this worst case, however, and the error decreases as more lagged outcomes are balanced; see \citet{AbadieAlbertoDiamond2010, BenMichael_2018_AugSCM,Arkhangelsky2018} for further discussion.

\section{Partially Pooled SCM}
\label{sec:generalizing_scm_overview}
\label{sec:partial_pool_scm}

We now turn to our main proposal, \emph{partially pooled SCM}. 
Motivated by the finite sample error bounds in Theorems \ref{thm:time_ar_error} and \ref{thm:lfm_error}, 
this chooses SCM weights to minimize a weighted average of the (squared) pooled and unit-specific pre-treatment fits:
\begin{equation}
  \label{eq:stag_avg_relative_scm_primal}
  \begin{aligned}
   \min_{\Gamma\in \Delta^{\text{scm}}}  \;\;\;   &  \nu~(\tilde{q}^{\text{pool}}(\Gamma))^2 +
    (1-\nu) 
      ~(\tilde{q}^{\rm{sep}}(\Gamma))^2
    \;+\; 
    \lambda \|\Gamma\|_F^2.
  \end{aligned}
\end{equation}
\noindent The hyperparameter $\nu \in [0,1]$ governs the relative importance of the two objectives; higher values of $\nu$ correspond to more weight on the pooled fit relative to the separate fit.
In Appendix \ref{sec:sim_scm_dual}, we show that intermediate values of $\nu$ correspond to a partial pooling solution for the weights in the dual parameter space, motivating our choice of a name.

The optimization in Equation \eqref{eq:stag_avg_relative_scm_primal} differs from the bounds in Section \ref{sec:error_bounds} in two practical ways. 
First, we minimize the normalized imbalance measures (e.g.,  $\tilde{q}^{\text{pool}}$ rather than $q^{\pool}$), so that the minimum with $\nu = 0$ and $\lambda = 0$ is indexed to 1.
This ensures that the two objectives are on the same scale, regardless of the number of treated units, and makes it easier to form intuition about $\nu$.
Second, we minimize the squared imbalances, which permits a computationally feasible quadratic program.  
As with the single synthetic controls problem in Equation \eqref{eq:scm_unit_j}, we penalize the sum of the squared weights, $\|\Gamma\|_F^2$.

\subsection{Special cases: Separate SCM ($\nu = 0$) and Pooled SCM ($\nu = 1$)}

We first consider two special cases of Equation \eqref{eq:stag_avg_relative_scm_primal}, which correspond to extreme values of the hyperparameter $\nu$, and then consider intermediate cases.

To date, common practice for staggered adoption applications of SCM is to estimate separate SCM fits for each treated unit, then estimate the ATT by averaging the unit-specific treatment effect estimates.  
This approach, which we refer to as \emph{separate SCM}, minimizes $q^{\sep}$ alone and is equivalent to our proposal in Equation \eqref{eq:stag_avg_relative_scm_primal} with $\nu=0$.
Since this separate SCM strategy prioritizes the unit-specific estimates, $\hat{\tau}_{jk}$,
an important question is when this approach will also give reasonable estimates of $\text{ATT}_k$.
From Theorems \ref{thm:time_ar_error} and \ref{thm:lfm_error}, we can see that if the unit-specific fits are all excellent, then the estimation error $\left|\text{ATT}_k - \widehat{\text{ATT}}_k\right|$ will be small. 
This is not the case in our application, however.
Figure~\ref{fig:separate_scm_fits_intro} shows SCM ``gap plots'' of $\hat{\tau}_{j\ell}$ against $\ell$ for three illustrative treated states, taken one at a time. While Ohio shows relatively good pre-treatment fit, there are no synthetic controls that closely track Illinois or New York's pre-treatment outcomes. Thus, simply averaging the estimated treatment effects across these three states without attention to the overall fit does not yield a convincing estimate. Other recent applications also face the same issue where several treated units have poor pre-treatment fit \citep[see e.g.][]{dube2015pooling, donohue2019right}.\footnote{One way to address this is to trim the sample and drop treated units with poor pre-treatment fit, noting that this changes the estimand.}

\begin{figure}[htb]
\centering
  \begin{subfigure}[t]{0.45\textwidth}  
{\centering \includegraphics[width=\maxwidth]{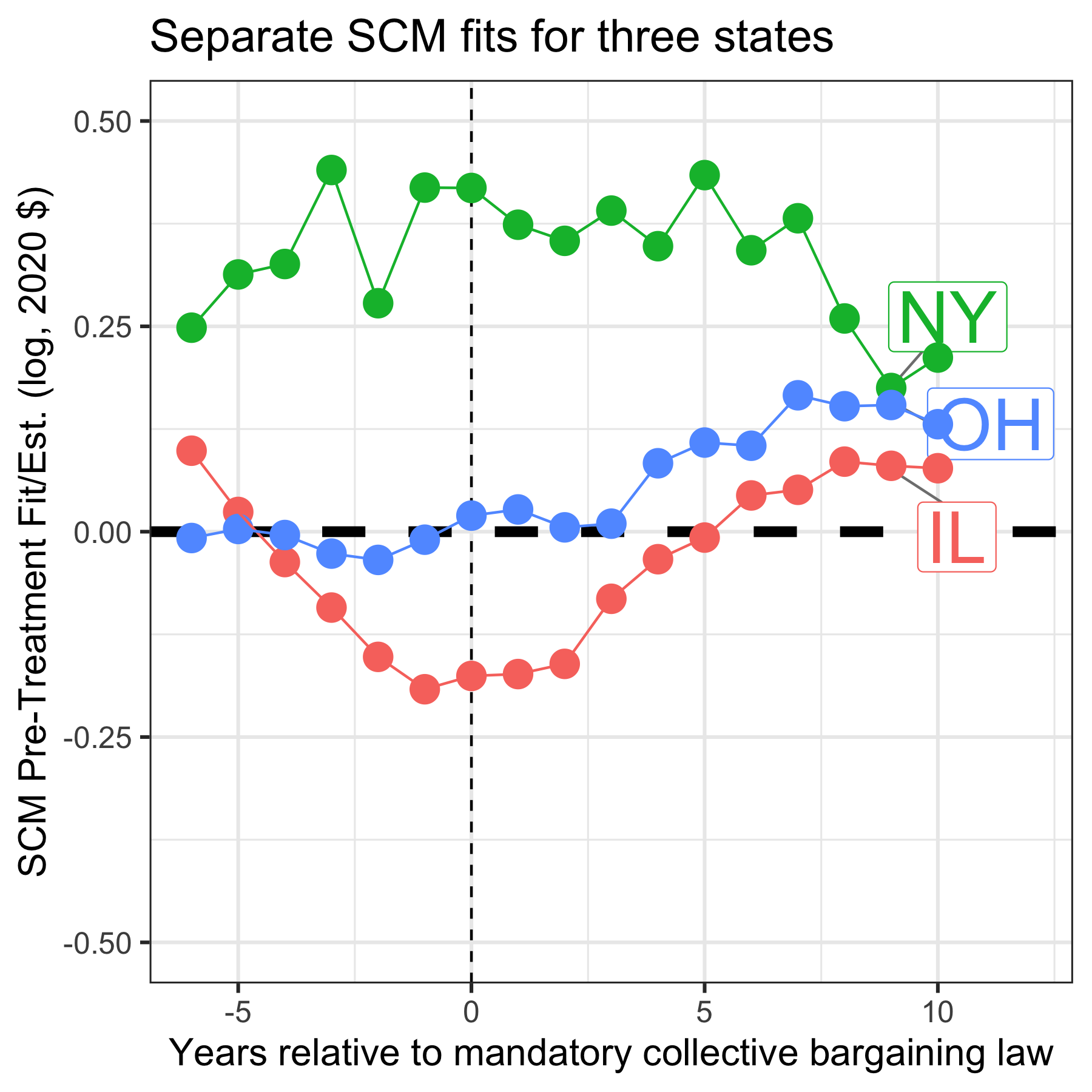} 
}
\caption{SCM ``gap plots'' for three illustrative states} 
  \label{fig:separate_scm_fits_intro}
  \end{subfigure}\quad
      \begin{subfigure}[t]{0.45\textwidth}  
  {\centering \includegraphics[width=\maxwidth]{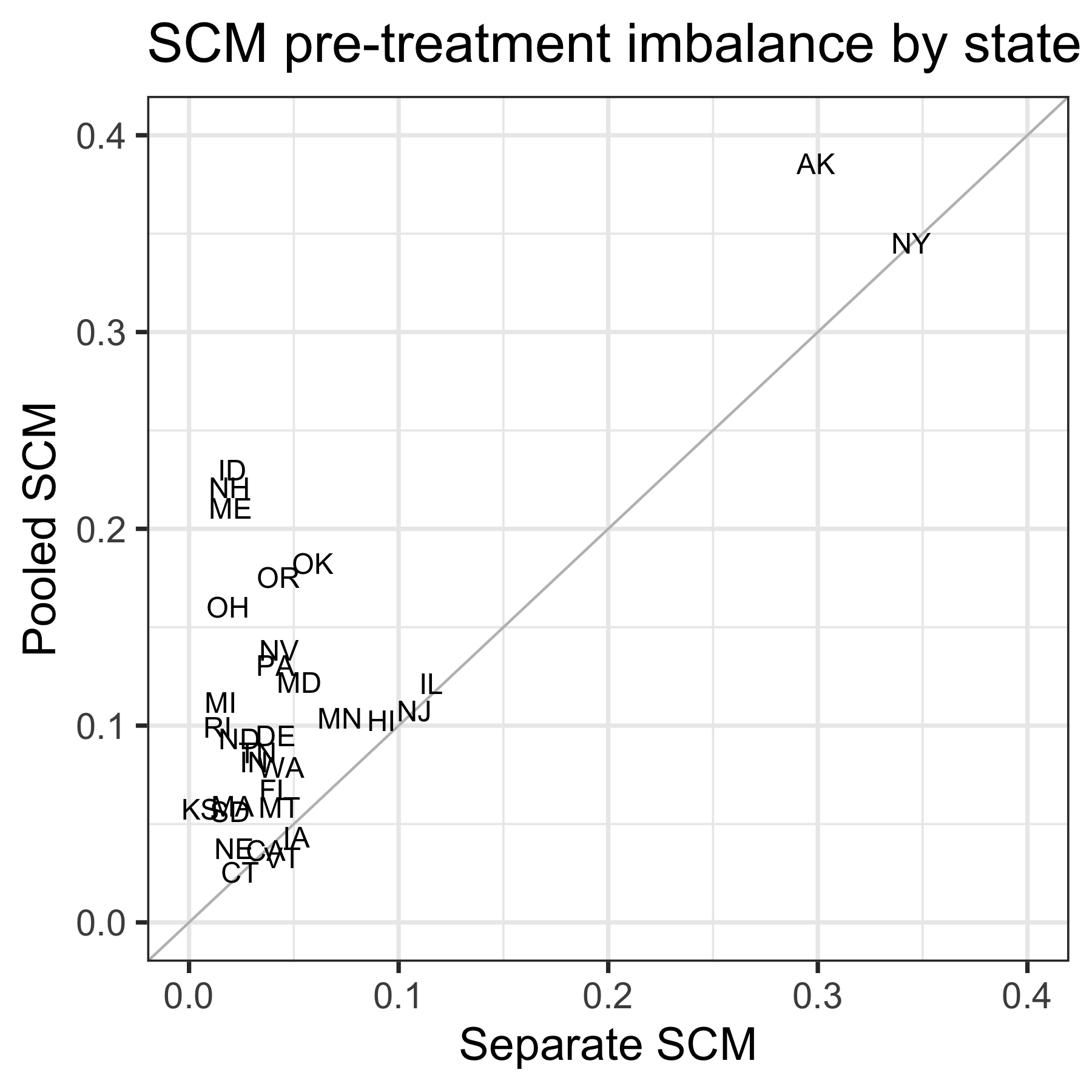} 
  }
  \caption{SCM pre-treatment fits by state}
    \label{fig:SCM_RMSE_comparison_no_vs_all_pool_ppexp}
    \end{subfigure} \\[2ex]

  \caption{(a) SCM pre-treatment fit for three states: (i) Ohio, with good overall fit, (ii) Illinois, where SCM fails to match an important pre-treatment trend, and (iii) New York, with pre-treatment imbalance roughly an order of magnitude larger than typical estimates for the impact of teacher mandatory bargaining.
  (b) SCM fits by state show that Separate SCM gives better pre-treatment fit than Pooled SCM for all treated states.} 
  \label{fig:SCM_no_all_pool_overview_plot}
\end{figure}

The other extreme case, which we refer to as \emph{pooled SCM}, instead sets $\nu=1$, finding weights that minimize $q^{\rm{pool}}$, the root mean squared placebo estimate of the ATT.
This ignores the unit-specific pre-treatment fits in the objective, resulting in poor unit-level synthetic controls and, in turn, leading to poor estimates of the unit-level treatment effects $\tau_{jk}$.
Furthermore,even if the ATT is the only estimand of interest,
Theorems \ref{thm:time_ar_error} and \ref{thm:lfm_error} indicate that Separate SCM is unlikely to control the error.
In particular, if the pooled weights do a poor job of matching individual treated units, the pooled synthetic control may involve a great deal of interpolation and the component of the error bound due to separate imbalance can be large.
In Section \ref{sec:sim_study_main} we validate through simulation that pooled SCM leads to substantially worse unit-level estimates than separate SCM, and also that
there are indeed settings where the bounds in Theorems \ref{thm:time_ar_error} and \ref{thm:lfm_error} do bind, leading to large error in pooled SCM estimates of the ATT.
See \citet{Abadie_LHour} for further discussion on interpolation bias in synthetic control settings.

There are special cases where only controlling $q^\pool$ with pooled SCM is sufficient, however. Theorems \ref{thm:time_ar_error} and \ref{thm:lfm_error} indicate that only the across-treated-unit variation in $\rho_{T_j+k}$ and $\mu_{T_j+k}$ lead to weight on the unit-specific fits. Thus, when this variation is zero, the ATT error bound is minimized with $\nu=1$.
As we discuss above, under simultaneous adoption, with $T_1=\ldots=T_J$, $S_{\rho}=0$ in the autoregressive model and $S_k=0$ in the linear factor model. 
The same arises in staggered adoption settings where the data generating process is homogeneous over time --- e.g., where $\rho_t \equiv \rho$ in the autoregressive model.
 It also holds approximately when the average autoregressive coefficient or factor values are large relative to the standard deviations --- i.e., $S_\rho \ll \bar{\rho}$ or $S_k \ll \bar{\mu}_k$, which could justify a choice of $\nu=1$.
Finally, when units are treated in cohorts (with $T_j=T_k$ for units in the same cohort), there is no variation in $\rho_t$ and $\mu_t$ across units in the same cohort. This suggests fully pooling (i.e., averaging) units that are treated at the same time, even if there is only partial pooling across treatment cohorts. We discuss this modification in Appendix \ref{sec:time_cohorts}.

\begin{figure}[]
  \centering
  \begin{subfigure}[t]{0.45\textwidth}  
{\centering \includegraphics[width=\maxwidth]{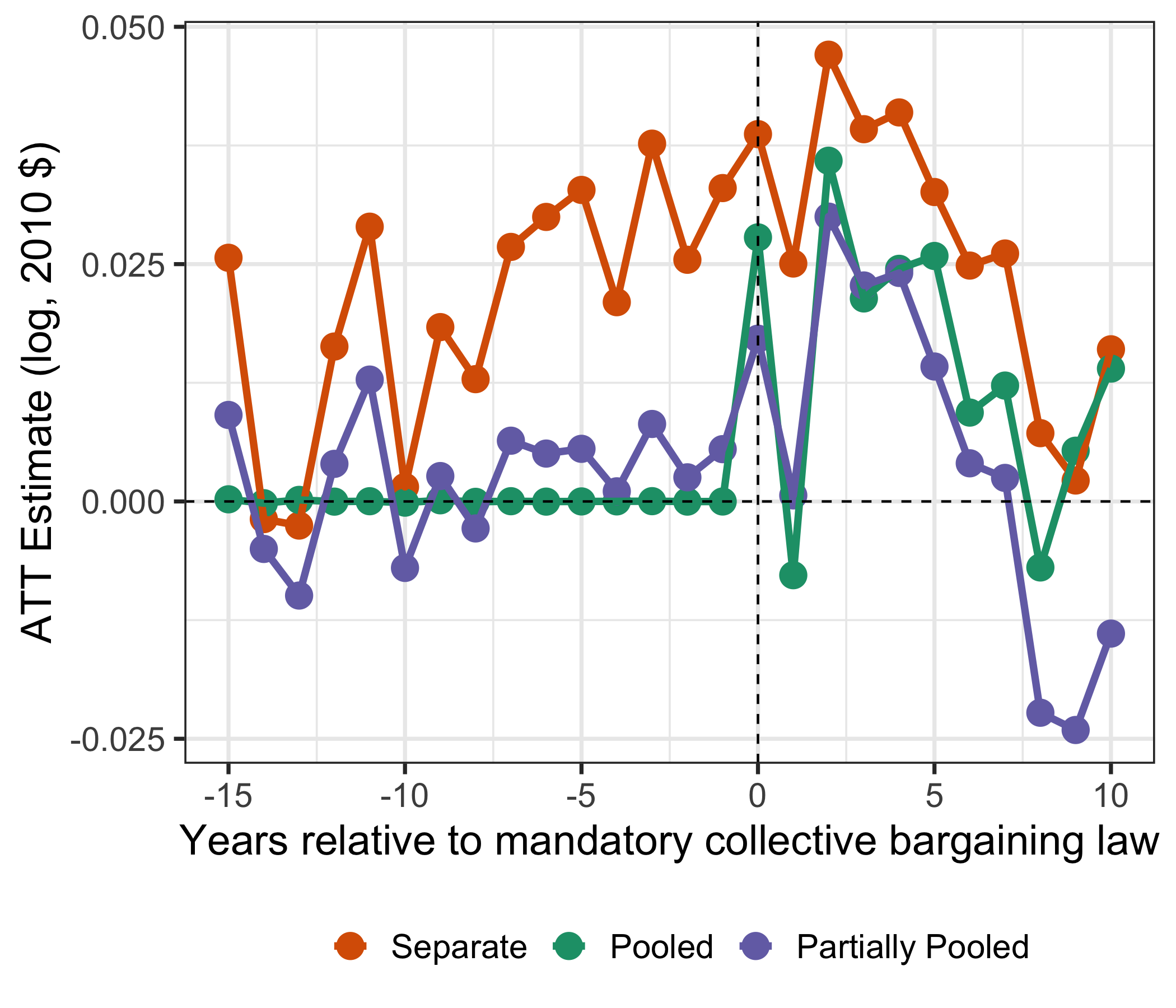} 
}
\caption{Estimated ATT on per-pupil expenditure (log, 2010 \$)} 
  \label{fig:scm_gaps}
  \end{subfigure}\quad
      \begin{subfigure}[t]{0.45\textwidth}  
  {\centering \includegraphics[width=\maxwidth]{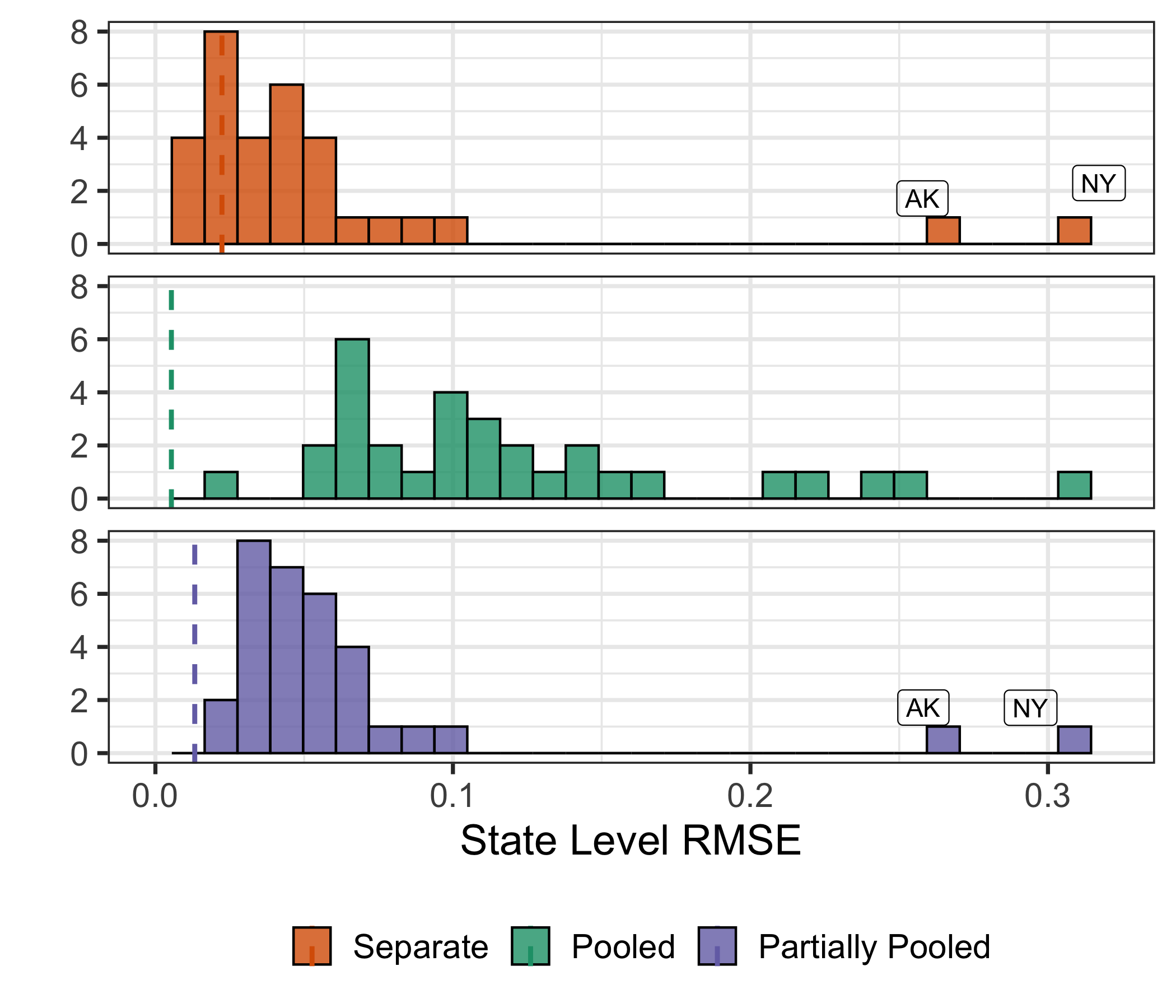} 
  }
  \caption{Distribution of state-level fits}
    \label{fig:scm_rmses}
    \end{subfigure} \\[2ex]
  \caption{(a) Series of estimated pre- and post-treatment effects $\widehat{\text{ATT}}_\ell$ and (b) state-level pre-treatment RMSE $\sqrt{\frac{1}{L}\sum_{\ell=1}^L \hat{\tau}_{j\ell}^2}$ using separate, pooled, and partially pooled SCM.}
  \label{fig:scm_plots}
\end{figure}

Figure \ref{fig:SCM_RMSE_comparison_no_vs_all_pool_ppexp} plots the state-level pre-treatment imbalances in our application for separate SCM versus pooled SCM. The separate SCM fit is better for all treated states, and so leads to more credible unit-level estimates.
However, these fits are far from perfect and so the results from Section \ref{sec:error_bounds} imply that there is room for improvement by controlling the pooled fit.
Figure \ref{fig:scm_gaps} shows the implied placebo estimates for the overall ATT  using the separate and pooled approaches: they are consistently positive for separate SCM weights and are all nearly zero for pooled SCM weights.
At the same time, Figure \ref{fig:scm_rmses} shows that pooled SCM has very poor unit-level fit, leading to the potential for error for \emph{both} the overall ATT estimate and the unit-level estimates. This motivates choosing an intermediate choice of $\nu \in (0,1)$.

\subsection{Intermediate choice of $\nu$}

\begin{figure}[tb]
  \centering
  \begin{subfigure}[t]{0.45\textwidth}  
  {\centering \includegraphics[width=\maxwidth]{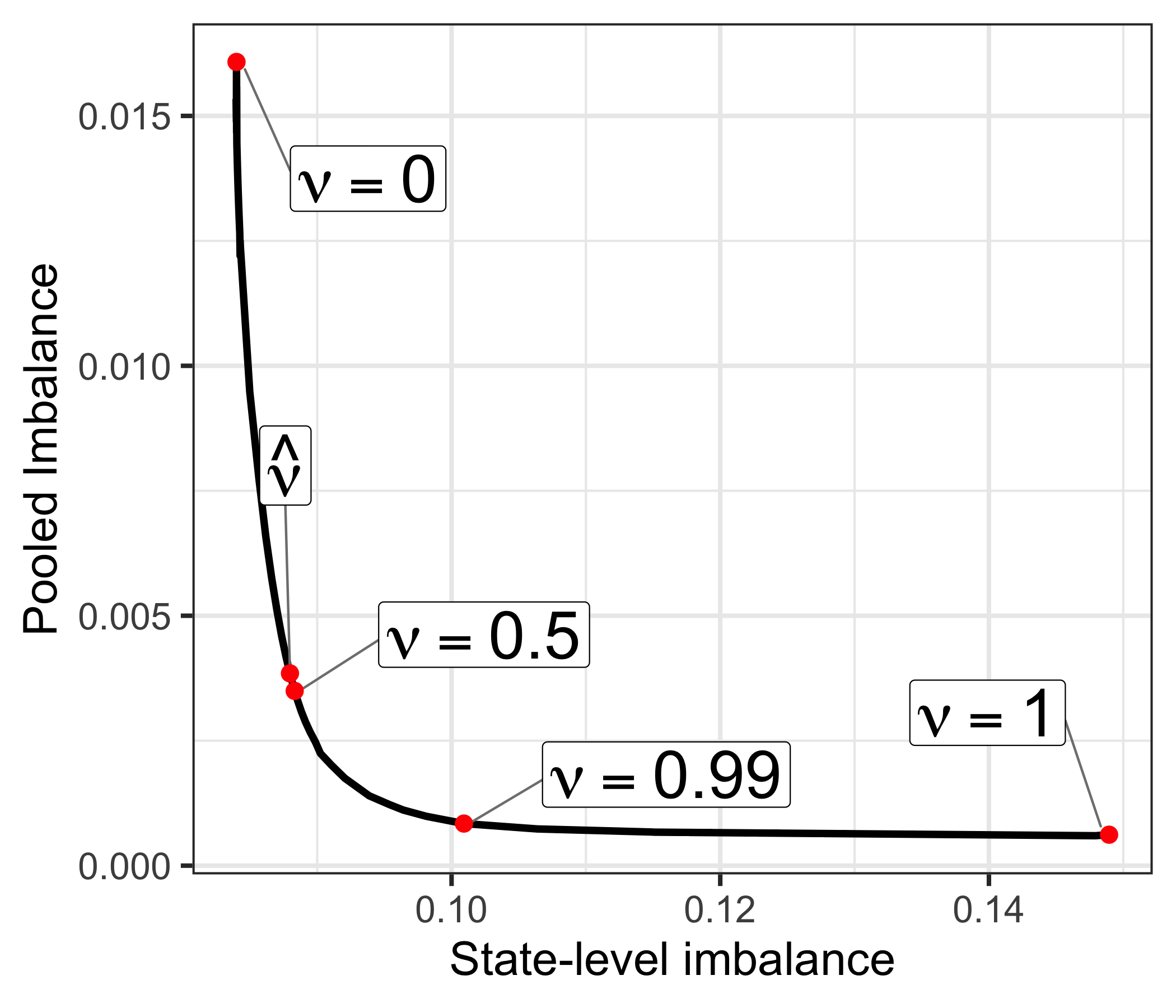} 
  }
  \caption{The \emph{balance possibility frontier}} 
    \label{fig:pareto_curve}
    \end{subfigure}\quad
    \begin{subfigure}[t]{0.45\textwidth}  
  {\centering \includegraphics[width=\maxwidth]{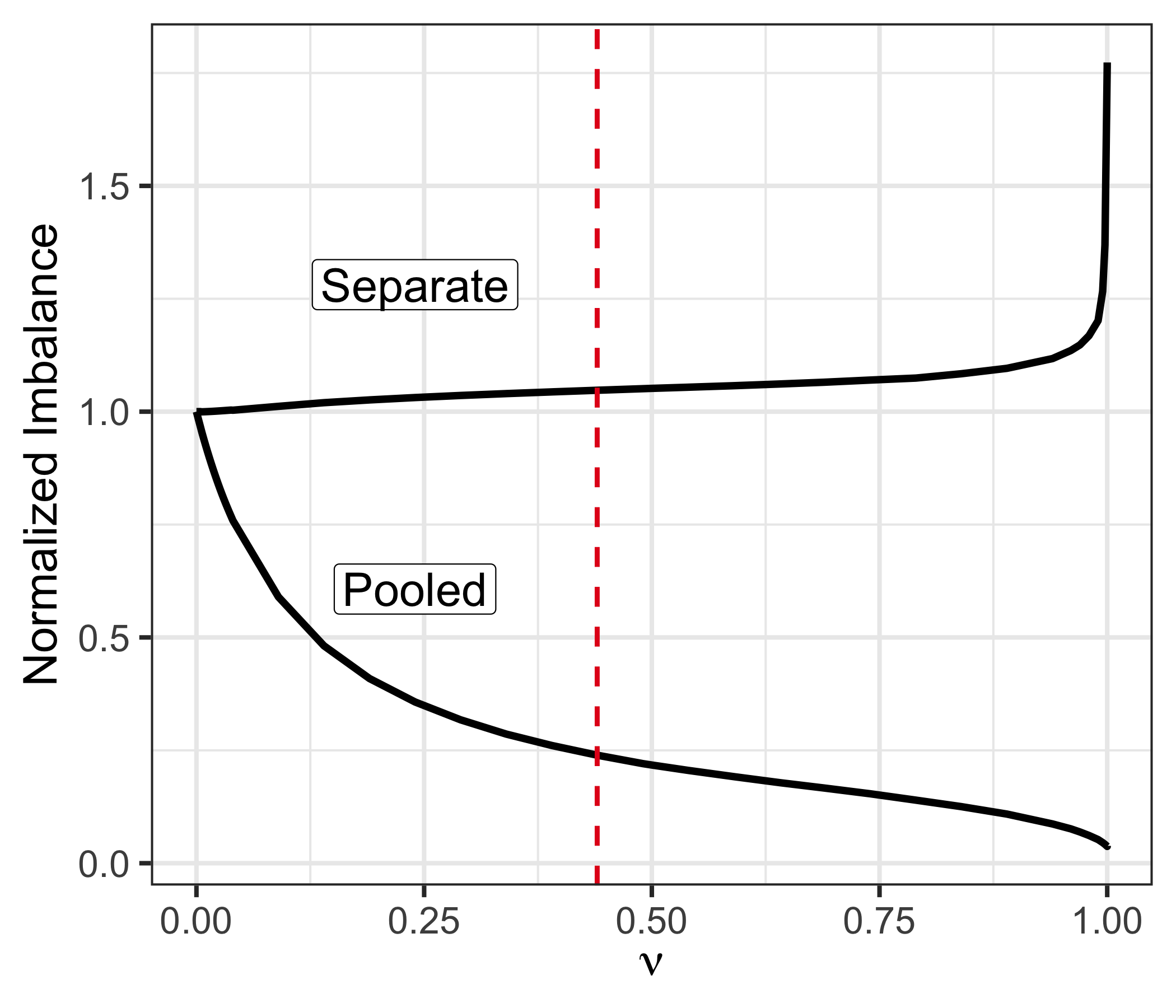} 
  }
  \caption{Separate and pooled balance versus $\nu$.}
    \label{fig:imbal_v_nu}
    \end{subfigure}

    \caption{
    (a) The trade-off between pooled imbalance ($q^\pool$) and unit-specific imbalance ($q^\sep$) as $\nu$ varies, where $\nu = 0$ is the separate SCM solution and $\nu = 1$ is the pooled SCM solution. 
    (b) $q^\sep$ and $q^\pool$ versus $\nu$, each normalized by their values for separate SCM. The dashed red line indicates $\hat{\nu}$.
    The large distance in unit-level imbalance between $\nu = 0.99$ and $\nu = 1$ suggest meaningful gains in balance from deviating from the complete pooling estimate even by a small amount.}
    
    \label{fig:imbalance_plot}
  \end{figure}

As we have seen, it is important to control both the pooled fit (for the ATT) and the unit-level fits (for both the ATT and the unit-level estimates). The hyper-parameter $\nu$ controls the relative weight of these in the objective.

One approach to choosing $\nu$ is to return to the error bounds in Theorems \ref{thm:time_ar_error} and \ref{thm:lfm_error}. The optimization problem in Equation \eqref{eq:stag_avg_relative_scm_primal} can be seen as a first-order approximation to the squares of the error bounds. Therefore, if the parameters of those bounds are known --- and our only goal is to estimate the ATT --- we can use these to choose an appropriate $\nu$.\footnote{
For example, in the autoregressive model, letting $a = \left\| \bar{\rho}\right\|_2 q^\pool(\widehat{\Gamma}^\sep)$ and $b = S_{\rho}q^\sep(\widehat{\Gamma}^\sep)$, we could choose  $\nu=\frac{a^2}{a^2 + b^2}$, with comparable quantities for the linear factor model.} Unfortunately, these will generally be infeasible as the analyst will not know these parameters, though in some applications it may be possible to obtain pilot estimates.

In general we want to find good estimates of both the overall ATT and the unit-level effects. It is therefore important to understand the implications of the choice of $\nu$ for the imbalance criteria. 
Figure \ref{fig:imbalance_plot} provides two views of this for the teacher collective bargaining application.
Figure~\ref{fig:pareto_curve} shows the \emph{balance possibility frontier}: the $y$-axis shows the pooled imbalance $q^\pool$ and 
the $x$-axis shows the unit-level imbalance 
$q^\sep$, and the curve traces out how these change as we vary $\nu$ from the separate SCM solution at the upper left to the pooled solution at the lower right. The relationship is strongly convex, indicating that by accepting a very small increase in pooled imbalance from the fully pooled solution we can obtain large reductions in unit-level imbalance, and vice versa starting from the separate $\nu=0$ solution.
See \citet{king2017balance} and \citet{pimentel2019} for other examples of balance frontiers in observational settings.

Figure \ref{fig:imbal_v_nu} plots the two imbalances, here normalized as $\tilde{q}^{\pool}$ and $\tilde{q}^{\sep}$, to put them on comparable scales, against $\nu$.
As $\nu$ rises, pooled imbalance falls while unit-level imbalance rises, though this is highly nonlinear, as the convex frontier in Figure \ref{fig:pareto_curve} suggests.
Moving from the separate SCM estimate of $\nu=0$ to a partially pooled SCM estimate of $\nu=0.5$ reduces the pooled imbalance by 80 percent, with more modest further reductions as $\nu \to 1$. 
Meanwhile, the unit-level imbalance declines quickly as $\nu$ falls from 1 to 0.9, then more slowly as $\nu$ declines further.
Even a very small deviation from the pooled SCM solution, such as moving from $\nu = 1$ to $\nu = 0.99$, cuts the unit-level imbalance by 30 percent with essentially no change in the pooled fit. 
Due to the number of degrees of freedom involved, the pooled imbalance will often be near zero for $\nu = 1$, and the objective function $q^{\rm{pool}}$ will be relatively flat in the neighborhood of the pooled solution. Therefore we expect that in many cases it will be possible to trade off a small increase in pooled imbalance for a large decrease in the unit-level imbalance, 
yielding a better estimator of both the overall ATT and the unit-level estimates at relatively little cost.
We view the balance possibility frontier plot in Figure \ref{fig:pareto_curve} as an important tool for using partially-pooled SCM in practice. By tracing out the curve, practitioners can see the trade-offs between the pooled and unit-level fit, and choose $\nu$ according to the trade-off they desire.

In our application, we use a simple heuristic to set $\nu$ based on the pooled fit of separate SCM, $q^{\pool}(\hat{\Gamma}^{\sep})$, which we also use to normalize our objective function in Equation  \eqref{eq:stag_avg_relative_scm_primal}.
We set $\nu$ to be the ratio of the pooled fit to the average unit-level fit:
$ \hat{\nu} = \sqrt{ L} \; q^\pool(\widehat{\Gamma}^\sep)/ \frac{1}{J}\sum_{j=1}^J \sqrt{L_j} \; q_j(\hat{\gamma}_j^\sep)$. This is bounded above by 1 due to the triangle inequality.\footnote{If the SCM fits with $\nu=0$ are perfect for each unit, $\frac{1}{J}\sum_{j=1}^J \sqrt{L_j} \; q_j = 0$, then the overall fit will also be perfect, $\sqrt{ L} \; q^\pool= 0$, and our heuristic sets $\hat{\nu} = 0$. This is not a common situation.}
The key idea is that, if the separate SCM problem with $\nu = 0$ achieves good \emph{pooled} fit on its own, then we want to select a small $\nu$, 
which will ensure both good unit-specific and pooled fit. 
Conversely, if the pooled fit of separate SCM is poor, 
then there can be substantial gains to giving $q^\pool$ higher priority by setting $\nu$ to be large.
In Section \ref{sec:sim_study_main} we find through simulation that this heuristic results in weights that significantly reduce both the estimation error for the ATT relative to separate SCM and the estimation error of the unit-level effects relative to pooled SCM.

In the teacher bargaining example, our heuristic yields $\hat{\nu} \approx 0.44$ for the per-pupil expenditure outcome, and we label this point in Figure \ref{fig:pareto_curve}.
The heuristic choice has similar global pre-treatment imbalance to the fully pooled estimator, $\nu=1$, with 
only a modest increase in unit-level imbalance relative to the separate SCM estimate, $\nu = 0$. 
This is reflected in Figure \ref{fig:scm_plots}, which also shows the placebo ATT estimates for partially pooled SCM. While the imbalance for the ATT is slightly larger than for pooled SCM, it is substantially better than for separate SCM.

There are many other potential choices for $\nu$, and, even if we focus solely on the ATT, this one is unlikely to be optimal. 
An alternative strategy when the balance possibility frontier exhibits a strong ``kink'' shape is to choose $\nu$ to be the point after which small improvements to the pooled fit lead to substantially worse unit-level fits. 
Another heuristic is to choose $\nu$ to be the point where the tangent of the frontier is equal to the slope between the end points at $\nu = 0$ and $\nu = 1$ ($\nu = .84$ in the teacher bargaining application).

In the end, the nonlinear relationship between $\nu$ and $\{q^{\sep},q^{\pool}\}$ in Figure \ref{fig:imbal_v_nu} suggests that the loss from choosing a suboptimal $\nu$ is likely to be small, so long as we do not choose something too close to 0 or 1.
We also recommend inspecting the sensitivity of estimates to the particular choice of $\nu$ in practice; we do this in Section \ref{sec:application}.

\section{Extensions}
\label{sec:extensions}

We now add two elaborations to the basic setup. First, we incorporate an intercept shift into the SCM problem, following proposals by  \citet{Doudchenko2017} and  \citet{ferman2018revisiting}. Second, we incorporate auxiliary covariates alongside lagged outcomes.
We conclude by briefly addressing inference in this setting.

\subsection{Incorporating intercept shifts}
\label{sec:intercept_shift}

We have established that the partially pooled SCM estimator achieves nearly as good overall balance as the fully pooled estimator, while achieving much better balance for each unit. Nevertheless, unit-level balance is often imperfect. Particularly when the scale of the outcome varies across units, it can be difficult to construct an adequate synthetic control, as one needs to match both the overall level and patterns over time.
Several recent papers have proposed modifying SCM for a single treated unit by allowing for an \emph{intercept shift} between the treated unit and its synthetic control \citep{Doudchenko2017, ferman2018revisiting, abadie2019synthreview}. We can adapt this approach to the staggered adoption setting by 
including an additional parameter vector $\alpha \in \R^J$, where $\alpha_j$ is an intercept term for unit $j$. We include this intercept in the counterfactual estimate as
\[
\hat{Y}_{jt}(\infty)=\alpha_j + \sum_{i=1}^N \gamma_{ij} (Y_{it}-\alpha_i)
\]
and in the separate and pooled imbalance measures as
\[
  (q^\sep(\alpha, \Gamma))^2 = \frac{1}{2J} \sum_{j = 1}^{J} \left[ \frac{1}{L_j}\sum_{\ell = 1}^{L_j} \left(Y_{j, T_j-\ell} \;-\; \alpha_j - \sum_{i=1}^N \gamma_{ij} Y_{iT_j-\ell}\right)^2\right],
\]
and
\[
  (q^\pool(\alpha, \Gamma))^2 = \frac{1}{L} \sum_{\ell=1}^{L} \left[\frac{1}{J}\sum_{T_j > \ell} 
   \left(Y_{jT_j-\ell} \;-\; \alpha_j - \sum_{i=1}^N \gamma_{ij} Y_{iT_j-\ell}\right)\right]^2.
\]
Again we can define normalized versions of these objectives, $\tilde{q}^\pool(\alpha, \Gamma) \equiv \nicefrac{q^\pool(\alpha, \Gamma)}{q^\pool(\hat{\alpha}^\sep, \widehat{\Gamma}^\sep)}$, where $\hat{\alpha}^\sep$ and $\widehat{\Gamma}^\sep$ are the minimizers of $(q^\sep(\alpha, \Gamma))^2$.
As above, we then form an overall objective function as a convex combination of the normalized squares:

\begin{equation}
  \label{eq:stag_avg_relative_scm_primal_intercept}
  \begin{aligned}
   \min_{\alpha \in \R^J, \Gamma\in \Delta^{\text{scm}}}  \;\;\;   &  \nu~(\tilde{q}^{\text{pool}}(\alpha, \Gamma))^2 +
    (1-\nu) 
      ~(\tilde{q}^{\rm{sep}}(\alpha, \Gamma))^2
    \;+\; 
    \lambda \|\Gamma\|_F^2.
  \end{aligned}
\end{equation}

\noindent The intercept $\hat{\alpha}$ that solves  Equation \eqref{eq:stag_avg_relative_scm_primal_intercept} has a closed form in terms of the solution for the weights, $\hat{\Gamma}^\ast$; $\hat{\alpha}_j$ is the average pre-treatment difference between treated unit $j$ and its synthetic control,
\begin{equation}
  \label{eq:intercept_sol}
  \hat{\alpha}_j = \frac{1}{L_j} \sum_{\ell = 1}^{L_j} Y_{jT_j - \ell} - \frac{1}{L_j}\sum_{i=1}^N\sum_{\ell = 1}^{L_j} \hat{\gamma}^\ast_{ij} Y_{jT_j - \ell}.
\end{equation}
Plugging this value of $\hat{\alpha}$ into Equation \eqref{eq:stag_avg_relative_scm_primal_intercept}, we see that this procedure is equivalent to solving the partially-pooled SCM problem \eqref{eq:stag_avg_relative_scm_primal} using the \emph{residuals} $\dot{Y}_{iT_j - \ell} \equiv Y_{iT_j-\ell} - \frac{1}{L_j}\sum_{\ell=1}^{L_j} Y_{iT_j - \ell}$. The resulting treatment effect estimates have a particularly useful form:
\begin{equation}
  \label{eq:tau_jt_aug}
    \hat{\tau}_{jk}^{\ast}  = \frac{1}{L_j}\sum_{\ell=1}^{L_j}\left[ \left(Y_{j T_j+k} - Y_{j T_j - \ell}\right) -
                \sum_{i=1}^N\hat{\gamma}_{ij}^\ast\left(Y_{i T_j + k} - Y_{i T_j - \ell}\right)\right],
\end{equation}
and 
\begin{equation}
  \label{eq:tau_att_aug}
    \widehat{\text{ATT}}_k^{\ast} = \frac{1}{J} \hat{\tau}_{jk}^\ast = 
     \frac{1}{J} 
     \sum_{j=1}^J \left[
     \frac{1}{L_j}
     \sum_{\ell=1}^{L_j}
     \left[\left(Y_{jT_j+k} - Y_{jT_j-\ell}
             \right)
           -\sum_{i=1}^N 
               \hat{\gamma}_{ij}^\ast
               \left(Y_{iT_j+k} -
                     Y_{iT_j-\ell}
             \right)      \right]\right].
\end{equation}
We can view this as a weighted difference-in-differences (DiD) estimator. In the special case with uniform weights over units, $\hat{\gamma}_{ij}^\ast = 1/\|\calD_j\|$, Equation \eqref{eq:tau_jt_aug} is the simple average over all two-period, two-group DiD estimates, averaging over all pre-treatment lags $\ell$ and donor units $i$. 
This is equivalent to recent proposals for DiD estimators that allow for treatment effect heterogeneity with a fixed donor set per treatment time cohort \citep[see][among others]{abraham2018estimating, Callaway2018}.
With non-uniform weights, $\hat{\tau}_{jk}^{\ast}$  compares the change in outcomes for treated unit $j$ to the change for the synthetic control, rather than the average change across all potential donors.
Equation \eqref{eq:tau_att_aug} averages these estimates across treated units $j$ to form $\widehat{\text{ATT}}_k^{\ast}$.

\begin{figure}[tb]
  \centering
  \begin{subfigure}[t]{0.45\textwidth}  
    {\centering \includegraphics[width=\maxwidth]{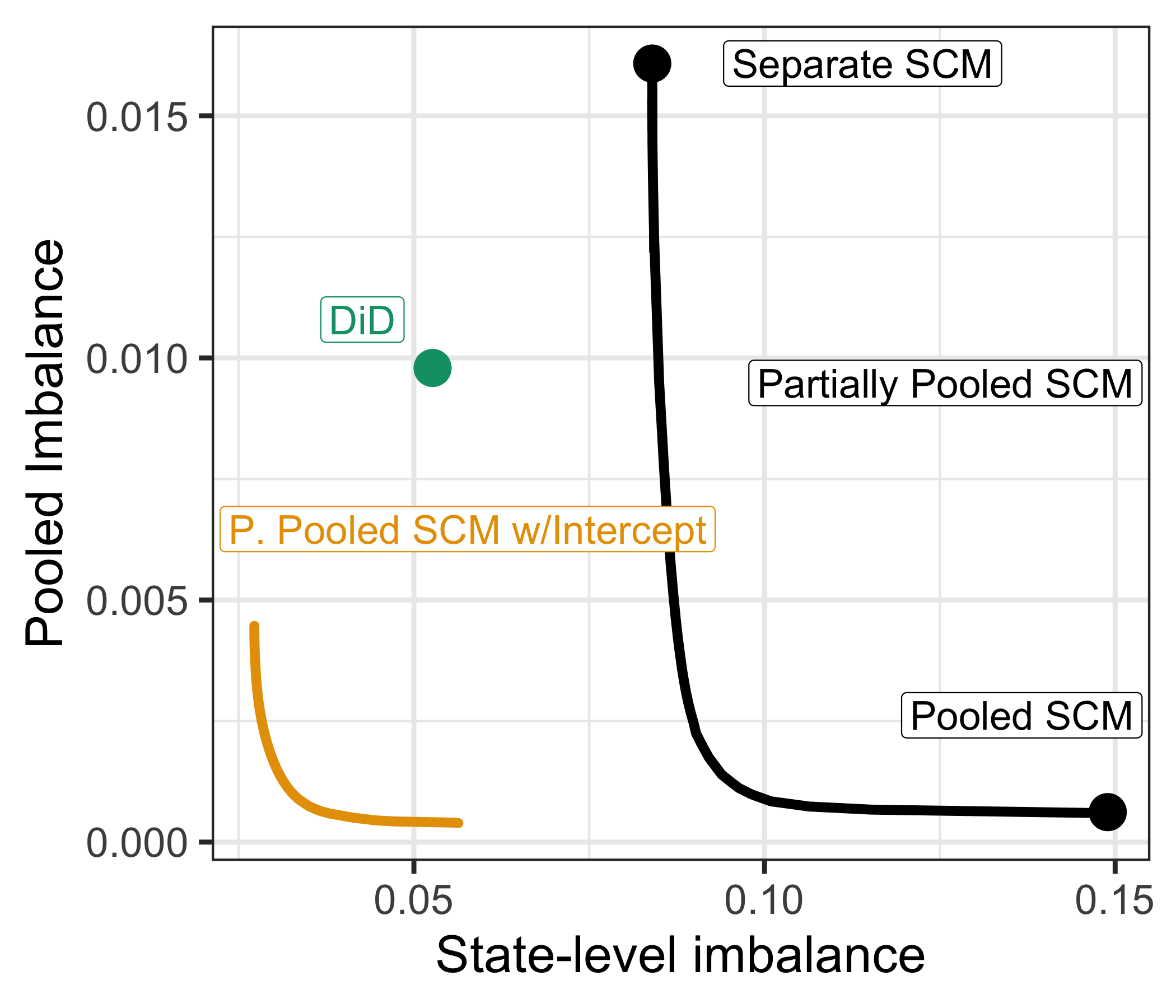} 
    }
    \caption{The balance possibility frontier for SCM with and without an intercept.}
      \label{fig:pareto_comparison}
    \end{subfigure}
    \begin{subfigure}[t]{0.45\textwidth}  
      {\centering \includegraphics[width=\maxwidth]{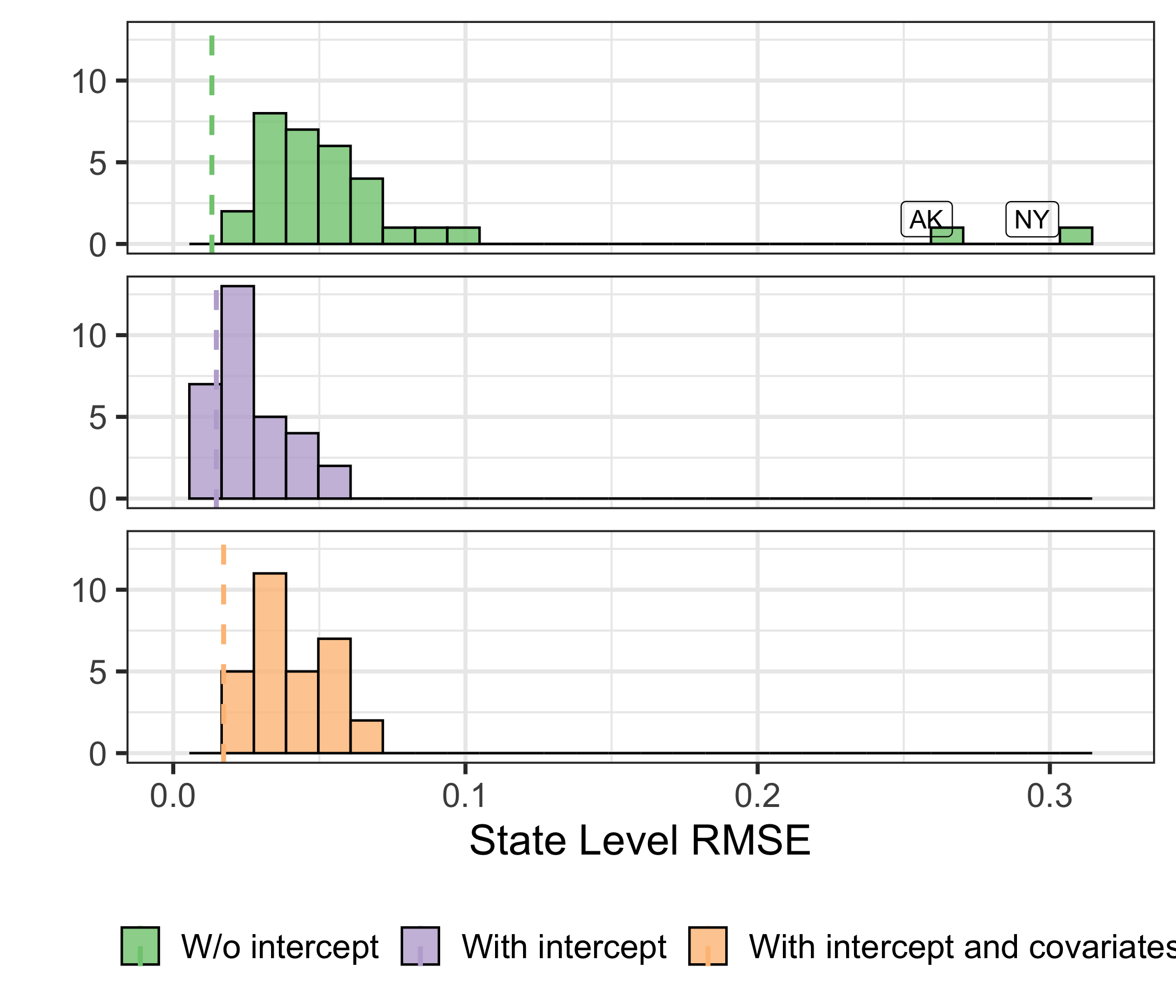} 
      }
      \caption{Distribution of unit-level fits}
        \label{fig:state_fits}
        \end{subfigure}
  
    \caption{(a) The balance possibility frontier for SCM with and without an intercept, as well as the implied imbalance for DiD. Incorporating unit-level fixed effects leads to substantial improvements in balance. For DiD, we compute the implied balance as  $\sqrt{\sum_{\ell=1}^{L} \left(\widehat{\text{ATT}}^\ast_{-\ell}\right)^2}$, the RMSE of the placebo estimates, from Equation \eqref{eq:tau_jt_aug} with uniform weights. 
    (b) The distribution of state-level fits (in terms of RMSE) with and without an intercept and covariates; dashed lines show the pooled pre-treatment RMSE.
    }
    \label{fig:wevent_ppexp}
  \end{figure}

Figure \ref{fig:wevent_ppexp} shows the value of including an intercept to improving  pre-treatment fit in the teacher collective bargaining application.
Figure \ref{fig:pareto_comparison} presents this as a balance possibility frontier for SCM with the weights alone and with the intercept, as well as the implied imbalance for the DiD estimator alone.
Here, simple unweighted DiD achieves unit-level and pooled balance that improves on the no-intercept SCM possibility frontier. 
However, the intercept-shifted estimator dominates both DiD and no-intercept SCM estimates on both criteria, for all but the largest $\nu$.
We see similar results when examining the state-specific fits. 
Figure \ref{fig:state_fits} shows the unit-level fit for both partially pooled SCM and the intercept-augmented version. Two states, New York and Alaska, have especially bad pre-treatment fits without including an intercept because they have the highest per-pupil expenditures of all the states for many years (see Appendix Figure \ref{fig:raw_data_highlight}). Accounting for the pre-treatment average through the intercept dramatically improves the fits for these states.

\subsection{Incorporating auxiliary covariates}
\label{sec:covs}

We have focused thus far on matching pre-treatment values of the outcome variable. In practice, we typically observe a set of auxiliary covariates $X_i \in \R^d$ as well.
In our collective bargaining application, we consider five covariates, measured as of the start of the sample in 1959-1960: income per capita, the student to teacher ratio, the percent of the population with 12+ and 13+ years of education, and the female labor force participation rate.\footnote{Due to missing data for these auxiliary covariates, we restrict our analysis here to the contiguous United States. Note that this drops Alaska, which we have seen is far outside the convex hull of its donor units.} We standardize each to have mean zero and variance one.

There are several ways to incorporate auxiliary covariates in the setting with a single treated unit. Here we directly include them into the optimization problem. Analogous to above, we define both the unit-level imbalance and pooled imbalance of $X$,
\[
  q^\sep_X(\Gamma) = \sqrt{\frac{1}{J}\sum_{j=1}^J\left\|X_j - \sum_{i=1}^n \gamma_{ij}X_i\right\|_2^2},
\]
and another for the pooled synthetic control,
\[
  q^\pool_X(\Gamma) = \left\|\frac{1}{J} \sum_{j=1}^J X_j - \sum_{i=1}^n \gamma_{ij}X_i\right\|_2,
\]
\noindent with normalized versions $\tilde{q}^{\sep}_X(\Gamma)$ and $\tilde{q}^{\pool}_X(\Gamma)$.\footnote{Specifically, let $\hat{\alpha}^\sep$ and $\widehat{\Gamma}^\sep$ be the minimizers of $(q^\sep(\alpha, \Gamma))^2 + \xi (q^\sep_X(\Gamma))^2$, and $(C^\sep)^2 = (q^\sep(\hat{\alpha}^\sep, \widehat{\Gamma}^\sep))^2 + \xi (q^\sep_X(\widehat{\Gamma}^\sep))^2$ and $(C^\pool)^2 = (q^\pool(\hat{\alpha}^\sep, \widehat{\Gamma}^\sep))^2 + \xi (q^\pool_X(\widehat{\Gamma}^\sep))^2$ be the combined separate and pooled imbalances. We define the normalized objectives as $\tilde{q}^\pool_X(\Gamma) = \nicefrac{q^\pool_X(\Gamma)}{C^\pool}$, $\tilde{q}^\sep_X(\Gamma) = \nicefrac{q^\sep_X(\Gamma)}{C^\sep}$, and slightly abuse notation by re-defining 
$\tilde{q}^\pool(\alpha, \Gamma) \equiv \nicefrac{q^\pool(\alpha, \Gamma)}{C^\pool}$ and
$\tilde{q}^\sep(\alpha, \Gamma) \equiv \nicefrac{q^\sep(\alpha, \Gamma)}{C^\sep}$.} 
We then include these in our objective, with an additional hyper-parameter $\xi$:

\begin{equation}
  \label{eq:primal_with_covs}
  \begin{aligned}
   \min_{\alpha \in \R^J, \Gamma\in \Delta^{\text{scm}}}  \;\;\;   &  \nu~\left((\tilde{q}^{\text{pool}}(\alpha, \Gamma))^2 + \xi (\tilde{q}^\pool_X(\Gamma))^2\right)+
    (1-\nu)
      ~\left((\tilde{q}^{\rm{sep}}(\alpha, \Gamma))^2 + \xi (\tilde{q}_X^\sep(\Gamma))^2\right)
    \;+\; 
    \lambda \|\Gamma\|_F^2.
  \end{aligned}
\end{equation}
While we write this optimization problem with an intercept shift, we could also include auxiliary covariates but no intercept.
The choice of $\xi$ determines the relative importance of the outcomes and the auxiliary covariates. Setting $\xi = 0$ recovers the optimization problem \eqref{eq:stag_avg_relative_scm_primal_intercept} without auxiliary covariates, while in the extreme case setting $\xi = \infty$ will, if feasible, enforce exact balance on the auxiliary covariates. We decide to give equal priority to both terms. Since the auxiliary covariates are standardized, we set $\xi$ to be the sample variance of the pre-$T_J$ outcomes for the never treated units.
This equally weights both components in the objective functions, and reduces the number of hyper-parameters and specification choices.  Finally, we can incorporate time-varying covariates by including the values at time periods before the first treatment time $T_1$ into the vector $X_i$.\footnote{\citet{AbadieAlbertoDiamond2010} suggests using average pre-treatment mean outcomes in $X$, as an alternative to the above intercept proposal. As noted above, this may increase the difficulty of finding adequate synthetic controls.}

  \begin{figure}[tb]
      \centering
      {\centering \includegraphics[width=\maxwidth]{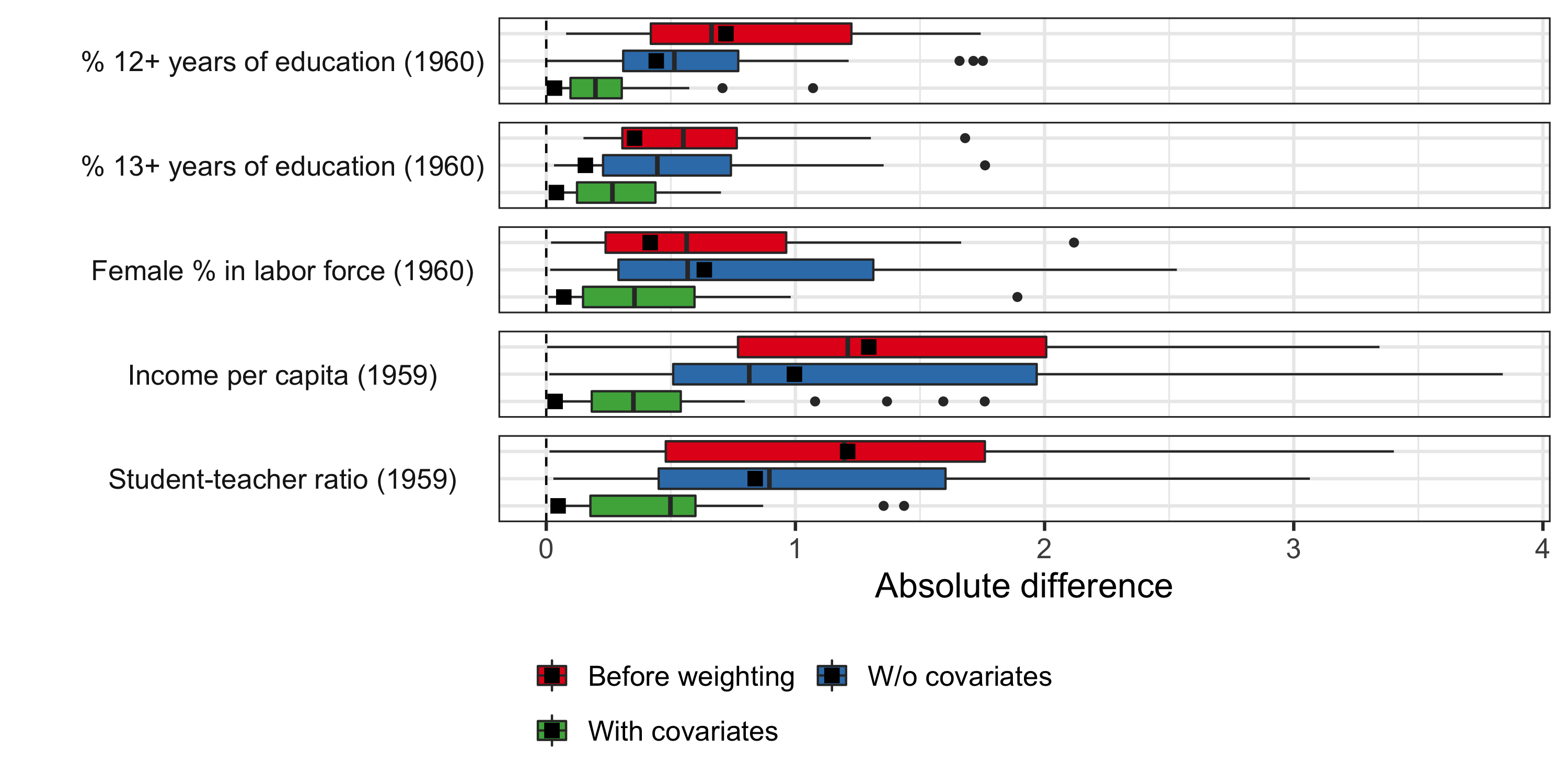} 
      }
      \caption{Distribution of the absolute difference between each treated unit and its synthetic control for the (standardized) auxiliary covariates, before weighting and with/without including covariates in the optimization procedure. Black squares show the absolute average difference.}
        \label{fig:balance_plot}
    \end{figure}

Figure \ref{fig:balance_plot} shows the level of covariate balance between each treated unit and its synthetic control, as well as for the average across treated units. Before weighting there are large differences between the treated units and their donor sets, and weighting on the outcomes alone does little to alleviate these differences. Including the auxiliary covariates into the optimization procedure finds weights that give nearly perfect covariate balance for the pooled synthetic control (indicated as the black squares), while also significantly improving covariate balance for the individual treated units (indicated as boxplots). Figure \ref{fig:state_fits} shows that this improved covariate balance comes at a small cost to the fit on the pre-treatment outcomes: the distribution of unit-level pre-treatment RMSE shifts slightly to the right.

\subsection{Inference}
\label{sec:inference}

There is a growing literature on inference for SCM-type estimators, though no proposed approach is fully satisfactory for all cases.
In settings where multiple units adopt treatment simultaneously, \citet{Abadie_LHour} propose an extension of the original permutation procedure of \citet{AbadieAlbertoDiamond2010}, and \citet{Arkhangelsky2018} propose resampling-based approaches. 
In a staggered adoption setting, \citet{toulis2018testing} propose a weighted permutation approach based on a Cox proportional hazards model. 
This is not appropriate in our application, however, since multiple units have the same treatment time, which is incompatible with the Cox model. Finally, \citet{cao2019synthetic} propose an Andrews test for inference with intercept-shifted SCM under staggered adoption.
Building on the existing literature, we consider 
constructing confidence intervals via the wild bootstrap.
We briefly describe this method here; we address asymptotic Normality and inference via the jackknife in Appendix \ref{sec:asymp_normal}.

The wild bootstrap approach we implement adapts the proposal from \citet{Otsu2017} for bias-corrected matching estimators; see also \citet{Imai2019_match}. 
First, we can re-write $\widehat{\text{ATT}}_k$ as the following average over units:
\begin{equation}
  \label{eq:linear_att}
  \widehat{\text{ATT}}_k = \frac{1}{J}\sum_{i=1}^N\sum_{g=T_1}^{T_J}\left(\bbone_{T_i  = g} - \sum_{T_j = g} \hat{\gamma}_{ij}\right) \left(Y_{i g+k} -\frac{1}{g-1}\sum_{\ell=1}^{g-1}Y_{ig-\ell}\right) = \frac{1}{J}\sum_{i=1}^N\tilde{\tau}_i.
\end{equation}
This bootstrap procedure draws a sequence of random variables $W^{(b)}_1,\ldots,W^{(b)}_N$ independently with $P(W_i = -(\sqrt{5} - 1) / 2 ) = (\sqrt{5} + 1)/2\sqrt{5}$ and $P(W_i = (\sqrt{5} + 1) / 2) = (\sqrt{5} - 1)/2\sqrt{5}$ for $b=1,\ldots,B$, and computes the boostrap statistic:
\begin{equation}
  \label{eq:mult_boot}
  S^{(b)} = \frac{1}{J}\sum_{i=1}^N W_i^{(b)}\left(\tilde{\tau}_i - \widehat{\text{ATT}}_k\right),
\end{equation}
for each draw. Letting $q_{\alpha/2}$ and $q_{1 - \alpha/2}$ denote the $\alpha/2$ and $1 - \alpha/2$ quantiles of $S^{(b)}$, we construct confidence intervals via $[\widehat{\text{ATT}}_k - q_{1 - \alpha/2}, \widehat{\text{ATT}}_k + q_{\alpha/2}]$. Importantly, we keep the weights and outcomes fixed, and only re-sample the multiplier variables $W_i^{(b)}$.

In the next section, we evaluate the coverage of the wild bootstrap with a simulation study that mimics the structure of the collective bargaining application.
In Appendix \ref{sec:asymp_normal}, we take an alternative route and motivate the use of resampling methods via asymptotic Normality. 
In particular, we provide a set of sufficient conditions for $\widehat{\text{ATT}}_k - \text{ATT}_k$ to be asymptotically Normal.
We consider an asymptotic regime in which $J,N_0\to\infty$, with the number of lags $L$ fixed and the number of control units growing faster than the number of treated units $\frac{J}{N_0} \to \infty$. 
We also adapt a generalization of the conditional parallel trends assumption in \citet{abadie2005semiparametric} to the staggered adoption setting. 
However, there are several ways such asymptotic results can be misleading. 
First, our result assumes that the synthetic control weights can achieve perfect fit within treatment time cohorts, which ensures that the distribution of $\widehat{\text{ATT}}_k$ is centered around $\text{ATT}_k$. Poor fit, either overall or across time cohorts, can lead to under-coverage.
Second, the asymptotic approximation can be poor when there are relatively few total units, and the use of resampling methods can exacerbate this. 
Thus, while we show that these approaches yield reasonable results in simulations, we suggest interpreting any confidence intervals for typical applications with caution.


\section{Simulation study}
\label{sec:sim_study_main}

We now consider the performance of different approaches in a simulation study calibrated to the collective bargaining dataset; we turn to the impacts of mandatory teacher collective bargaining laws in the actual data in the next section.
We evaluate performance with three different data generating processes. 
First, we generate never treated outcomes according to a two-way fixed effects model,
\begin{equation}
  \label{eq:sim_model1}
  Y_{it}(\infty) = \text{int} + \text{unit}_i + \text{time}_t + \varepsilon_{it},
\end{equation}
with both unit and time effects are normalized to have mean zero.
This model satisfies the parallel trends assumption needed for the DiD estimator we consider below.
We estimate \eqref{eq:sim_model1} using only the never-treated observations, and extract the estimated variance of the unit effects, $\hat{\Sigma}$, and of the error term, $\hat{\sigma}^2_{\varepsilon}$. We then generate $\text{unit}_i \overset{\text{iid}}{\sim} N(0, \hat{\Sigma})$ and $\varepsilon_{it} \overset{\text{iid}}{\sim} N(0, \hat{\sigma}_\varepsilon^2)$. 

Second, we use a factor model
with a 2-dimensional latent time-varying factor $\mu_t \in \R^2$ and unit-specific coefficients $\phi_i \in \R^2$:
\begin{equation}
  \label{eq:sim_model2}
  Y_{it}(\infty) = \text{int} + \text{unit}_i + \text{time}_t + \phi_{i}'\mu_{t} + \varepsilon_{it}.
\end{equation}
We estimate \eqref{eq:sim_model2} using the \texttt{R} package \texttt{gsynth} \citep{Xu2017} for the untreated units and time periods, then estimate the variance-covariance matrix of the unit fixed effects and factor loadings, $\hat{\Sigma}$, and the variance of the error term $\hat{\sigma}^2_\varepsilon$. 
Here we use the estimated $\{\widehat{\text{time}}_t, \hat{\mu}_t\}$, and draw $\{\text{unit}_i, \phi_i\} \overset{\text{iid}}{\sim} \text{MVN}(0, \hat{\Sigma})$ and $\varepsilon_{it} \overset{\text{iid}}{\sim} N(0, \hat{\sigma}_\varepsilon^2)$.

Finally, we have a random effects autoregressive model:
\begin{equation}
  \label{eq:sim_model_ar}
  \begin{aligned}
    Y_{it}(\infty) & = \sum_{\ell=1}^3 \rho_\ell Y_{i t - \ell}(\infty) + \varepsilon_{it}, \qquad \rho  \sim N(\mu_\rho, \sigma_\rho^2),
  \end{aligned}
\end{equation}
that we fit using \texttt{lme4} \citep{Bates2015} to obtain estimates $\hat{\mu}_\rho$ and $\hat{\sigma}_\rho$. In order to increase the level of heterogeneity across time, we simulate from this hierarchical model with 8 times the standard deviation $8\hat{\sigma}_\rho$. For all three outcome processes we generate simulated data sets with the same dimensions as the data, $N = 49$ and $T = 39$, and impose a sharp null of no treatment effect, $Y_{it}(s)=Y_{it}(\infty)=Y_{it}$.

A key component of the simulation model is selection into treatment. We fix the treatment times to be the same as in the teacher unionization application. 
For each treatment time, we assign treatment to those units not already treated with probability $\pi_i$, sweeping through the fixed set of treatment times.
For the two-way fixed effects model, we set the probability that unit $i$ is treated at each treatment time to be 
$\pi_i = \text{logit}(\theta_0 + \theta_1 \cdot \text{unit}_i)$, with $\theta_0 = -2.7$ and $\theta_1=-1$, yielding around 30 units that are eventually treated in each simulation draw. For the factor model we choose $\pi_i=\text{logit}(\theta_0 + \theta_1(\text{unit}_i+\phi_{i1} + \phi_{i2}))$, and set $\theta_0 = -2.7$ and $\theta_1=-1$ so that around 32 units are eventually treated in each simulation draw, following the distribution of the data. For the autoregressive process we allow selection to depend on the three lagged outcomes $\pi_i = \text{logit}\left(\theta_0 + \theta_1\sum_{\ell=1}^3Y_{i, t - \ell}\right)$, where $\theta_0 = \log 0.04$ and $\theta_1 = -2$.

\begin{figure}
  \centering
  {\centering \includegraphics[width=\textwidth]{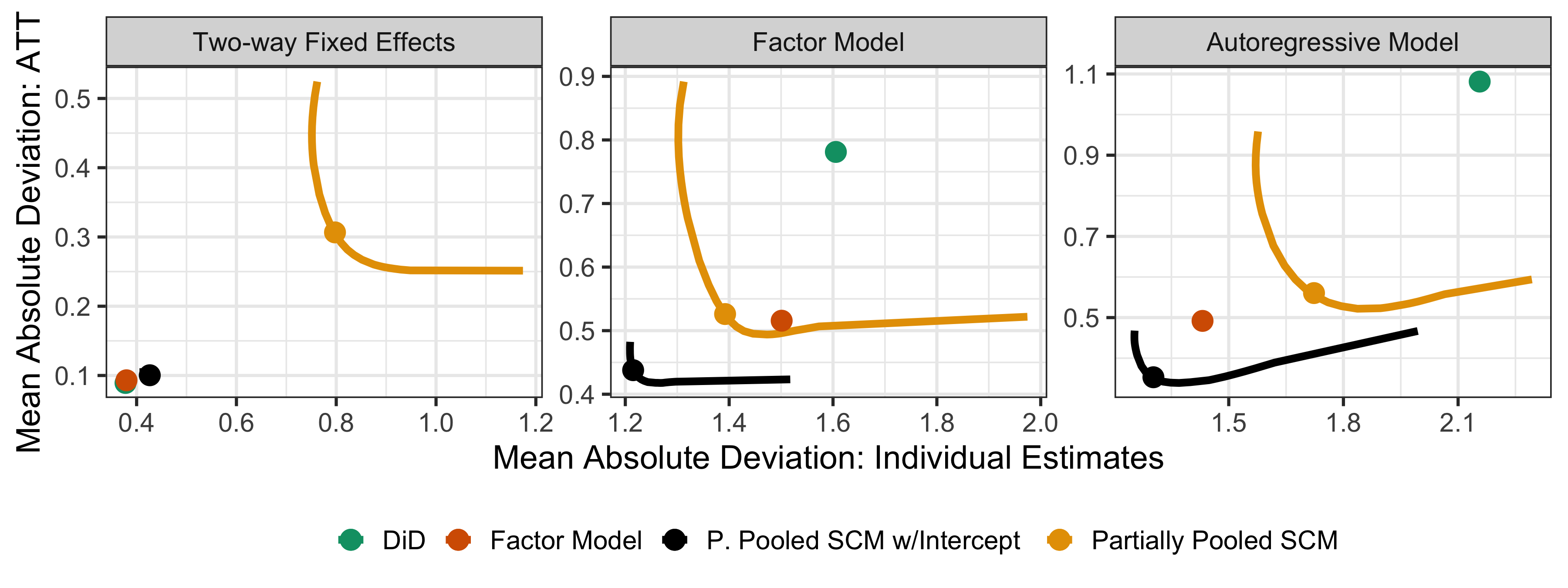}
  }
  \caption{\label{fig:sim_study_main}Monte Carlo estimates of the MAD for the overall ATT vs the MAD the individual ATT estimates. The lines trace out values for $\nu \in [0,1]$, the solid points are the average value using the heuristic $\hat{\nu}$. In the two-way fixed effects and factor model simulations, the estimated factor model is the oracle estimator. Among the alternatives, the intercept-shifted partially pooled SCM has lowest MAD for both the overall ATT and the individual ATT estimates.} 
  \end{figure}       

\paragraph{Estimation.} We consider several estimators for the average post-treatment effect $\text{ATT}$.
Figure \ref{fig:sim_study_main} shows four: (1) A difference-in-differences estimator following Equation \eqref{eq:tau_jt_aug} with uniform weights, (2) the partially pooled SCM estimator, as we vary $\nu$ between 0 and 1, (3) partially pooled SCM with an intercept, again varying $\nu$, and (4) directly estimating the factor model.
Solid points indicate the heuristic choice of $\hat{\nu}$ above.
The vertical axis of each panel shows the Mean Absolute Deviation (MAD) for the ATT, $\E\left[\left|\text{ATT} - \widehat{\text{ATT}}\right|\right]$, while the horizontal axis shows the average of the individual post-treatment effect estimates, $\E\left[\frac{1}{J}\sum_{j=1}^J|\tau_{j}-\hat{\tau}_{j}|\right]$. Appendix Figures \ref{fig:sim_study_bias} and \ref{fig:sim_study_rmse} show the analogous results for the bias and Root Mean Square Error (RMSE). 

There are several key takeaways from Figure \ref{fig:sim_study_main}. First, under each data generating process there is a tradeoff between estimating the ATT and the individual effects, with $\nu = 1$ at the top left of the ``MAD frontier'' and $\nu = 0$ at the bottom right.
Partially pooled SCM significantly reduces the bias for the overall ATT relative to separate SCM, and a small amount of pooling also leads to slightly better individual ATT estimates.
The gains to pooling, however, diminish for $\nu$ close to 1,
with the fully pooled SCM yielding poor individual ATT estimates under all three models.
Under a two-way fixed-effects model there is no penalty to pooling in terms of MAD for the overall ATT. This comports with Theorem \ref{thm:lfm_error}, which shows that targeting the pooled pre-treatment fit is sufficient under a two-way fixed effects model. However, under the factor model and AR process the fully pooled estimator leads to worse MAD for the overall ATT estimates than partially pooled SCM.
Second, when mis-specified, the DiD estimator does not do particularly well at controlling the MAD for either overall ATT or the unit-level estimates. 
Third, the intercept-shifted estimator dominates either of the alternatives in terms of both overall and unit-level estimates. 
Here again there are gains to partially pooling SCM, albeit with the possibility for a large amount of error from over-pooling. Fourth, our heuristic choices of $\nu$ perform reasonably well at selecting a point close to the value that minimizes the MAD for the ATT, while also reducing the MAD for the individual estimates. 
Finally, the partially-pooled SCM estimator with an intercept shift performs as well as or better than fitting the factor model directly.

\paragraph{Inference.}
We conclude by examining the finite-sample coverage of approximate 95\% confidence intervals from the wild bootstrap.
Figure \ref{fig:coverage} shows the coverage of approximate confidence intervals for partially pooled SCM with an intercept shift, using
the wild bootstrap to construct the intervals. Under the two-way fixed effects model, in which there is no bias from inexact fit, the wild bootstrap has close to 95\% coverage. 
Under both the linear factor model and the autoregressive model, however, the wild bootstrap is somewhat conservative.\footnote{Appendix Figure \ref{fig:coverage_scm} shows the analogous results for partially-pooled SCM without including an intercept. In this case, the wild bootstrap is extremely conservative.}
Overall, the wild bootstrap appears to be a reasonable, if conservative, choice.

\begin{figure}
  \centering
  {\centering \includegraphics[width=\textwidth]{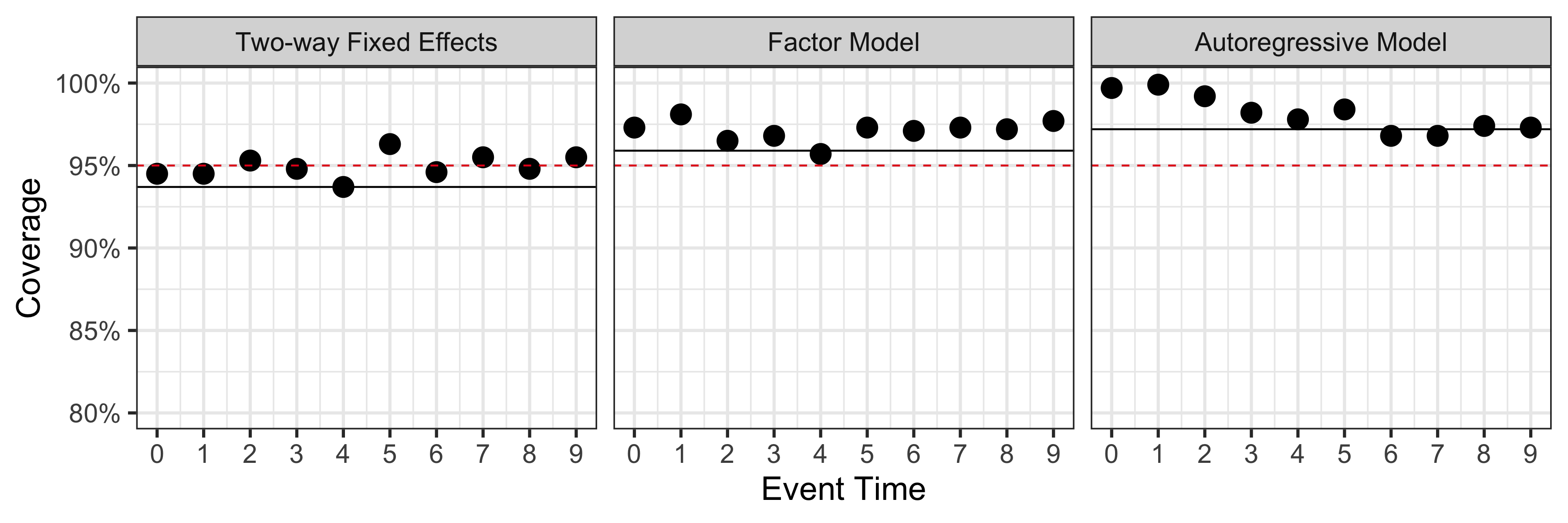}
  }
  \caption{\label{fig:coverage} Monte Carlo estimates of the coverage of approximate 95\% confidence intervals $k=0,\ldots,9$ periods after treatment. The solid line indicates the coverage for the overall ATT estimate averaged across all post-treatment periods.} 
  \end{figure}

\section{Impacts of mandatory teacher collective bargaining laws}
\label{sec:application}

We now return to measuring the impact of mandatory teacher collective bargaining.
The left of Figure \ref{fig:cov_int_results_ppexp} shows the placebo estimates from Equation \eqref{eq:tau_jt_aug}, where $k < 0$.\footnote{These placebo checks differ from those typically performed in traditional event studies, which test for the parallel trends assumption by comparing pre-treatment outcomes between treated and control units. 
These tests generally have low power, however; see, e.g., \citet{roth2018did,bilinski2018seeking, kahn2019promise}.
In contrast, the intercept-shifted estimator uses pre-treatment outcomes to select donor units that best balance the treated units, in effect optimizing for the placebo test. It is still possible to inspect pre-treatment fit, as in standard SCM, but this is best seen as an assessment of the quality of the match rather than as a formal placebo test.}
We see that along with the good unit-specific fits shown in Figure \ref{fig:state_fits} and the good covariate balance shown in Figure \ref{fig:balance_plot}, the pooled synthetic control estimate is near zero for $k < 0$.
The right side of the figure shows the estimated impact on per-pupil current expenditures, with approximate 95\% confidence intervals computed via the wild bootstrap.

Consistent with \citet{paglayan2019public}, we find weakly negative effects of mandatory teacher collective bargaining laws on student expenditures.
Pooled across the eleven years after treatment adoption, the overall estimate is $\widehat{\text{ATT}} = -0.03$, or a 3 percent decrease in per-pupil expenditures, with an approximate 95\% confidence interval of $[-0.06, +0.005]$.
In Appendix Figure \ref{fig:cov_int_ppexp_unit_level} we show the average post-treatment effect for each state and the unit-level fits. For those states with good pre-treatment fit, we find small positive and negative effects, while we estimate larger negative effects for those with worse fit.
These estimates are in stark contrast to the results from \citet{hoxby1996teachers}, who argues for a 12 percent positive effect, although she gives a range of estimates. 
One possible explanation for this is that school districts are able to divert funds from other purposes to fund higher teacher salaries with minimal net effect on total expenditures. In Appendix Figure \ref{fig:results_teachsal} we show estimates of the effect on teacher salaries, finding evidence against a positive effect.

 \begin{figure}[tb]
  \centering
      \begin{subfigure}[t]{0.45\textwidth}  
        {\centering \includegraphics[width=\maxwidth]{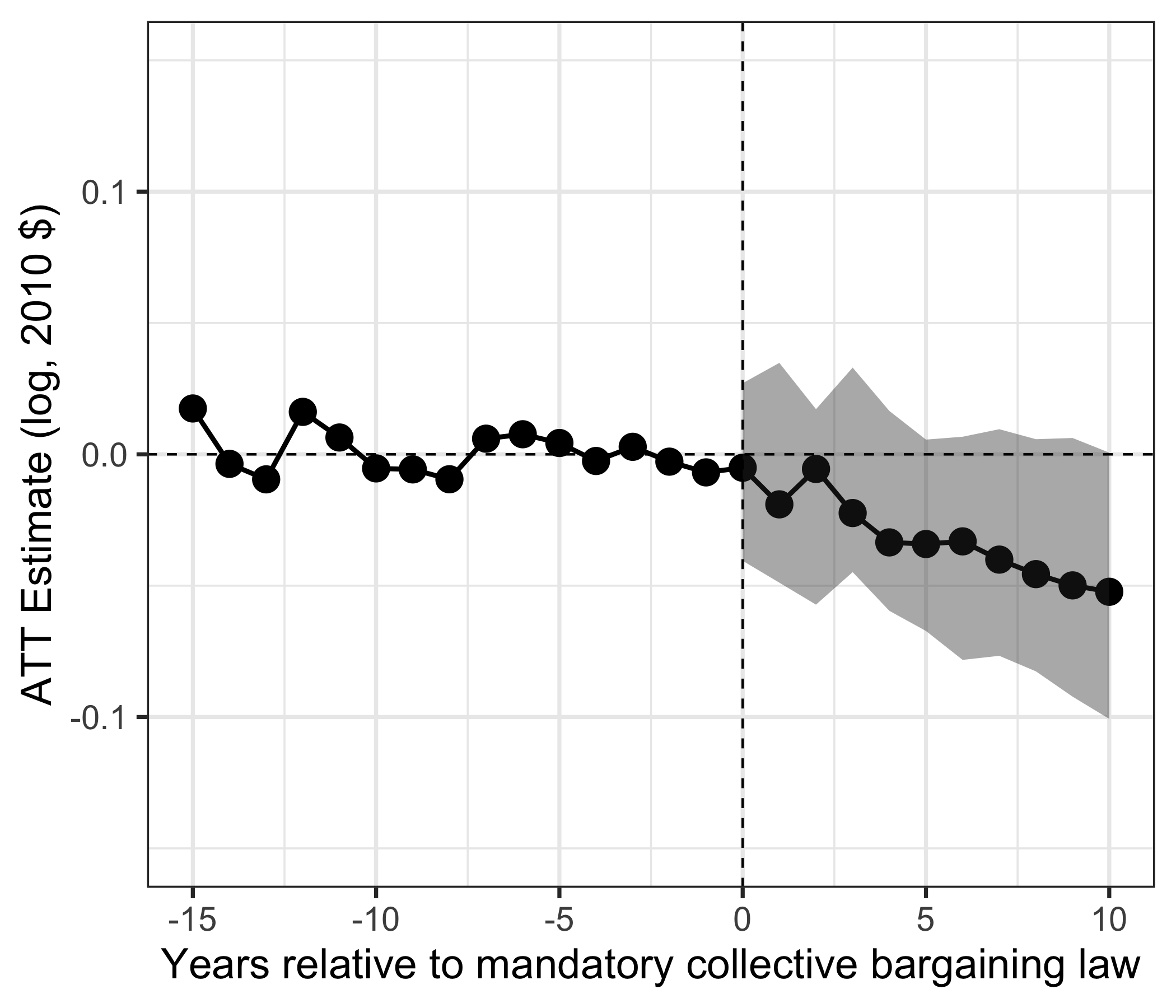} 
        }
        \caption{Effect of mandatory collective bargaining on per-pupil expendisures ($\hat{\nu} = 0.22$)}
          \label{fig:cov_int_results_ppexp}
    \end{subfigure} \quad
    \begin{subfigure}[t]{0.45\textwidth}  
      {\centering \includegraphics[width=\maxwidth]{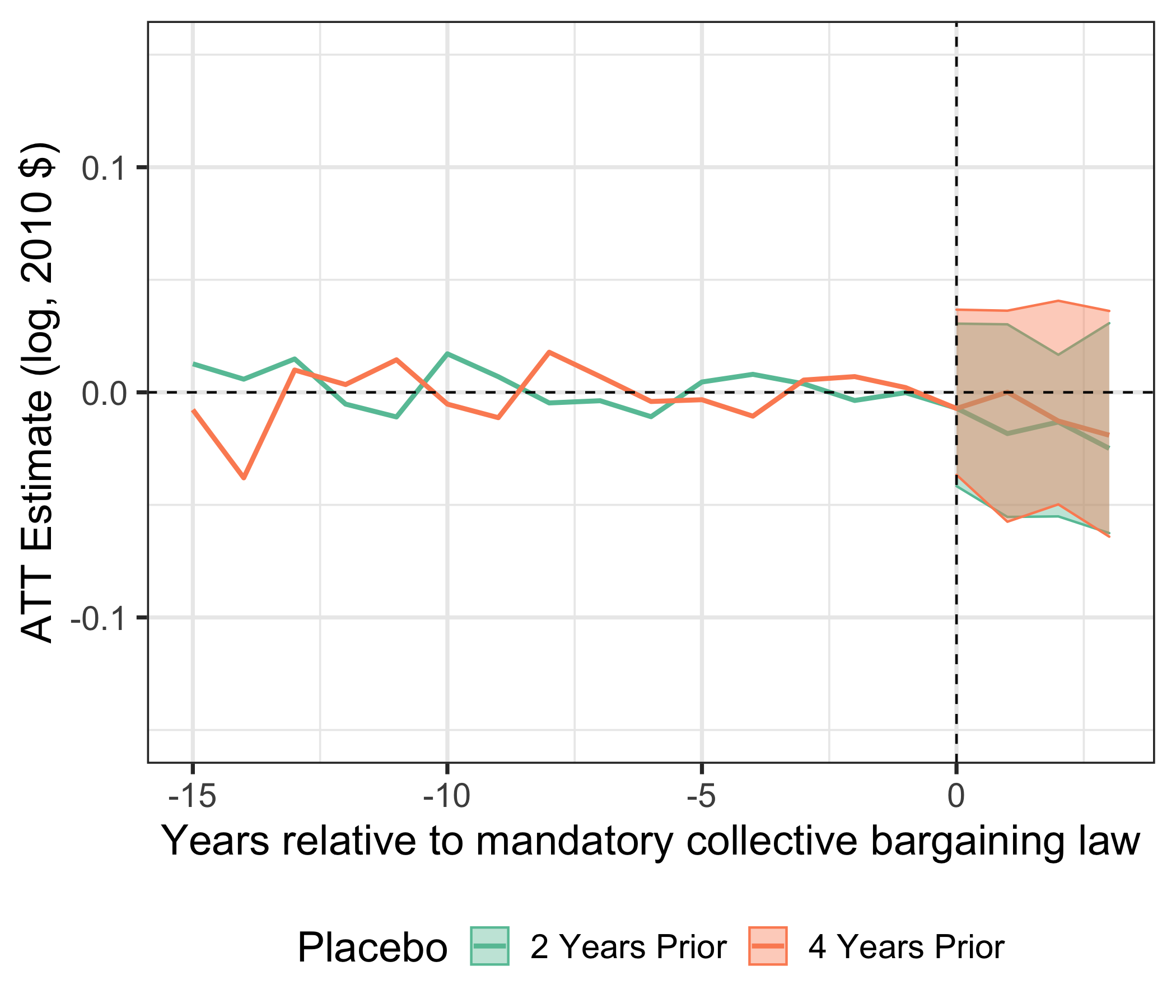}
      }
      \caption{Placebo estimates}
        \label{fig:time_placebos}
        \end{subfigure} 
    \caption{Estimates of the ATT on  per-pupil current expenditures (log, 2010 \$) and placebo estimates re-indexing treatment time to two and four years before the true treatment time. The placebo effects are very close to zero and are indistinguishable from zero at this level of precision.}
    \label{fig:cov_plot}
  \end{figure}

We can assess the strength of evidence by conducting robustness and placebo checks. First, following \citet{Abadie2015}, we begin by assessing out-of-sample validity via \emph{in time placebo checks}. 
These checks hold out some pre-treatment time periods by re-indexing treatment time to be earlier (i.e. setting $T^\prime_j = T_j - x$ for some $x$), then estimate placebo effects for the held-out pre-intervention time periods. Figure \ref{fig:time_placebos} shows the placebo estimates for the intercept-shifted partially pooled SCM estimator with covariates using a placebo treatment time two and four periods before the true treatment time. Both estimators achieve excellent pre-treatment fit and estimate placebo effects that are indistinguishable from zero.

Another important check that we recommend in practice is to gauge the sensitivity of the ATT estimates to the particular choice of pooling parameter $\nu$.
Figure \ref{fig:nu_sensitivity} shows the overall ATT estimates varying $\nu$ from separate SCM $\nu = 0$ to pooled SCM $\nu = 1$.
No choice of $\nu$ substantively changes the conclusions, and each rules out large positive effects.
Finally, we consider the result of trimming states with poor pre-treatment fit, following common practice in the matching and SCM literatures.
Figure \ref{fig:trimmed_scm} shows the overall ATT estimates when removing an increasing number of treated units with poor fits, in order of decreasing unit-level fit.
Overall, omitting the worst-fit states decreases the magnitude of the estimated effect, and increases the variability of the estimate. However, all estimates still rule out large positive effects.

  \begin{figure}[tb]
    \centering
    \begin{subfigure}[t]{0.45\textwidth}
      {\centering \includegraphics[width=\maxwidth]{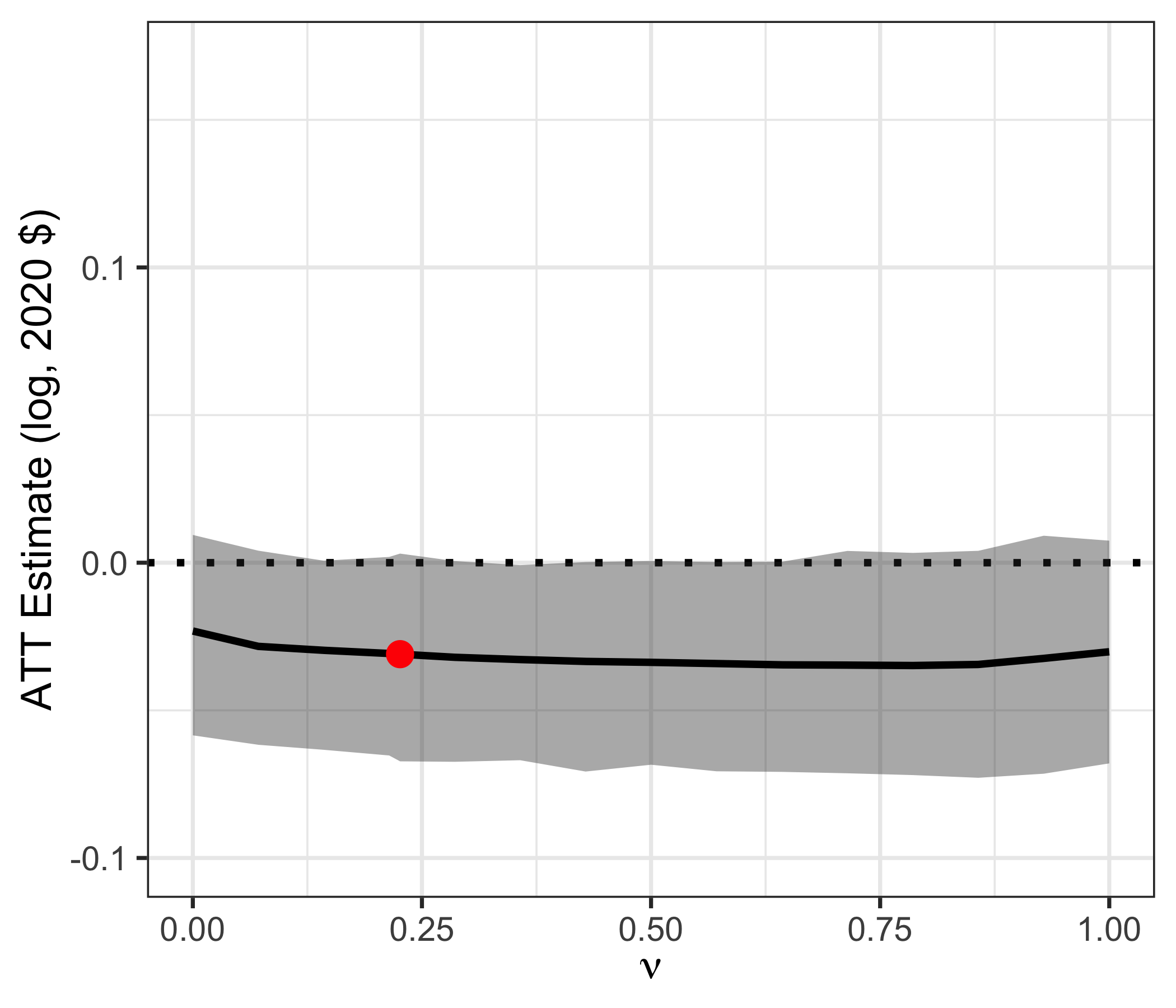} 
      }
      \caption{\label{fig:nu_sensitivity}Varying $\nu$ from 0 to 1.}
    \end{subfigure}\quad
      \begin{subfigure}[t]{0.45\textwidth}  
      {\centering \includegraphics[width=\maxwidth]{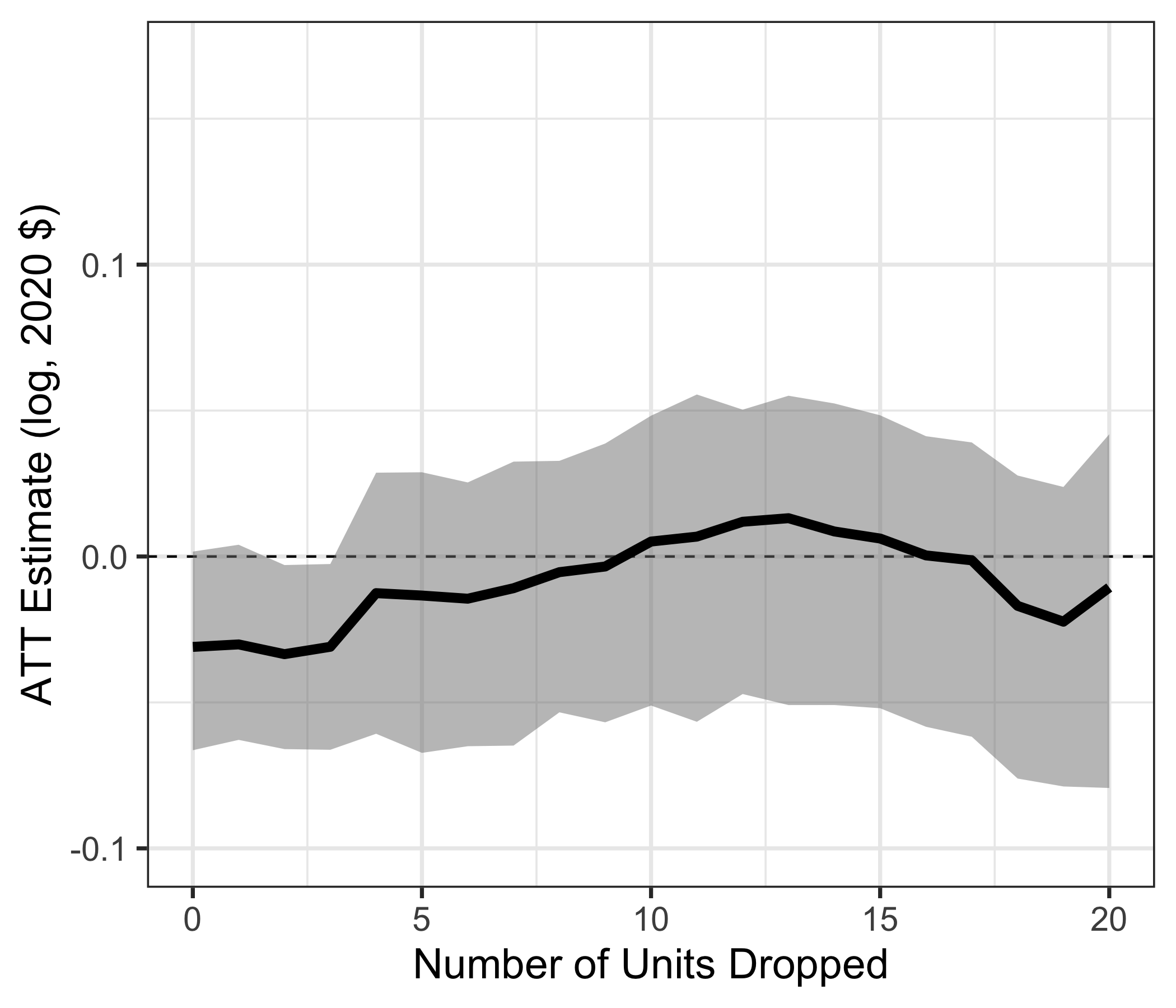} 
      }
      \caption{\label{fig:trimmed_scm}Dropping 1 to 20 treated units according to their worst fit.}
        \end{subfigure}

      \caption{(a) $\widehat{\text{ATT}}$ and approximate 95\% confidence intervals as $\nu$ varies between 0 and 1, $\hat{\nu}$ highlighted. (b) Estimates are not especially sensitivity to dropping an increasing number of units (ranked by pre-treatment imbalance), although the uncertainty intervals are wider with fewer units in the analysis.}
      \label{fig:trimmed_scm_atts}
    \end{figure}

An important feature of SCM-based methods over model-based methods is that we can directly inspect the weights, and that these weights are non-negative and sum to one. Appendix Figures \ref{fig:scm_weights} and \ref{fig:did_scm_weights} show the state-specific weights over donor states for each treated unit for partially pooled SCM without an intercept and with both an intercept and auxiliary covariates, respectively.
Without the intercept, both Illinois and Wyoming are consistently important donor states. Both states had relatively high levels of per-pupil expenditures throughout the study period and several synthetic controls place nearly all of the weight on these two states in order to match the level. However, after removing pre-treatment averages via an intercept, the weights are much more evenly distributed across the donor pool, suggesting that estimates are not overly reliant on a single control unit.

\section{Discussion}
\label{sec:discussion}

In this paper, we develop a new framework for estimating the impact of a treatment adopted gradually by units over time. 
In our motivating example, 33 states have enacted laws mandating school districts to bargain with teachers unions \citep{paglayan2019public}, and we seek to estimate the effects of these laws on educational expenditures.
To do so, we adapt SCM to the staggered adoption setting. We argue that current practice of estimating separate SCM weights for each treated unit is unlikely to yield good results, but also that fully pooled SCM may over-correct; our preferred approach, partially pooled SCM, finds weights that balance both state-specific and overall pre-treatment fit. 
We then extend this basic approach to incorporate an intercept shift as well as auxiliary covariates. 
We apply this approach to the teacher bargaining example and, consistent with recent analyses, find weakly negative estimates on student expenditures.

We briefly note some directions for future work. 
First, we could extend these ideas to other settings with multiple treated units, such as where treatment can ``shut off'' for some units \citep{imai2019twoway}, or where all units are eventually treated \citep{athey2018design}. This would likely require additional assumptions.
We could similarly incorporate other structure from our application. 
For example, in staggered adoption settings where multiple units adopt treatment at the same time, we could add a layer in the hierarchy and more closely pool units treated at the same time while still partially pooling different treatment cohorts. See Appendix \ref{sec:time_cohorts}.

Second, many SCM analyses explore multiple outcomes. As in other SCM studies, we treat each outcome separately, choosing different synthetic control weights for each. In many settings, however, lagged values from one outcome may predict future values of another, suggesting that balancing multiple outcome variables would be useful. This seems especially important in settings like ours with relatively few units.

Finally, we could adapt recent proposals for bias correction and other ``doubly robust'' estimators to this setting, which will be important for both estimation and inference \citep{BenMichael_2018_AugSCM, Abadie_LHour, Arkhangelsky2018}.
Existing approaches have largely been limited to the case with a single treated unit or, if multiple units are treated, to a single adoption time. More complex models are possible and may be desirable in the staggered adoption setting. For example, \citet{fesler2019promise} apply the Ridge Augmented SCM proposal in \citet{BenMichael_2018_AugSCM} to a staggered adoption setting, modeling each treated unit separately. Partial pooling may be helpful here.
In another direction, we might consider an outcome model that incorporates the time weights used in \citet{Arkhangelsky2018}. 
We anticipate that, unlike in the simple case with unit fixed effects, these augmented approaches likely require more elaborate shrinkage estimation, such as via matrix penalties.

\clearpage
\singlespacing
\bibliographystyle{chicago}
\bibliography{citations}

 
\clearpage

\appendix
\renewcommand\thefigure{\thesection.\arabic{figure}}
\renewcommand\thetable{\thesection.\arabic{table}}
\renewcommand\thetheorem{A.\arabic{theorem}}
\renewcommand\thecorollary{A.\arabic{corollary}}
\renewcommand\thelemma{A.\arabic{lemma}}
\renewcommand\theproposition{A.\arabic{proposition}}
\renewcommand\theequation{A.\arabic{equation}}
\renewcommand\theassumption{A.\arabic{assumption}}
\setcounter{figure}{0}
\setcounter{assumption}{0}
\setcounter{theorem}{0}

\section{Additional theoretical results}
\subsection{Further discussion of inference}
\label{sec:asymp_normal}

We now continue the discussion of inference from the main text in Section \ref{sec:inference}.
Our goal here is to discuss the conditions under which the proposed estimator is asymptotically Normal. Since asymptotic theory is not the focus of our paper, we leave for future work a rigorous derivation of the validity of the wild bootstrap procedure, in particular, adapting the proof of the main theorem in \citet{Otsu2017} and showing that the additional conditions in that proof are satisfied with our proposed procedure.

In order to discuss inferential procedures for partially pooled SCM with an intercept shift, we will consider a generalization of parallel trends. For each time period $g$, we assume that the expected differences between post-$g$ and pre-$g$ outcomes do not depend on whether unit $i$ is treated at time $g$, conditional on auxiliary covariates $i$ and the vector of pre-$g$ residuals $\dot{Y}_i^g \equiv \left(Y_{ig-L}, \ldots, Y_{ig-1}\right) - \frac{1}{L}\sum_{\ell=1}^L Y_{ig-\ell}$.

\begin{assumption}[Conditional parallel trends] 
  \label{a:condl_trends}
  With $L < T_1$, for all $k >0$ and $\ell \geq 1$
    \[
      \E[Y_{ig+k}(\infty) - Y_{ig-\ell}(\infty) \mid T_i = g, \dot{Y}_i^g, X_i] = \E[Y_{ig+k}(\infty) - Y_{ig-\ell}(\infty) \mid \dot{Y}_i^g, X_i] \equiv m_{gk\ell}( \dot{Y}_i^g, X_i)
    \]
\end{assumption}
Assumption \ref{a:condl_trends} is a generalization of the conditional parallel trends assumption in \citet{abadie2005semiparametric} to the staggered adoption setting, including the pre-treatment residuals $\dot{Y}_i^g$. 
It loosens the usual parallel trends assumption by allowing trends to differ depending on the auxiliary covariates and the deviation of lagged outcomes from their baseline value. 
Thus, we are essentially conditioning on pre-treatment ``dynamics,'' rather than pre-treatment levels. For instance, even if two states have very different levels of student expenditures, under conditional parallel trends we can compare them so long as they have similar pre-treatment trends and shocks.
See \citet{hazlett2018trajectory} and \citet{Callaway2018} for related conditional parallel trends assumptions. In addition, we will assume that the conditional expectation of the post- and pre-$g$ differences is linear.
\begin{assumption}
    \label{a:linear}
    \[
      m_{gk\ell}(\dot{Y}_i^g, X_i) = \beta_{gk\ell}^Y \cdot \dot{Y}_i^g + \beta_{gk\ell}^X \cdot X_i
    \]
\end{assumption}
We make two further assumptions that allow for asymptotic normality as the number of units grows while the number of lags $L$ stays fixed. First, we assume that the synthetic controls have perfect fit when averaged within time-cohorts; second, we assume that the sum of the squared weights is bounded.
\begin{assumption}[Exact balance within treatment cohorts and bounded weights]
  \label{a:exact_cohorts}
  Assume that
  \[
    \frac{1}{n_g}\sum_{T_i = g} \dot{Y}_i^g = \frac{1}{n_g}\sum_{i=1}^N\sum_{T_j = g}\hat{\gamma}_{ij} \dot{Y}_i^g  \text{ and } \frac{1}{n_g}\sum_{T_i = g} X_i = \frac{1}{n_g}\sum_{i=1}^N\sum_{T_j = g}\hat{\gamma}_{ij} X_i,
  \]
  for all $g = T_1,\ldots,T_J$. Furthermore, $\|\hat{\gamma}_j\|_2 \leq \frac{C}{\sqrt{N_0}}$ for all $j=1,\ldots,J$ and some constant $C$.
\end{assumption}
\noindent Note that by transforming from the penalized optimization problem \eqref{eq:stag_avg_relative_scm_primal_intercept} to the constrained form, there is a choice of $\lambda$ that guarantees that the the constraint on the weights are satisfied, if there exists a feasible solution.
Finally, we make two assumptions on the noise terms $\varepsilon_{igk} \equiv Y_{ig+k}(\infty) - \frac{1}{L}\sum_{\ell=1}^{L}Y_{ig-\ell}(\infty) - \frac{1}{L}\sum_{\ell=1}^{L}m_{k\ell}(g, \dot{Y}_i^g, X_i)$. First, we assume that they are independent across units; second, we assume that they are sufficiently regular so that their average satisfies a central limit theorem.

\begin{assumption}
  \label{a:lyapunov}
  $\varepsilon_{igk}$ are independent across units $i=1,\ldots,N$, and
  for some $\delta > 0$, the $2 + \delta$\super{th} moment exists, $\E\left[\left|\varepsilon_{igk}\right|^{2 + \delta}\right] < \infty$, and furthermore
  \[
    \lim_{N \to \infty}\frac{\sum_{T_i \neq \infty} \E\left[\left|\varepsilon_{iT_ik}\right|^{2 + \delta}\right]}{\left(\sum_{T_i \neq \infty} \E\left[\varepsilon_{iT_ik}^{2}\right]\right)^{1 + \frac{\delta}{2}}} = 0.
  \]
\end{assumption}

Under these assumptions, the estimate of the effect $k$ periods after treatment, $\widehat{\text{ATT}}_k$, will be asymptotically normal as $N$ grows with a fixed number of lags $L$, and where the number of control units $N_0$ grows more quickly than the number of treated units $J$.
\begin{theorem}
    \label{thm:asymp_normal}
   Assume that $\frac{J}{N_0} \to 0$ as both $J, N_0 \to \infty$, with $L$ fixed. Under Assumptions \ref{a:condl_trends}, \ref{a:linear}, \ref{a:exact_cohorts}, and \ref{a:lyapunov}
  \[
    \sqrt{J}\left(\widehat{\text{ATT}}_k - \text{ATT}\right) = \frac{1}{\sqrt{J}} \sum_{T_i \neq \infty} \varepsilon_{iT_j + k} + o_p(1).
  \]
  Furthermore, $
    \frac{\widehat{\text{ATT}}_k - \text{ATT}}{\frac{1}{J}\sum_{T_i \neq \infty}\E\left[\varepsilon_{iT_ik}^{2}\right]} \overset{d}{\to} N(0, 1)$.
\end{theorem}

\paragraph{Jackknife.} Finally, we briefly discuss constructing confidence intervals via the
leave-one-unit-out jackknife approach, which proceeds as follows. Fix hyperparameter values $\nu, \xi$, and $\lambda$; for each unit $i=1,\ldots,N$: drop unit $i$ and re-fit the intercepts and the weights via Equation \eqref{eq:primal_with_covs} to obtain $\hat{\alpha}^{(-i)}$ and $\widehat{\Gamma}^{(-i)}$ and get the synthetic control estimates $\hat{Y}_{j T_j +k}^{(-i)}$. Then compute the leave-one-unit-out estimate  $\widehat{\text{ATT}}_k^{(-i)} = \frac{1}{J^{(-i)}}\sum_{j=1}^J\bbone_{j \neq i} \left\{Y_{jT_j+k} - \hat{Y}_{jT_j+k}^{(-i)}\right\}$, where $J^{(-i)} \equiv J - \bbone_{T_i < \infty}$. The jackknife estimate of the standard error is then:
\begin{equation}
  \label{eq:jackknife_se}
  \hat{V}_k = \frac{n-1}{n}\sum_{i=1}^n\left(\widehat{\text{ATT}}_k^{(-i)} - \frac{1}{n}\sum_{j=1}^n \widehat{\text{ATT}}_k^{(-j)}\right)^2,
\end{equation}
with an approximate 95\% confidence interval  $\widehat{\text{ATT}}_k \pm 2 \hat{V}_k$.
We include Monte Carlo estimates of the coverage under our simulation setup in Figures \ref{fig:coverage_both} and \ref{fig:coverage_scm}.

\subsection{Fully pooling within time cohorts}
\label{sec:time_cohorts}

As we discuss in Section \ref{sec:error_bounds}, if all units are treated at the same time, $T_1 = \cdots = T_J$, our error bounds depend only on the pooled imbalance and do not include the unit-level imbalance. Thus, if units are treated in cohorts (i.e., several units treated at the same time), then the bounds suggest modeling 
variation in pre-treatment outcomes \emph{between} treatment cohorts 
separately from 
the pooled average.
This leads to a natural modification of our partially pooled estimator: We can fully pool within cohorts by applying the estimator to treatment cohorts rather than individual treated units, optimizing a weighted average of the overall imbalance and the average cohort-level imbalance.
Concretely, let $G$ be the number of distinct treatment times, 
which we denote $T(g)$, $g=1,\ldots,G$, and $n_g = \sum_{i=1}^N \bbone\{T_i = T(g)\}$ is the number of units treated in time $T(g)$.
We can modify the optimization problem to find $G$ sets of weights, where the individual objective for treatment cohort $g$ is
\[
  q_g(\gamma_g)^{\text{cohort}} = \sqrt{\frac{1}{L_g}\sum_{\ell=1}^{L_g} \left(\sum_{i=1}^N \bbone\{T_j = T(g)\} Y_{i T(g) -\ell} - \sum_{i=1}^N \gamma_{ig} Y_{iT(g) - \ell}\right)^2}.
\]
As before, we will restrict the set of donor units for cohort $g$ to those not yet treated $K$ periods after $T(g)$, $\calD(g) \equiv \{i : T_i > T(g) + K\}$, and we will restrict the weights to so that $\gamma_g \in \Delta^\scm(g)$ satisfies $\gamma_{ig} \geq 0$ for all $i$, $\sum_i \gamma_{ig} = n_g$, and $\gamma_{ig} = 0$ if $i \not \in D(g)$. We then similarly define the separate and pooled balance measures:
\[
    q^{\text{sep cohort}}(\Gamma) = \sqrt{\frac{1}{G}\sum_{g=1}^G \frac{1}{L_g}\sum_{\ell=1}^{L_g} \left(\sum_{i=1}^N \bbone\{T_j = T(g)\} Y_{i T(g) -\ell} - \sum_{i=1}^N \gamma_{ig} Y_{iT(g) - \ell}\right)^2},
\]
and 
\[
  q^{\text{pool cohort}}(\Gamma) = \sqrt{\frac{1}{\max_g L_g}\sum_{\ell=1}^{\max_g L_g} \left(\frac{1}{G}\sum_{g=1}^G \sum_{i=1}^N \bbone\{T_j = T(g)\} Y_{i T(g) -\ell} - \sum_{i=1}^N \gamma_{ig} Y_{iT(g) - \ell}\right)^2}
\]
We can then use these cohort-level measures of imbalance in the partially pooled SCM optimization problem \eqref{eq:stag_avg_relative_scm_primal}, and similarly can include an intercept as in \eqref{eq:stag_avg_relative_scm_primal_intercept}.
More generally, if we do not want to fully pool within clusters, we can include three (or more) imbalance terms in our objective function to capture unit-level, pooled, and intermediate cluster-level imbalance.

\subsection{Partially pooled SCM: Dual shrinkage}
\label{sec:sim_scm_dual}

We now inspect the Lagrangian dual problem to the partially pooled SCM problem in Equation \eqref{eq:stag_avg_relative_scm_primal}, showing that the optimization problem partially pools a set of unit-specific dual variables toward global dual variables. We focus on balancing the first $L_j = L \leq T_1-1$ lagged outcomes, which are observed for each treated unit.

For each treated unit $j$, the sum-to-one constraint induces a Lagrange multiplier $\alpha_j \in \R$, and
the state-level balance measure induces a set of Lagrange multipliers $\beta_j \in \R^{L}$, with elements $\beta_{\ell j}$. We combine these dual parameters into a vector $\alpha=[\alpha_1,\ldots,\alpha_J]\in\R^J$ and a matrix $\beta = [\beta_1,\ldots,\beta_J] \in \R^{L \times J}$. In addition to the $J$ sets of Lagrange multipliers --- one for each treated unit --- the pooled balance measure in the partially pooled SCM problem Equation \eqref{eq:stag_avg_relative_scm_primal} induces a set of global Lagrange multipliers $\mu_\beta\in\R^{L}$. 
As we see in the following proposition, in the dual problem the parameters $\beta_1,\ldots,\beta_J$ are regularized toward this set of pooled Lagrange multipliers, $\mu_\beta$.

\begin{proposition}
\label{prop:combined_avg_scm_dual}

The Lagrangian dual to Equation \eqref{eq:stag_avg_relative_scm_primal} with un-normalized objevtices $q^\sep$ and $q^\pool$ with $L_j = L < T_1$ and $\lambda > 0$ is:
\begin{equation}
  \label{eq:combined_avg_scm_dual}
  \min_{\alpha, \mu_\beta, \beta} 
  \calL(\alpha, \beta)
  + \frac{\lambda L}{2} \left(\frac{1}{(1-\nu)J}\sum_{j=1}^J\|\beta_j-\mu_\beta\|_2^2 + \frac{1}{\nu}\|\mu_\beta\|_2^2\right),
\end{equation}
where the dual objective function is
\begin{equation}
  \label{eq:dual_objective}
  \calL(\alpha, \beta) \equiv \frac{1}{J}\sum_{j=1}^J\left[\sum_{i \in \calD_j}\left[\alpha_j + \sum_{\ell=1}^{L}\beta_{\ell j}Y_{i T_j-\ell}\right]_+^2 - 
  \left(\alpha_j + \sum_{\ell=1}^{L}\beta_{\ell j}Y_{j T_1-\ell}\right)\right],
\end{equation}
where $[x]_+ = \max\{0, x\}$.
For treated unit $j$, the synthetic control weight on unit $i$ is $\hat{\gamma}_{ij} = \left[\hat{\alpha}_j + \sum_{\ell=1}^{L}\hat{\beta}_{\ell j}Y_{j T_j-\ell}\right]_+$.
\end{proposition}

\noindent Proposition \ref{prop:combined_avg_scm_dual} highlights that the estimator partially pools the individual synthetic controls to the pooled synthetic control \emph{in the dual parameter space}, with $\nu$ controlling the level of pooling. 
When $\nu = 0$ in the separate SCM problem, the parameters $\beta_1,\ldots\beta_J$ are shrunk towards zero rather than a set of global parameters. 
By contrast, when $\nu = 1$, $\beta_1,\ldots,\beta_J$ are constrained to be equal to $\mu_\beta$, fitting a single pooled synthetic control in the dual parameter space. By choosing $\nu \in (0,1)$, we move continuously between the two extremes of $J$ separate Lagrangian dual problems and a single dual problem, regularizing the individual $\beta_j$s toward the pooled $\mu_\beta$,
allowing for some limited differences between the $J$ dual parameters.

\clearpage
\section{Additional figures}
\subsection{Additional simulation results}
\label{sec:additional_sims}
\setcounter{figure}{0}

\begin{figure}
  \centering
  {\centering \includegraphics[width=\textwidth]{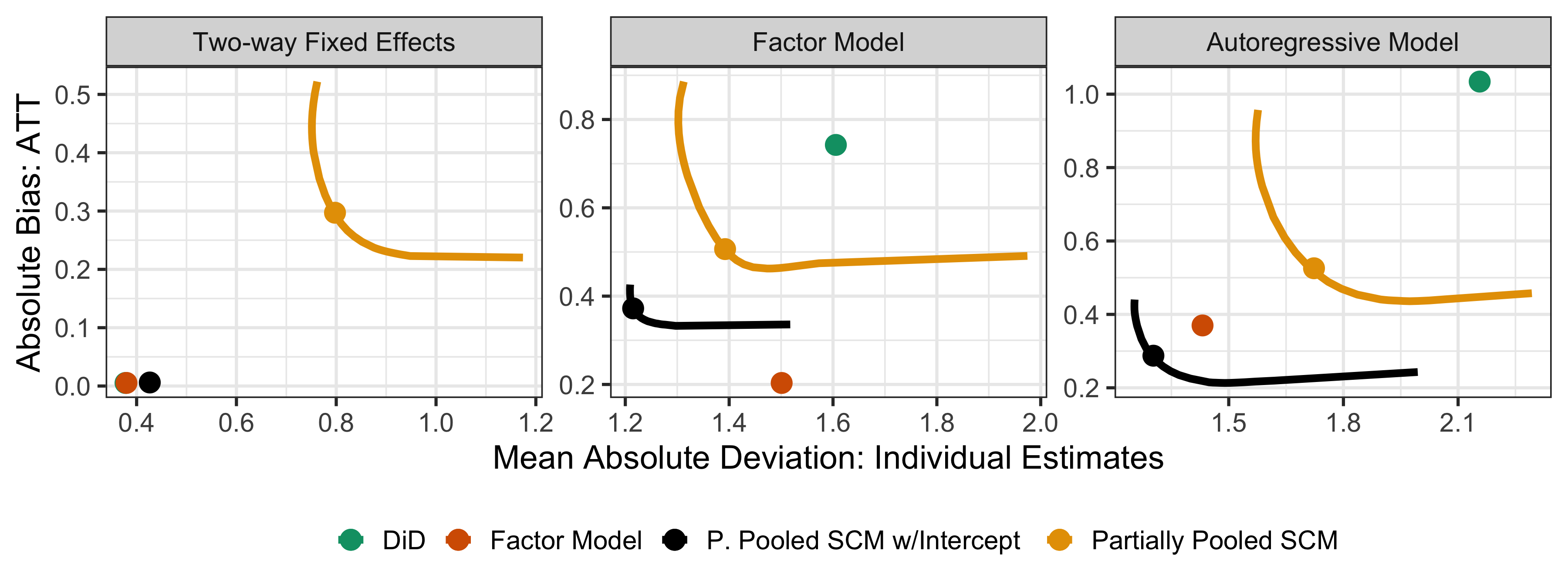}
  }
  \caption{\label{fig:sim_study_bias}Monte Carlo estimates of the bias for the overall ATT vs the MAD for the individual ATT estimates.} 
  \end{figure}

  \begin{figure}
    \centering
    {\centering \includegraphics[width=\textwidth]{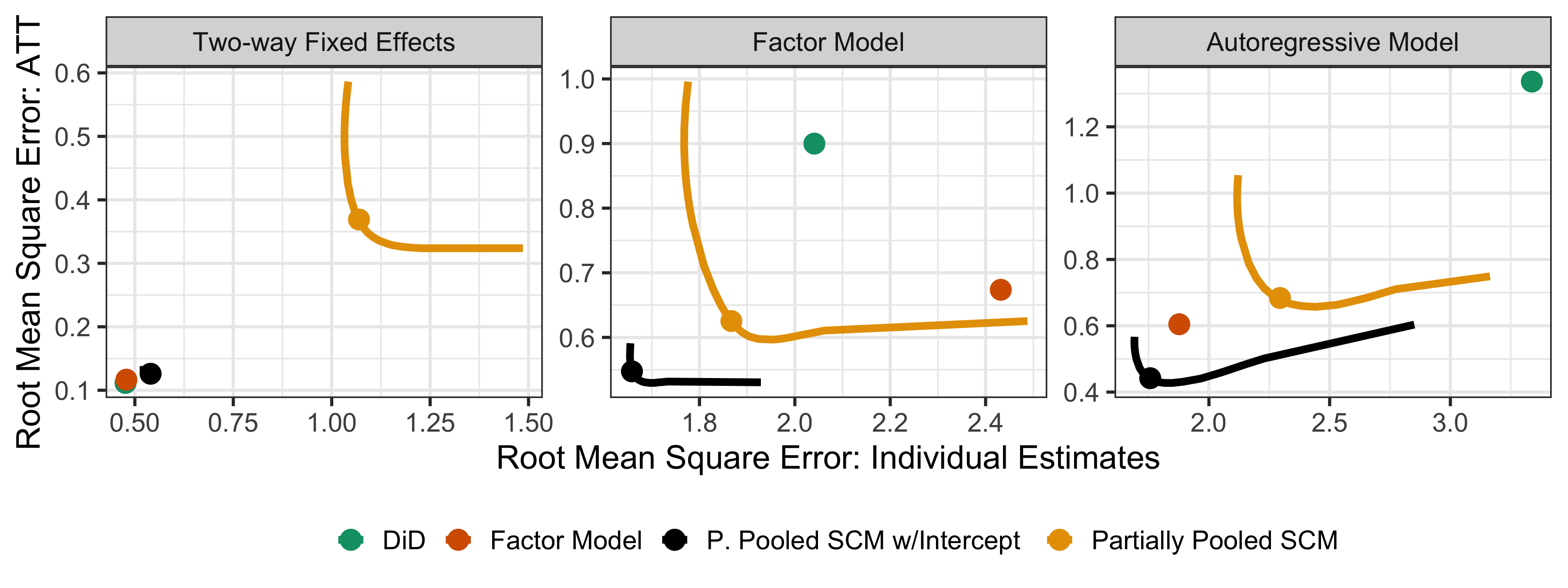}
    }
    \caption{\label{fig:sim_study_rmse}Monte Carlo estimates of the RMSE for the overall ATT vs the RMSE of the individual ATT estimates.} 
    \end{figure}

    \begin{figure}
      \centering
      {\centering \includegraphics[width=\textwidth]{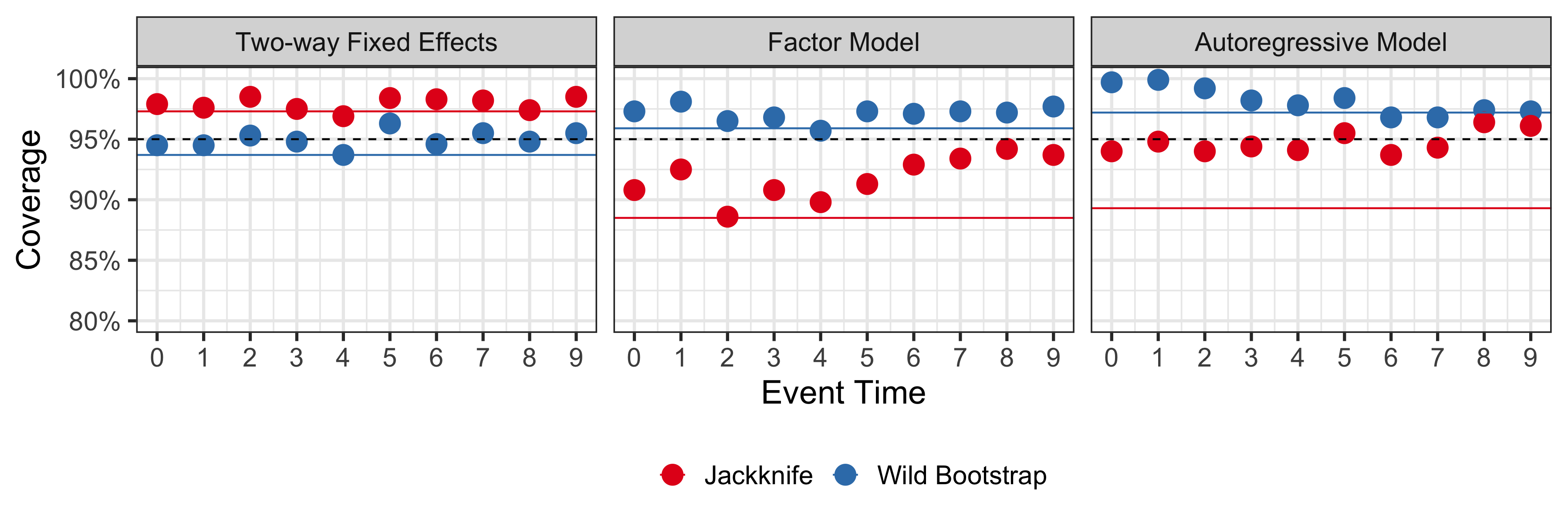}
      }
      \caption{\label{fig:coverage_both} Monte Carlo estimates of the coverage of approximate 95\% confidence intervals $k=0,\ldots,9$ periods after treatment using partially pooled SCM with an intercept. The solid line indicates the coverage for the overall ATT estimate averaged across all post-treatment periods.} 
      \end{figure}   
  
    \begin{figure}
      \centering
      {\centering \includegraphics[width=\textwidth]{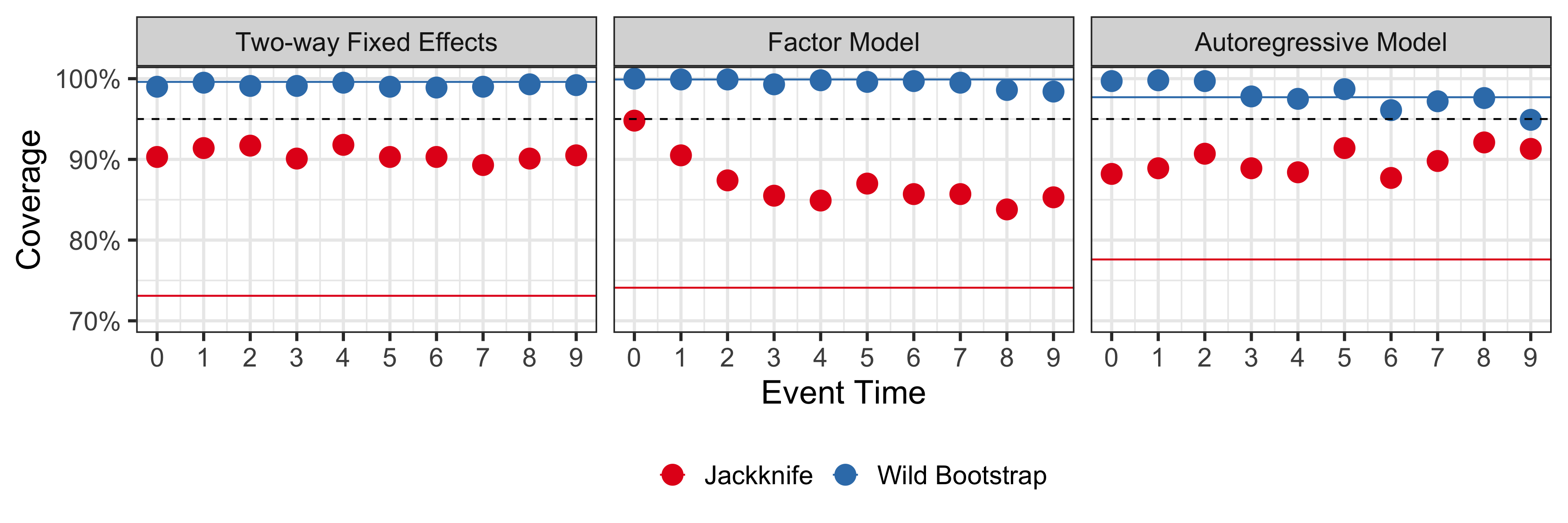}
      }
      \caption{\label{fig:coverage_scm} Monte Carlo estimates of the coverage of approximate 95\% confidence intervals $k=0,\ldots,9$ periods after treatment using partially pooled SCM \emph{without} an intercept. The solid line indicates the coverage for the overall ATT estimate  averaged across all post-treatment periods.} 
      \end{figure}

\clearpage
\subsection{Additional results for the mandatory collective bargaining application}
\label{sec:appendix_results}

\begin{figure}[tb]
  \centering
  {\includegraphics[width=0.6\maxwidth]{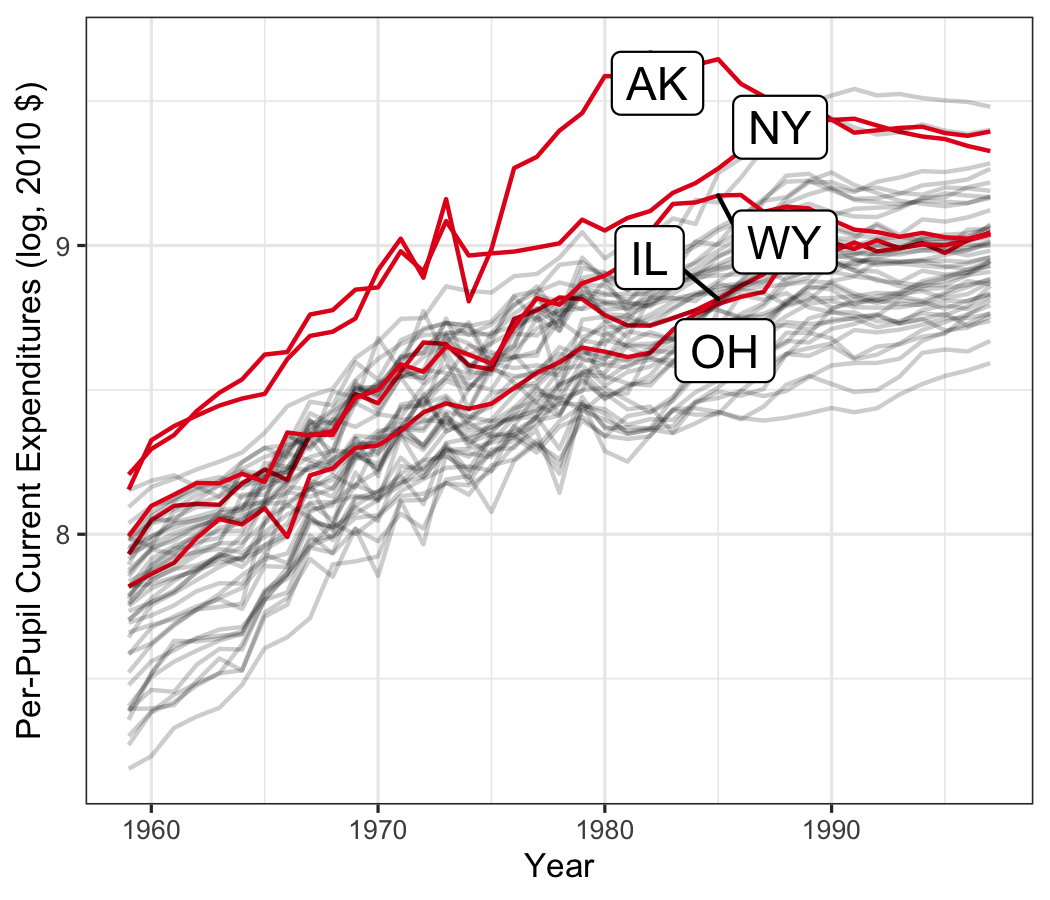}}
  \caption{Per-pupil expenditures for US states over the study period.}
    \label{fig:raw_data_highlight}
\end{figure}

\begin{figure}[tb]
  \centering
  {\includegraphics[width=0.6\maxwidth]{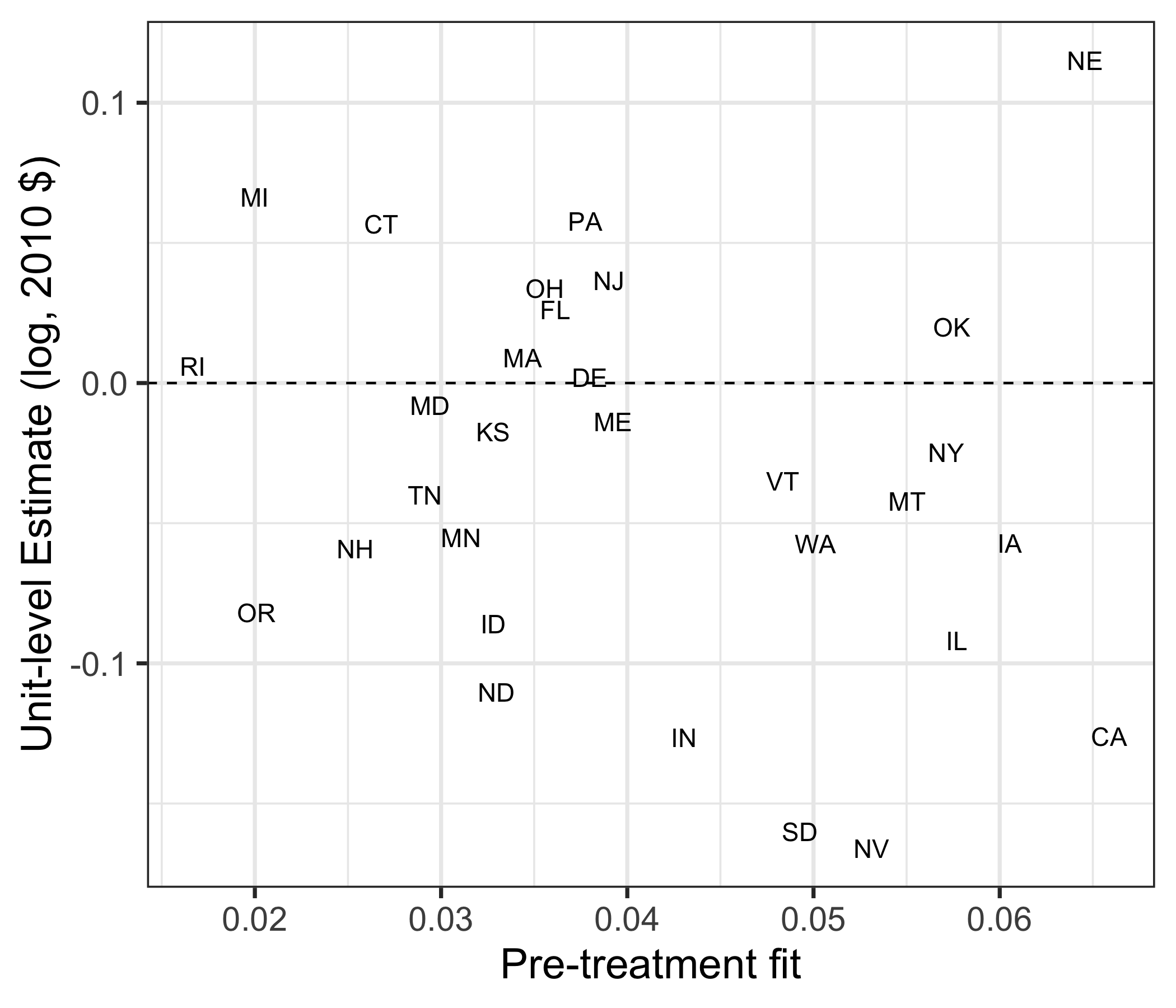}}
  \caption{Average post-treatment effect estimates $\frac{1}{K}\sum_{k=0}^K \hat{\tau}_{jk}$ for the treated states, plotted against the root-mean square pre-treatment fit $q_j(\hat{\gamma}_j)$.}
    \label{fig:cov_int_ppexp_unit_level}
\end{figure}

\begin{figure}[tb]
  \centering
  {\centering \includegraphics[width=0.5\maxwidth]{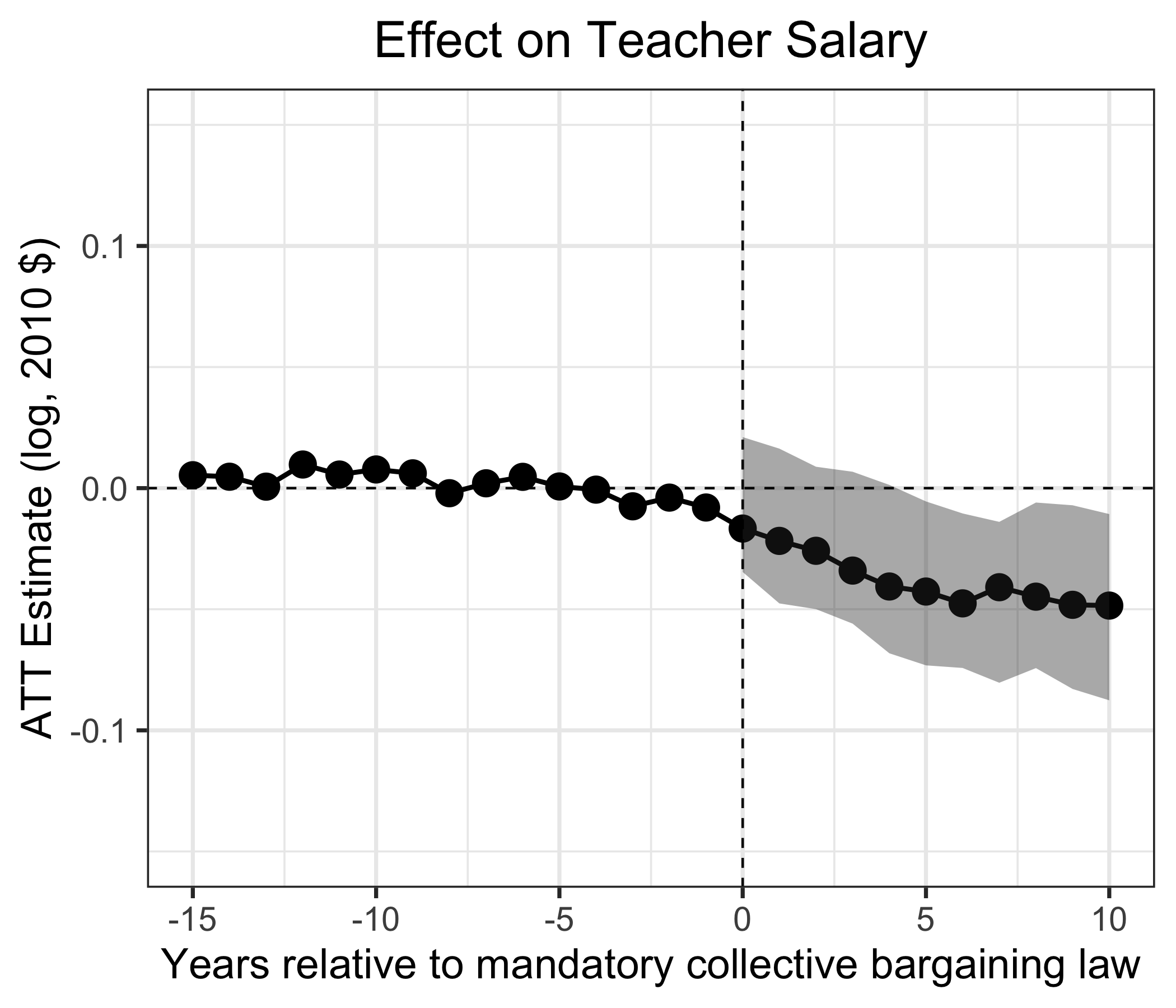} 
  }
  \caption{Partially-pooled SCM with intercept shifts and covariates ($\hat{\nu}=0.26$), estimates of the impact of mandatory collective bargaining laws on average teacher salary (log, 2010 \$).}
    \label{fig:results_teachsal}
  \end{figure}

\begin{figure}[tb]
  \centering
  {\centering \includegraphics[width=\maxwidth]{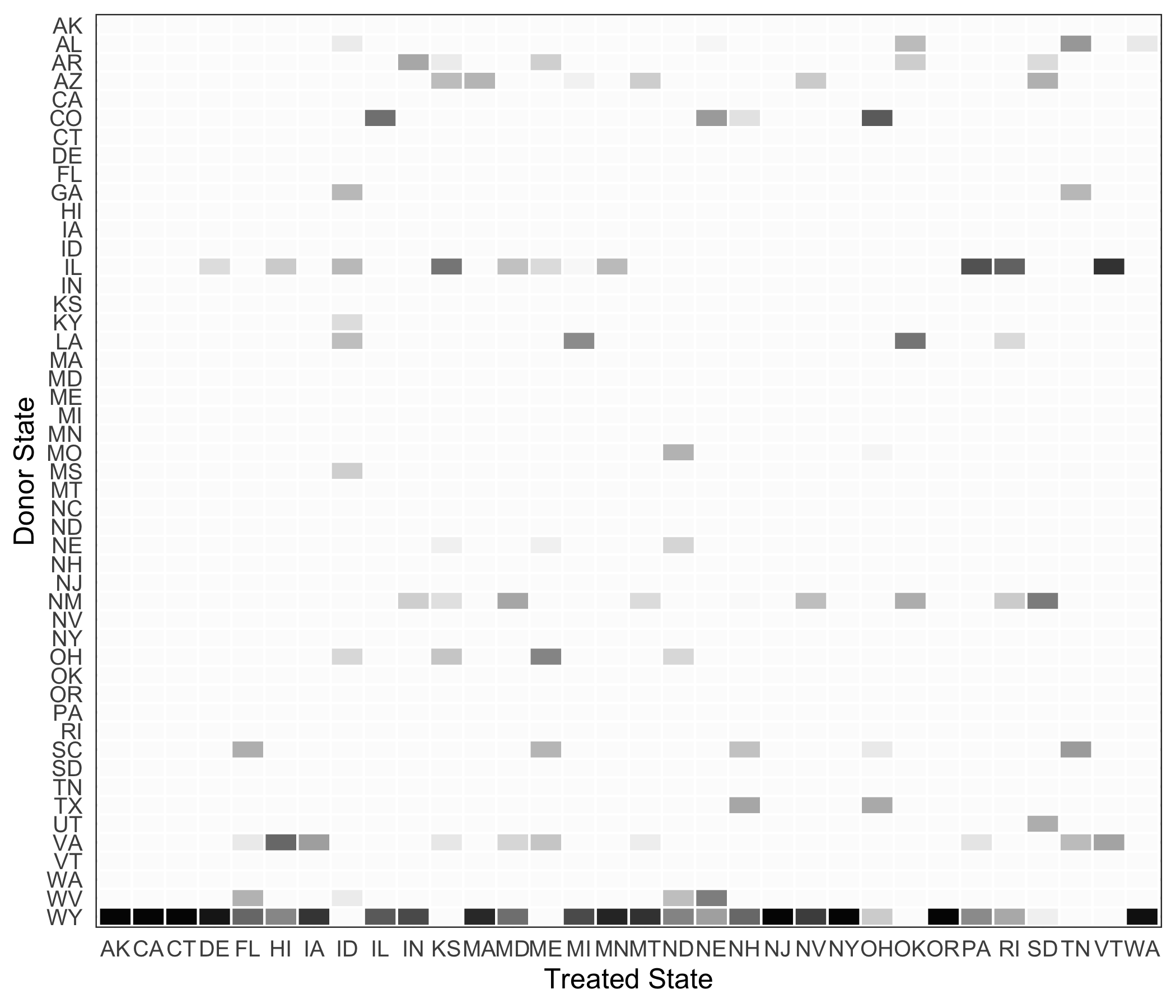} 
  }
  \caption{Partially pooled SCM weights. White cells indicate zero weight, black cells indicate a weight of 1.}
    \label{fig:scm_weights}
\end{figure}

\begin{figure}[tb]
  {\centering \includegraphics[width=\maxwidth]{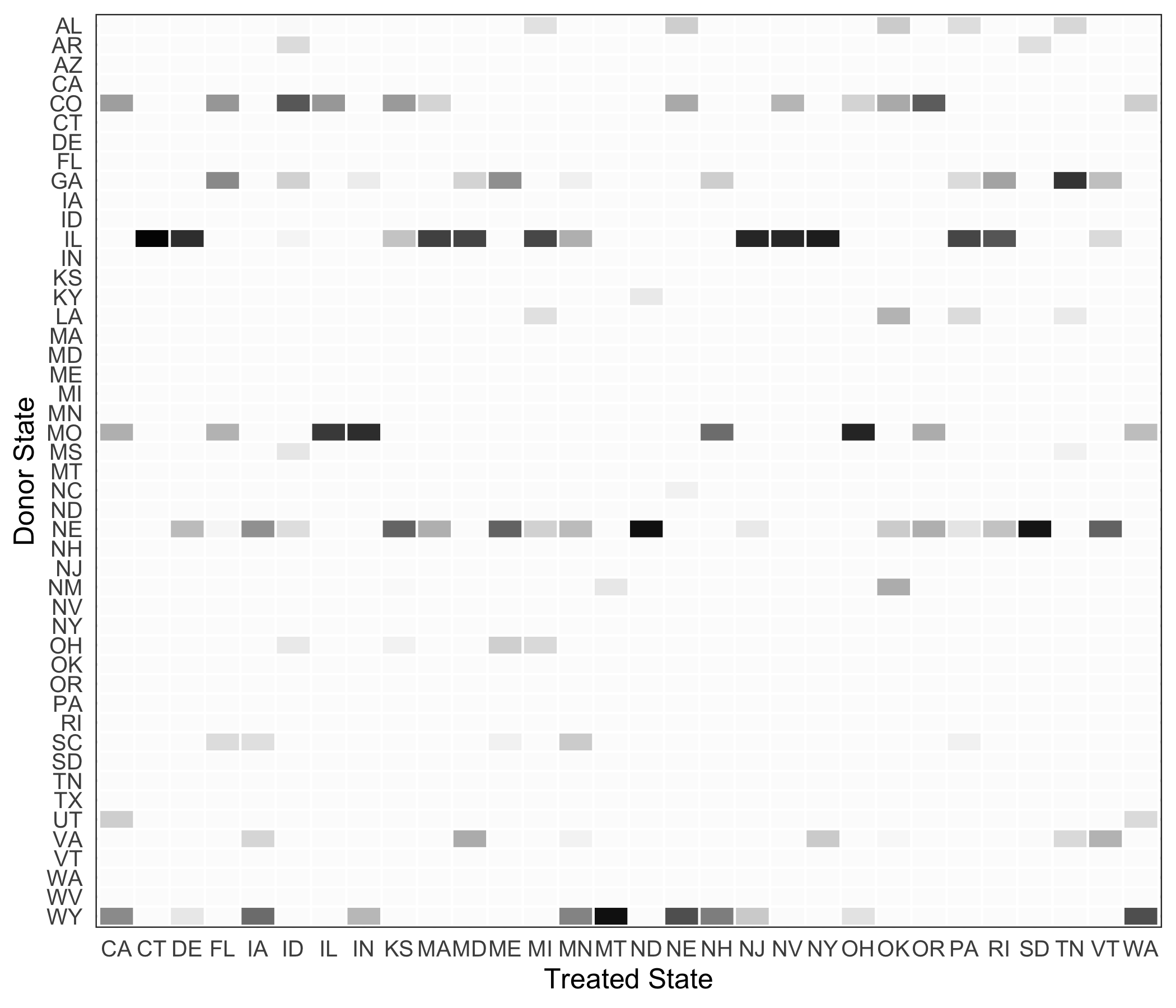} 
  }
  \caption{Partially pooled SCM weights when including an intercept. White cells indicate zero weight, black cells indicate a weight of 1.}
    \label{fig:did_scm_weights}
  \end{figure}


\clearpage
\section{Proofs}

\subsection{Error bounds}
\label{sec:time_ar}

\begin{proof}[Proof of Theorem \ref{thm:time_ar_error}]
  Defining $\xi_t = \rho_t - \bar{\rho}$, the error is
\[
  \hat{\tau}_{j0} - \tau_{j0} = \sum_{\ell=1}^L (\bar{\rho} + \xi_{T_j}) \left(Y_{j T_j - \ell} - \sum_{i \in \calD_j} \hat{\gamma}_{ij} Y_{i T_j-\ell}\right) + \left(\varepsilon_{jT_j} - \sum_{i \in \calD_j} \hat{\gamma}_{ij} \varepsilon_{iT_j}\right)
\]  
So by the triangle and Cauchy-Schwarz inequalities,

\[
  \left|\hat{\tau}_{j0} - \tau_{j0}\right| \leq \|\bar{\rho} + \xi_{T_j}\|_2\sqrt{\sum_{\ell=1}^L \left(Y_{j T_j - \ell} - \sum_{i \in \calD_j} \gamma_{ij} Y_{i T_j-\ell}\right)^2} + \left|\varepsilon_{jT_j} - \sum_{i \in \calD_j} \gamma_{ij} \varepsilon_{iT_j}\right|
\]  

Since $\hat{\gamma}_{j}$ is fit on pre-$T_j$ outcomes, the weights are independent of $\varepsilon_{T_j}$, and so the second term above is sub-Gaussian with scale parameter $\sigma \sqrt{1 + \|\hat{\gamma}_j\|_2^2} \leq \sigma (1 + \|\hat{\gamma}_j\|_2)$. This implies that 
\[
  P\left(\left|\varepsilon_{jT_j} - \sum_{i \in \calD_j} \hat{\gamma}_{ij} \varepsilon_{iT_j}\right| \geq \delta \sigma \left(1 + \|\hat{\gamma}_j\|_2\right)\right) \leq 2\exp\left(-\frac{\delta^2}{2}\right)
\]
For the bound on $\widehat{\text{ATT}}_0$, notice that

\begin{equation}
  \label{eq:ar_att_error_expand}
  \begin{aligned}
      \widehat{\text{ATT}}_0 - \text{ATT}_0 = \frac{1}{J}\sum_{j=1}^J \hat{\tau}_{j0} - \tau_{j0} & = \frac{1}{J}\sum_{j=1}^J \left[\sum_{\ell=1}^L (\bar{\rho}_\ell + \xi_{T_j \ell}) \left(Y_{j T_j - \ell} - \sum_{i \in \calD_j} \hat{\gamma}_{ij} Y_{i T_j-\ell}\right) + \left(\varepsilon_{jT_j} - \sum_{i \in \calD_j} \hat{\gamma}_{ij} \varepsilon_{iT_j}\right)\right]\\
      & = \sum_{\ell=1}^L \bar{\rho}_{\ell} \frac{1}{J}\sum_{j=1}^J \left(Y_{j T_j - \ell} - \sum_{i \in \calD_j} \hat{\gamma}_{ij} Y_{i T_j-\ell}\right)\\
      &+ \frac{1}{J}\sum_{j=1}^J\sum_{\ell=1}^L \xi_{T_j \ell} \left(Y_{j T_j - \ell} - \sum_{i \in \calD_j} \hat{\gamma}_{ij} Y_{i T_j-\ell}\right)\\
      & + \frac{1}{J}\sum_{j=1}^J\left(\varepsilon_{jT_j} - \sum_{i \in \calD_j} \hat{\gamma}_{ij} \varepsilon_{iT_j}\right)
  \end{aligned}
\end{equation}
By Cauchy-Schwarz the absolute value of the first term is
\[
  \left|\sum_{\ell=1}^L \bar{\rho}_{\ell} \frac{1}{J}\sum_{j=1}^J \left(Y_{j T_j - \ell} - \sum_{i \in \calD_j} \hat{\gamma}_{ij} Y_{i T_j-\ell}\right)\right| \leq \|\bar{\rho}\|_2 \sqrt{\sum_{\ell=1}^L\left(\frac{1}{J}\sum_{j=1}^J \left[Y_{j T_j - \ell} - \sum_{i \in \calD_j} \hat{\gamma}_{ij} Y_{i T_j-\ell}\right]\right)^2}.
\]
Similarly, the absolute value of the second term is
\[
  \begin{aligned}
    \left|\frac{1}{J}\sum_{j=1}^J\sum_{\ell=1}^L \xi_{T_j \ell} \left(Y_{j T_j - \ell} - \sum_{i \in \calD_j} \hat{\gamma}_{ij} Y_{i T_j-\ell}\right)\right| & \leq \frac{1}{J}\sum_{j=1}^J \|\xi_{T_j}\|_2 \sqrt{\sum_{\ell=1}^L\left(Y_{j T_j - \ell} - \sum_{i \in \calD_j} \hat{\gamma}_{ij} Y_{i T_j-\ell}\right)^2}\\
    & \leq S \sqrt{\frac{1}{J}\sum_{j=1}^J\sum_{\ell=1}^L\left(Y_{j T_j - \ell} - \sum_{i \in \calD_j} \hat{\gamma}_{ij} Y_{i T_j-\ell}\right)^2}
  \end{aligned}
\]
Finally, notice that $\frac{1}{J} \sum_{j=1}^J \varepsilon_{j T_j}$ is the average of $J$ independent sub-Gaussian random variables and so is itself sub-Gaussian with scale parameter $\frac{\sigma}{\sqrt{J}}$. However, $\frac{1}{J}\sum_{j=1}^J\sum_{i \in \calD_j} \hat{\gamma}_{ij} \varepsilon_{iT_j}$ is the weighted average of sub-Gaussian variables that are independent over $i$ but not necessarily independent over $j$, and so the weighted average is sub-Gaussian with scale parameter $\frac{\sigma}{\sqrt{J}}\|\Gamma\|_F$. The two averages are independent of each other, so
\[
  P\left(\frac{1}{J}\sum_{j=1}^J\left(\varepsilon_{jT_j} - \sum_{i \in \calD_j} \hat{\gamma}_{ij} \varepsilon_{iT_j}\right) \geq \frac{\delta \sigma}{\sqrt{J}}\left(1 + \|\hat{\Gamma}\|_F\right)\right) \leq 2\exp\left(-\frac{\delta^2}{2}\right)
\]  

\noindent Putting together the pieces completes the proof.
\end{proof}

\begin{proof}[Proof of Theorem \ref{thm:lfm_error}]
Following \citet{AbadieAlbertoDiamond2010}, we can re-write $\phi_i$ in terms of the lagged outcomes as 
\begin{equation}
  \label{eq:phi_to_lags}
  \begin{aligned}
    \phi_i & = (\Omega_j'\Omega)^{-1} \sum_{\ell=1}^L\mu_{T_j - \ell}(Y_{iT_j - \ell} - \varepsilon_{iT_j-\ell})\\
    & = \frac{1}{\sqrt{L}}\sum_{\ell=1}^L P^{(j)}_\ell(Y_{iT_j - \ell} - \varepsilon_{i T_j - \ell})
  \end{aligned}
\end{equation}
where $\Omega_j \in \R^{L \times F}$ is the matrix of factors from time $t=T_j-L,\ldots,T_j$, $\frac{1}{\sqrt{L}} P^{(j)}_\ell = (\Omega_j' \Omega)^{-1}\mu_{T_j - \ell} \in \R^F$, and $\frac{1}{\sqrt{L}} P^{(j)} = \frac{1}{\sqrt{L}} [P^{(j)}_1,\ldots,P^{(j)}_J] \in \R^{F \times L}$.
Using Equation \eqref{eq:phi_to_lags}, we can write the error for the ATT as
  \begin{equation}
  \label{eq:lfm_att_error_expand}
  \begin{aligned}
    \widehat{\text{ATT}}_k - \text{ATT}_k = \frac{1}{J}\sum_{j=1}^J \hat{\tau}_{jk} - \tau_{jk} & = \frac{1}{J\sqrt{L}}\sum_{j=1}^J \sum_{\ell=1}^L\mu_{T_j + k}'P^{(j)}_\ell \left(Y_{jT_j - \ell} - \sum_{i \in \calD_j} \hat{\gamma}_{ij} Y_{iT_j - \ell}\right)\\
    & - \frac{1}{J\sqrt{L}}\sum_{j=1}^J\sum_{\ell=1}^L\mu_{T_j + k}'P^{(j)}_\ell\left(\varepsilon_{jT_j - \ell} - \sum_{i \in \calD_j} \hat{\gamma}_{ij} \varepsilon_{iT_j - \ell}\right)\\
    &+ \frac{1}{J}\sum_{j=1}^J \left(\varepsilon_{jt} - \sum_{i \in \calD_j} \hat{\gamma}_{ij} \varepsilon_{iT_j}\right).
  \end{aligned}
\end{equation}

From the proof of Theorem \ref{thm:time_ar_error}, we can bound the final term in Equation \eqref{eq:lfm_att_error_expand}. We now bound the first two terms. First, as in the proof of Theorem \ref{thm:time_ar_error}, we decompose the first term into a time constant, and a time varying component:
\[
\begin{aligned}
\underbrace{\frac{1}{J\sqrt{L}}\sum_{j=1}^J \sum_{\ell=1}^L\mu_{T_j + k}'P^{(j)}_\ell \left(Y_{jT_j - \ell} - \sum_{i \in \calD_j} \hat{\gamma}_{ij} Y_{iT_j - \ell}\right)}_{(\ast)}& = \frac{1}{J\sqrt{L}}\sum_{\ell=1}^L \bar{\mu}_{k \ell} \sum_{j=1}^J \left(Y_{jT_j - \ell} - \sum_{i \in \calD_j} \hat{\gamma}_{ij} Y_{iT_j - \ell}\right)\\
& \;\;\;\;\; + \frac{1}{J\sqrt{L}}\sum_{j=1}^J \sum_{\ell=1}^L\xi_{(T_j + k)\ell} \left(Y_{jT_j - \ell} - \sum_{i \in \calD_j} \hat{\gamma}_{ij} Y_{iT_j - \ell}\right),
\end{aligned}
\]
where $\bar{\mu}_{k\ell} \equiv \frac{1}{J}\sum_{j=1}^J P^{(j) \prime}_\ell\mu_{T_j + k}$, and $\xi_{(T_j + k)\ell} \equiv P^{(j) \prime}_\ell \mu_{T_j + k} - \bar{\mu}_{k\ell}$.
Now by Cauchy-Schwarz,
we get that

\[
\begin{aligned}
|(\ast)| & \leq \|\bar{\mu}_k\|_2\sqrt{\frac{1}{L}\sum_{\ell=1}^L\left( \frac{1}{J}\sum_{j=1}^J Y_{jT_j - \ell} - \sum_{i \in \calD_j} \hat{\gamma}_{ij} Y_{iT_j - \ell}\right)^2} + \frac{1}{J}\sum_{j=1}^J \|\xi_{T_j + k}\|_2 \sqrt{\frac{1}{L}\sum_{\ell=1}^L  \left(Y_{jT_j - \ell} - \sum_{i \in \calD_j} \hat{\gamma}_{ij} Y_{iT_j - \ell}\right)^2}\\
& \leq  \|\bar{\mu}_k\|_2\sqrt{\frac{1}{L}\sum_{\ell=1}^L\left( \frac{1}{J}\sum_{j=1}^J Y_{jT_j - \ell} - \sum_{i \in \calD_j} \hat{\gamma}_{ij} Y_{iT_j - \ell}\right)^2} + S_k\sqrt{\frac{1}{JL}\sum_{j=1}^J\sum_{\ell=1}^L  \left(Y_{jT_j-\ell} - \sum_{i \in \calD_j} \hat{\gamma}_{ij} Y_{iT_j-\ell}\right)^2}
\end{aligned}
\]

We now turn to the second term in Equation \eqref{eq:lfm_att_error_expand}. Since $\varepsilon_{it}$ are independent sub-Gaussian random variables and $\frac{1}{\sqrt{L}}\|\mu_{T_j + k}'P^{(j)}\|_2 \leq \frac{M^2 F}{\sqrt{L}}$,

\[
  P\left(\frac{1}{\sqrt{L}}\left|\frac{1}{J}\sum_{j=1}^J\sum_{\ell=1}^L\mu_{T_j + k}'P^{(j)}_\ell \varepsilon_{jT_j-\ell}\right| \geq  \frac{\delta \sigma M^2 F}{\sqrt{JL}}\right) \leq 2\exp\left(-\frac{\delta^2}{2}\right)
\]

Next, since $\hat{\gamma}_1,\ldots,\hat{\gamma}_J \in \Delta^\scm$, $\frac{1}{J}\sum_{j=1}^J \|\hat{\gamma}_j\|_1 = 1$, by H\"{o}lder's inequality

\[
  \left|\frac{1}{J\sqrt{L}}\sum_{j=1}^J\sum_{\ell=1}^L \mu_{T_j + k}'P^{(j)}_\ell \sum_{i \in \calD_j} \hat{\gamma}_{ij} \varepsilon_{iT_j - \ell}\right| \leq \max_{j \in \{1,\ldots,J\}, i \in \calD_j}\left|\frac{1}{\sqrt{L}}\sum_{\ell=1}^L \mu_{T_j + k}'P^{(j)}_\ell \varepsilon_{iT_j - \ell}\right| \leq 2\frac{\sigma M^2 F}{\sqrt{L}} \left(\sqrt{\log NJ} + \delta\right)
\]
where the final inequality holds with probability at least $1 - 2\exp\left(-\frac{\delta^2}{2}\right)$ by the standard tail bound on the maximum of sub-Gaussian random variables. Putting together the pieces with a union bound completes the proof.

\end{proof}

\subsection{Asymptotic normality}

\begin{proof}[Proof of Theorem \ref{thm:asymp_normal}]
  Define $\bar{\beta}_{gk}^Y = \frac{1}{L}\sum_{\ell=1}^L\beta_{gk\ell}^Y$ and $\bar{\beta}_{gk}^X = \frac{1}{L}\sum_{\ell=1}^L\beta_{gk\ell}^X$. Note that under linearity in Assumption \ref{a:linear},
  \[
    Y_{ig+k}(\infty) - \frac{1}{L}\sum_{\ell=1}^LY_{ig-\ell}(\infty) = \bar{\beta}_{gk}^Y \cdot \dot{Y}_i^g + \bar{\beta}_{gk}^X \cdot X_i + \varepsilon_{igk}.
  \]
  So the estimation error for the treatment effect for unit $j$ at time $k$ is 
  \[
    \begin{aligned}
      \hat{\tau}_{jk} - \tau_{jk} & = Y_{jT_j+k}(\infty) - \frac{1}{L}\sum_{\ell = 1}^L Y_{iT_j-\ell}(\infty) - \sum_i \hat{\gamma}_{ij} \left(Y_{iT_j+k} - \frac{1}{L}\sum_{\ell=1}^L Y_{iT_j-\ell}\right)\\
      & = \bar{\beta}_{T_jk}^Y \cdot \left(\dot{Y}_{j}^{T_j} - \sum_{i}\hat{\gamma}_{ij}\dot{Y}_i^{T_j}\right) + \bar{\beta}_{T_jk}^X \cdot \left(X_j - \sum_{i}\hat{\gamma}_{ij}X_i\right) + \varepsilon_{jT_jk} - \sum_{i}\hat{\gamma}_{ij} \varepsilon_{iT_jk}
    \end{aligned}
  \]
  Aggregating across treated units we see that
  \[
    \begin{aligned}
      \widehat{\text{ATT}}_k - \text{ATT} & = \frac{1}{J}\sum_{j=1}^J \hat{\tau}_{jk} - \tau_{jk}\\
      & = \frac{1}{J}\sum_{g=1}^{T_J}n_g \bar{\beta}_{gk}^Y \cdot\left( \frac{1}{n_g}\sum_{T_i = g} \dot{Y}_i^g - \frac{1}{n_g}\sum_{i=1}^N\sum_{T_j = g}\hat{\gamma}_{ij} \dot{Y}_i^g\right) + n_g \bar{\beta}_{gk}^X \cdot \left(\frac{1}{n_g}\sum_{T_i = g} X_i - \frac{1}{n_g}\sum_{i=1}^N\sum_{T_j = g}\hat{\gamma}_{ij} X_i\right)\\
      & \;\;\;\; + \frac{1}{J}\sum_{j=1}^J \varepsilon_{jT_jk} - \sum_{i}\hat{\gamma}_{ij}\varepsilon_{iT_jk},
    \end{aligned}
  \]
  where $n_g$ is the number of units treated at time $g$.
  Now from Assumption \ref{a:exact_cohorts}, we have exact balance within each cohort, so this reduces to $\widehat{\text{ATT}}_k - \text{ATT} = \frac{1}{J}\sum_{j=1}^j \varepsilon_{jT_jk} - \sum_{i}\hat{\gamma}_{ij}\varepsilon_{iT_jk}$. We now show that the second term is $o_p(J^{-1/2})$.
  Denote $\sigma^2_\text{max} = \max_{igk} \Var(\varepsilon_{igk})$.
  Since the noise terms $\varepsilon_{i\ell k}$ are independent across units $i$,
  \[
    \begin{aligned}
      \Var\left(\frac{1}{J}\sum_{j=1}^{J}\sum_{i} \varepsilon_{iT_jk} \hat{\gamma}_{ij} \right) & = \E\left[\Var\left(\frac{1}{J}\sum_{j=1}^{J}\sum_{i} \varepsilon_{iT_jk}\hat{\gamma}_{ij} \mid \Gamma \right)\right] + \Var\left(\E\left[\frac{1}{J}\sum_{j=1}^{J}\sum_{i} \varepsilon_{igk}\hat{\gamma}_{ij} \mid \Gamma \right]\right)\\
      & = \E\left[\frac{1}{J^2}\sum_i \Var\left(\sum_{j=1}^J \varepsilon_{iT_jk} \hat{\gamma}_{ij} \mid \Gamma \right)\right]\\
      & \leq \E\left[\frac{1}{J^2}\sum_i \sigma_{\text{max}}^2 \sum_{j,j'}\hat{\gamma}_{ij} \hat{\gamma}_{ij'}\right]\\
      & \leq \E\left[\frac{1}{J^2}\sum_i \sigma_{\text{max}}^2 \sum_{j,j'}\|\hat{\gamma}_j\|_2 \|\hat{\gamma}_{j'}\|_2 \right]\\
      & \leq \frac{C^2 \sigma_{\text{max}}^2 }{N_0}
    \end{aligned}
    \]
  By Chebyshev's inequality, $P\left(\left|\frac{1}{\sqrt{J}}\sum_{j=1}^{J}\sum_{i} \varepsilon_{iT_jk} \hat{\gamma}_{ij}\right| \geq \delta \right) \leq \frac{\sigma^2_{\text{max}}C^2 J}{\delta^2 N_0}$. Now since $\frac{J}{N_0} \to 0$, this implies that $\sqrt{J}\left(\widehat{\text{ATT}}_k - \text{ATT}_k\right) = \frac{1}{\sqrt{J}}\sum_{T_i \neq \infty} \varepsilon_{iT_ik} + o_p(1)$. Applying the Lyapunov central limit theorem to the first term and Slutsky's theorem shows asymptotic normality.
\end{proof}

\subsection{Partial pooling of dual parameters}

\begin{lemma}
\label{lem:combined_scm_dual}
The Lagrangian dual to Equation \eqref{eq:stag_avg_relative_scm_primal} with $\nu = 0$, $\lambda > 0$, and $L_j = L < T_1$ is
\begin{equation}
  \label{eq:combined_scm_dual}
  \min_{\alpha, \beta} \underbrace{\frac{1}{J}\sum_{j=1}^J\left[\sum_{i \in \calD_j}\left[\alpha_j + \sum_{\ell=1}^{L}\beta_{\ell j}Y_{iT_j-\ell}\right]_+^2 - 
  \left(\alpha_j + \sum_{\ell=1}^{L}\beta_{\ell j}Y_{j T_1-\ell}\right)\right]}_{\calL(\alpha, \beta)} + \sum_{j=1}^J\frac{\lambda L}{2} \|\beta_j\|_2^2,
\end{equation}
\noindent The resulting donor weights are $\hat{\gamma}_{ij} = \left[\hat{\alpha}_j - \sum_{\ell=1}^L\hat{\beta}_{\ell j} Y_{i T_j-\ell}\right]_+$.
\end{lemma}

\begin{proof}[Proof of Lemma \ref{lem:combined_scm_dual}]
Notice that the separate synth problem 
separates into $J$ optimization problems:

\begin{equation}
  \label{eq:stag_scm_primal_sep}
  \begin{aligned}
    &    \;\; \min_{\gamma_1,\ldots,\gamma_J \in \Delta^{\text{scm}}_j} \;\;\; \frac{1}{2}q^{\rm{sep}}(\Gamma) \;+\;\; 
    \frac{\lambda}{2} \sum_{j = 1}^J \sum_{i=1}^N \gamma_{ij}^2\\
    = & \frac{1}{J} \sum_{j = 1}^{J} \min_{\gamma_j \in \Delta_j^{\text{scm}}}\left\{\left[ \frac{1}{2 L}\sum_{\ell = 1}^{L} \left(Y_{j T_j-\ell} \;-\; \sum_{i=1}^N \gamma_{ij} Y_{iT_j-\ell}\right)^2\,\right]\;+\;\; 
    \frac{\lambda}{2}  \sum_{i=1}^N \gamma_{ij}^2 \right\}
  \end{aligned}
\end{equation}
Thus the Lagrangian dual objective is the sum of the Langrangian dual objectives of the individual objectives in Equation \eqref{eq:stag_scm_primal_sep}. Inserting the dual objectives derived by \citet{BenMichael_2018_AugSCM} yields the result.
\end{proof}

\begin{proof}[Proof of Proposition \ref{prop:combined_avg_scm_dual}]
We start be defining auxiliary variables, $\calE_0,\calE_1, \ldots,\calE_J \in \R^{L}$ where $\calE_{j\ell} = Y_{jT_j-\ell} - \sum_{i=1}^N \gamma_{ij}Y_{iT_j-\ell}$ for $j \geq 1$ and $\calE_{0\ell} = \sum_{T_j > \ell} \left(Y_{jT_j-\ell} - \sum_{i=1}^N \gamma_{ij} Y_{iT_j - \ell}\right)$. Additionally we rescale by $\frac{1}{\lambda}$. Then we can write the partially pooled SCM problem \eqref{eq:stag_avg_relative_scm_primal} as 

\begin{equation}
    \label{eq:hybrid_scm_proof}
\begin{aligned}
     \min_{\gamma_1,\ldots,\gamma_J, \calE_0,\ldots,\calE_J} \;\; & 
    \frac{\nu}{2J^2L\lambda} \sum_{\ell=1}^L\calE_{0\ell}^2  + \frac{1-\nu}{2J\lambda}\sum_{j=1}^J \frac{1}{L} \calE_{j\ell}^2
    + \sum_{j=1}^J \sum_{i=1}^N \frac{1}{2}\gamma_{ij}^2  \\
    \text{subject to  } \;\;\;\;& \calE_{j\ell} = Y_{jT_j-\ell} - \sum_{i=1}^N \gamma_{ij}Y_{iT_j-\ell}\\
    & \calE_{0\ell} = \sum_{T_j > \ell} \left(Y_{jT_j-\ell} - \sum_{i=1}^N \gamma_{ij} Y_{iT_j - \ell}\right)\\
    & \gamma_j \in \Delta_j^{\text{scm}}
\end{aligned} 
\end{equation}
With Lagrange multipliers $\mu_\beta,\zeta_1,\ldots,\zeta_J \in \R^L$ and $\alpha_1,\ldots,\alpha_J \in \R$, the Lagrangian to Equation \eqref{eq:hybrid_scm_proof} is

\[
\begin{aligned}
&\calL(\Gamma, \calE_0, \ldots,\calE_J,\alpha_1,\ldots,\alpha_J, \mu_\beta,\zeta_1,\ldots,\zeta_J) = \\
& \qquad\qquad   \sum_{\ell=1}^L\left[ \frac{\nu}{2LJ^2\lambda}\calE_{0\ell}^2 - \mu_{\beta \ell} \left(\sum_{j=1}^J Y_{jT_j-\ell} - \sum_{i\in \calD_j}\gamma_{ij}Y_{iT_j-\ell}\right) - \calE_{0\ell}\mu_{\beta\ell}\right]\\
& \qquad\qquad  + \quad \sum_{j=1}^J\sum_{\ell=1}^L \left[\frac{1-\nu}{2JL\lambda} \calE_{j\ell}^2 - \zeta_{\ell j}\left(Y_{jT_j-\ell} - \sum_{i \in \calD_j}\gamma_{ij}Y_{iT_j-\ell}\right) - \zeta_{\ell j}\calE_{j\ell}\right]\\
& \qquad\qquad  + \quad \sum_{j=1}^J\sum_{i\in\calD_j}\frac{1}{2}\gamma_{ij}^2 - \alpha_j \gamma_{ij} - \alpha_j
\end{aligned}
\]
Defining $\beta_j = \mu_\beta + \zeta_j$, the dual problem is:

\[
\begin{aligned}
-\min_{\Gamma, \calE_0, \calE_1,\ldots,\calE_J} L(\cdot) & = -\frac{1}{J}\sum_{j=1}^J \sum_{i \in \calD_j} \min_{\gamma_{ij}}\left\{\frac{1}{2}\gamma_{ij}^2 - \left(\alpha_j - \sum_{\ell=1}^L\beta_{\ell j} Y_{i T_j-\ell}\right)\gamma_{ij}\right\} + \sum_{j=1}^J \alpha_j + \sum_{\ell=1}^L \beta_{\ell j}Y_{j T_j-\ell}\\
-\qquad & \sum_{\ell=1}^{L}\min_{\calE_{j\ell}} \left\{\frac{1-\nu}{2L\lambda}\calE_{j\ell}^2 - \calE_{j\ell}(\beta_{\ell j } - \mu_{\beta \ell})\right\}\\
- \qquad & \sum_{\ell=1}^L\min_{\calE_{0\ell}} \left\{\frac{\nu}{2LJ\lambda}\calE_{0\ell}^2 - \calE_{0\ell}\mu_{\beta \ell}\right\}\\
\end{aligned}
\]

From Lemma \ref{lem:combined_scm_dual}, we see that the first term is $\calL(\alpha, \beta)$ and we have the same form for the implied weights. The next two terms are the convex conjugates of a scaled $L^2$ norm. Using the computation that the convex conjugate of $\frac{a}{2}\|x\|_2^2$ is $\frac{1}{2a} \|x\|_2^2$. Finally, the primal problem \eqref{eq:stag_avg_relative_scm_primal} is still convex and a primal feasible point exists, so by Slater's condition strong duality holds.
\end{proof}

\end{document}